\documentclass{article}

\usepackage[style=alphabetic, backend=biber, backref=true]{biblatex}
\usepackage{graphicx} 
\usepackage{tikz}
\usetikzlibrary{arrows.meta, positioning, calc}
\usepackage{quantikz}
\usepackage{booktabs}
\usepackage{xcolor}
\usepackage{makecell}

\usepackage{amsmath}
\usepackage{amsthm}
\usepackage{amssymb}
\usepackage{braket}

\usepackage[margin=1in]{geometry}

\definecolor{darkblue}{rgb}{0.0, 0.0, 0.55}
\usepackage{xurl}
\usepackage{hyperref}
\hypersetup{colorlinks=true, allcolors=darkblue}

\usepackage{aliascnt} 

\theoremstyle{plain}
\newtheorem{thm}{Theorem} 

\newaliascnt{lemma}{thm}       
\newtheorem{lemma}[lemma]{Lemma} 
\aliascntresetthe{lemma}

\newaliascnt{cor}{thm}
\newtheorem{cor}[cor]{Corollary}
\aliascntresetthe{cor}

\theoremstyle{definition}
\newaliascnt{defn}{thm}       
\newtheorem{defn}[defn]{Definition} 
\aliascntresetthe{defn}

\theoremstyle{remark} 

\newtheorem{rem}{Remark}

\usepackage{cleveref}

\crefname{defn}{Definition}{Definitions}  
\Crefname{defn}{Definition}{Definitions}

\crefname{thm}{Theorem}{Theorems}
\Crefname{thm}{Theorem}{Theorems}

\crefname{lemma}{Lemma}{Lemmas}
\Crefname{lemma}{Lemma}{Lemmas}

\crefname{cor}{Corollary}{Corollaries}
\Crefname{cor}{Corollary}{Corollaries}

\crefname{rem}{Remark}{Remarks}
\Crefname{rem}{Remark}{Remarks}

\begin{document}
\title{The Complexity of Single-Interaction Hamiltonians}
\author{
Elad Hecker\\
The Rachel and Selim Benin School of Engineering and Computer Science\\
The Hebrew University of Jerusalem\\
\texttt{elad.hecker@mail.huji.ac.il}
}
\date{}
\maketitle

\begin{abstract}
We study the complexity of Local Hamiltonian problems under the restriction that every local term is an unweighted occurrence of the same fixed interaction. 
In the resulting Single-Interaction Hamiltonian (SIH) model, once the interaction is fixed, the Hamiltonian is specified entirely by its ordered interaction hypergraph.

Our main technical tool is an exact compiler that converts any fixed finite alphabet of local interactions into a single interaction. 
On a designated auxiliary subspace, the compiled Hamiltonian reproduces the source Hamiltonian exactly, while the complete spectrum below a chosen separation scale is preserved with multiplicity.
Combined with a finite-alphabet normalization, this gives a universal polynomial-time reduction from $k$-Local Hamiltonian to $\mathrm{SIH}_{3k+3}$ up to a known global energy scale and inverse-polynomial approximation.

Sharper constructions give QMA-completeness of $\mathrm{SIH}_3$ and of geometrically local $\mathrm{SIH}_8$ with bounded degree and interaction range. 
In the stoquastic setting, we obtain StoqMA-completeness of $\mathrm{Stoq\text{-}SIH}_3$ and QMA-completeness of $1$-Pinned $\mathrm{Stoq\text{-}SIH}_4$. 
We also establish hardness results for uniformly system-size-dependent interactions and derive single-interaction formulations of the Hamiltonian quantum PCP conjecture, frustration-free, exponentially precise, and guided Local Hamiltonian problems.
These results show that independently chosen local matrices and independently tunable coupling strengths are not necessary for a broad range of Hamiltonian-complexity phenomena.
\end{abstract}

\tableofcontents
\newpage

\section{Introduction}
\label{sec:introduction}

The Local Hamiltonian problem is a central problem in quantum complexity theory. 
Given a Hamiltonian
\[
H=\sum_i H_i
\]
whose terms act on only a constant number of qubits, the task is to distinguish whether its ground-state energy is at most $a$ or at least $b$, under the promise of an inverse-polynomial gap $b-a$. 
Kitaev's circuit-to-Hamiltonian construction established QMA-completeness of the $5$-Local Hamiltonian problem \cite{Kitaev2002}; subsequent work reduced the required locality and established hardness under additional spatial restrictions \cite{Kempe2004,Oliveira2005}.

In addition to the choice of interaction supports, a standard Local Hamiltonian instance has two further termwise freedoms: the local interaction acting on each hyperedge and the coefficient multiplying that interaction. 
This raises a basic structural question. 
How much Hamiltonian complexity survives if both of these freedoms are removed, so that every local term is exactly the same interaction with exactly the same coupling strength?

We study this question through the \emph{Single-Interaction Hamiltonian} framework. 
Fix a $k$-qubit Hermitian interaction $h$. 
An $\mathrm{SIH}_k$ instance is specified by a collection $L$ of ordered $k$-tuples of physical qubits and has Hamiltonian
\[
H_L=\sum_{e\in L} h_e.
\]
Thus every occurrence of the interaction has coefficient one, and, once $h$ is fixed, all nontrivial instance dependence in the Hamiltonian is carried by the ordered interaction hypergraph. 
This restriction is stronger than fixing only the set or type of allowed interactions while
retaining independently tunable coupling strengths.

We also consider an $n$-dependent variant, $\mathrm{DSIH}_k$, in which the common interaction belongs to a uniformly computable family $\{h^{(n)}\}_{n\geq 1}$ depending only on the final physical system size. 
In both the SIH and DSIH settings we study pinned, stoquastic, and geometrically local variants.  These models are distinct from translationally invariant Hamiltonians: the local interaction is uniform, but the interaction hypergraph remains instance dependent.

The results of this work show that this severe interaction uniformity still permits substantial computational complexity. 
At the heart of our results is an exact finite-alphabet compilation framework that preserves a controlled low-energy spectrum while enforcing a single common interaction.
Combined with a separate finite-alphabet normalization, it yields a universal reduction from $k$-Local Hamiltonian to $\mathrm{SIH}_{3k+3}$.

Sharper locality bounds arise from specialized constructions. 
We obtain QMA-complete $\mathrm{SIH}_3$ and geometrically local QMA-complete $\mathrm{SIH}_8$, as well as StoqMA-complete $\mathrm{Stoq\text{-}SIH}_3$ and QMA-complete $1$-Pinned $\mathrm{Stoq\text{-}SIH}_4$.
Together with the $n$-dependent and further structural results developed below, these results show that independently chosen local matrices and independently tunable coupling strengths are not necessary for several standard hardness phenomena of Local Hamiltonian problems.

\subsection{The Single-Interaction Compiler}

A central technical contribution of this work is an exact compiler that converts any fixed finite interaction alphabet to a single interaction.
Let
\[
\mathcal A=\{h^{(1)},\ldots,h^{(m)}\}
\]
be a fixed collection of $k$-local Hermitian matrices, and consider a Hamiltonian
\[
K=\sum_{i=1}^{M} h^{(j_i)}_{V_i}
\]
whose local terms are drawn from $\mathcal A$. 
The compiler augments the working system with auxiliary selector qubits and constructs a single interaction whose action depends on their local configuration.

In the pinned construction, the auxiliary selector register is fixed to a prescribed state, so compression to the pinned sector reproduces $K$ exactly. 
The unpinned construction replaces this external restriction by local energetic checks. 
Each target occurrence receives its own constant-size auxiliary block, and the correct selector configuration is separated energetically from every invalid configuration.

The resulting compiler preserves considerably more than the ground-state energy. 
By \cref{thm:unpinned-spectral-preservation}, for any chosen separation parameter $\gamma>0$, the valid auxiliary sector is an exact isometric copy of $K$, the ground space is preserved exactly, and every eigenvalue below
\[
\lambda_{\min}(K)+\gamma
\]
is reproduced with the same multiplicity.
In particular, if the spectral gap of $K$ is at most $\gamma$, then that gap is preserved exactly. 
For a fixed alphabet of size $m$, the compiled interaction has locality
\[
k+\left\lceil\log_2 m\right\rceil+2.
\]
The construction also preserves bounded degree and geometric locality for geometrically local finite-alphabet Hamiltonians, and preserves stoquasticity when every interaction in the alphabet is stoquastic.

Extending the construction to arbitrary fixed-locality Hamiltonians requires a separate normalization step before the exact compiler, described in \cref{sec:universal-finite-alphabet-normalization}. 
Each local term is expanded in a fixed Pauli basis, its coefficients are rounded to a common inverse-polynomial scale, and the resulting integer coefficients are represented by multiplicities. 
For every inverse-polynomial accuracy $\varepsilon$, this produces a finite-alphabet Hamiltonian $K$ and an efficiently computable integer $Q$ satisfying
\[
\left\|
H-\frac{1}{Q}K
\right\|
\leq \varepsilon.
\]
Composing this normalization with the exact compiler yields the universal reduction of \cref{thm:universal-single-interaction-reduction}. For every fixed $k$, the $k$-Local Hamiltonian problem admits a polynomial-time, promise-preserving reduction to $\mathrm{SIH}_{3k+3}$, up to the known global energy scale $Q$ and the chosen inverse-polynomial approximation.

The distinction between these two stages is important. 
The finite-alphabet compiler is exact and can therefore transport low-energy structural information, whereas the preceding normalization of arbitrary local terms need not preserve properties such as exact ground-space degeneracy, frustration freeness, geometric locality, stoquasticity, or the exact spectral gap. 
Several of the later results therefore use the exact compiler directly on structured finite alphabets rather than invoking the universal normalization.

Finally, when the source alphabet consists of a single interaction with nonnegative coupling strengths, the generic compiler can be sharpened.
The construction of \cref{cor:positive-weight-singleton-sih} converts a positive-weight
singleton $k$-local Hamiltonian to an unweighted $\mathrm{SIH}_{k+1}$ problem. 
This specialized reduction underlies the three-local QMA- and StoqMA-completeness results developed below.

\subsection{Main Results}
\label{subsec:intro-main-results}

Our results show that strict interaction uniformity is compatible with hardness across several standard Hamiltonian-complexity settings. 
The principal complexity classifications are summarized in \cref{tab:main_results}.

\paragraph{Constant interactions.}
We first obtain QMA-completeness with a fixed three-local interaction. 
Starting from the QMA-complete positive-weight antiferromagnetic Heisenberg Hamiltonian ~\cite{Piddock2015}, the specialized singleton unweighting of \cref{cor:positive-weight-singleton-sih} replaces all coupling strengths by repeated unit-strength occurrences and one auxiliary qubit per occurrence.  
This gives \cref{thm:sih3-qma-complete}:
\[
\boxed{\mathrm{SIH}_3\ \text{is QMA-complete}}
\]
for an explicit fixed $3$-local interaction.

The multiplicity construction underlying this result does not preserve bounded degree.  We therefore treat geometric locality separately.
Using a spatially localized fixed-gate circuit-to-Hamiltonian construction ~\cite{Oliveira2005}, together with local auxiliary checks, we prove in \cref{thm:unpinned-fixed} that
\[
\boxed{\text{geometrically local }\mathrm{SIH}_8
       \ \text{is QMA-complete}},
\]
with both bounded interaction range and bounded vertex degree.

\paragraph{$n$-dependent interactions.}
We next consider uniform interaction families whose common matrix may depend only on the final physical system size. 
Starting from the QMA-complete XY Hamiltonian at fixed Hamming weight \cite{Childs2015}, we replace the global sector restriction by a local energetic penalty and encode the required parameters into the final system size. 
This yields \cref{thm:pinned-n-dependent,thm:unpinned-n-dependent}:
\[
\boxed{
2\text{-Pinned }\mathrm{DSIH}_3
\ \text{and unpinned }\mathrm{DSIH}_5
\ \text{are QMA-complete}.
}
\]

\paragraph{Stoquastic interactions.}
The same singleton-unweighting principle also applies within the stoquastic setting. 
Using a StoqMA-complete positive-weight two-qubit interaction ~\cite{Piddock2015}, we obtain \cref{thm:stoq-sih3-stoqma-complete}:
\[
\boxed{\mathrm{Stoq\text{-}SIH}_3
       \ \text{is StoqMA-complete}}
\]
for a fixed stoquastic interaction.

We also formulate realification and the stoquastic sign embedding directly in the single-interaction framework. 
For real interactions, the embedding increases the locality by one and introduces one pinned qubit while preserving the restricted spectrum exactly. 
Applying it to the QMA-complete $\mathrm{SIH}_3$ interaction gives
\cref{thm:pinned-stoq-sih4}:
\[
\boxed{
1\text{-Pinned }\mathrm{Stoq\text{-}SIH}_4\ \text{is QMA-complete}.
}
\]

Finally, combining a geometrically local StoqMA verification Hamiltonian~\cite{Waite2025}, with a uniform $n$-dependent single interaction gives \cref{thm:StoqMA-complete}:
\[
\boxed{
\text{geometrically local }\mathrm{Stoq\text{-}DSIH}_9\ \text{is StoqMA-complete}.
}
\]

\begin{table}[ht]
\centering
\renewcommand{\arraystretch}{1.2}
\begin{tabular}{@{}llcl@{}}
\hline
\textbf{Framework} &
\textbf{Boundary condition} &
\textbf{Locality} &
\textbf{Complexity}
\\
\hline
$\mathrm{SIH}$ &
Unpinned &
3 &
QMA-complete
\\
Geometric $\mathrm{SIH}$ &
Unpinned &
8 &
QMA-complete
\\
$\mathrm{DSIH}$ &
$2$-Pinned &
3 &
QMA-complete
\\
$\mathrm{DSIH}$ &
Unpinned &
5 &
QMA-complete
\\
$\mathrm{Stoq\text{-}SIH}$ &
Unpinned &
3 &
StoqMA-complete
\\
$\mathrm{Stoq\text{-}SIH}$ &
$1$-Pinned &
4 &
QMA-complete
\\
Geometric $\mathrm{Stoq\text{-}DSIH}$ &
Unpinned &
9 &
StoqMA-complete
\\
\hline
\end{tabular}
\caption{Principal complexity classifications established in this work.}
\label{tab:main_results}
\end{table}

\paragraph{Structural consequences.}
The compiler also transfers several other Hamiltonian-complexity phenomena to the single-interaction setting. 
First, for the Hamiltonian quantum PCP conjecture~\cite{AharonovAradVidick2013}, \cref{cor:sih-qpcp-normal-form} shows that the Hamiltonian quantum PCP conjecture is equivalent to its restriction to Hamiltonians in which every local term is an unweighted copy of one fixed positive-semidefinite interaction. 
Thus independently chosen local matrices and independently tunable coupling strengths are not required for the qPCP question.

Second, exact finite-alphabet compilation is useful for promises that are unstable under approximation. 
Using perfect-completeness QMA verification ~\cite{GrewalRudolph2026}, \cref{thm:ff-sih9} gives a fixed positive-semidefinite $9$-local interaction for which
\[
\boxed{\mathrm{FF\text{-}SIH}_9\ \text{is QMA-complete}.}
\]
For exponentially precise energies, using $\mathrm{QMA}_{\exp}=\mathrm{PSPACE}$~\cite{FeffermanLin2016}, \cref{thm:precise-geometric-sih8} shows that the same fixed $8$-local interaction used in the geometrically local QMA construction also gives
\[
\boxed{
\text{geometrically local }\mathrm{Precise\text{-}SIH}_8\ \text{is PSPACE-complete}.
}
\]
Finally, applying the compilation framework to the Guided Local Hamiltonian construction of Cade et al.~\cite{CadeEtAl2023},
\cref{thm:guided-sih9} establishes
\[
\boxed{\mathrm{Guided\text{-}SIH}_9\ \text{is BQP-complete},}
\]
even with a unique inverse-polynomially gapped ground state and a semi-classical guiding state whose overlap with that ground state is $1-1/\operatorname{poly}(n)$.

\subsection{Relation to Previous Work}
\label{subsec:related-work}

Several established Hamiltonian-complexity frameworks impose forms of interaction uniformity closely related to, but distinct from, the single-interaction restriction considered here.

\paragraph{Fixed interaction sets and weighted interactions.}
Cubitt and Montanaro study Local Hamiltonian problems in which the allowed local terms are restricted to a fixed finite interaction set $S$ \cite{Cubitt2013}. Of particular relevance, Piddock and Montanaro consider Hamiltonians generated by a single fixed two-qubit interaction $h$ with independently chosen coupling strengths
\cite{Piddock2015}:
\[
H=\sum_e \alpha_e h_e .
\]
They establish complexity classifications for broad classes of fixed interactions. 
In particular, positive-weight antiferromagnetic Heisenberg and XY interactions yield QMA-complete problems.

The strict SIH model removes the remaining coefficient freedom:
\[
H=\sum_{e\in L} h_e ,
\]
so every occurrence has coefficient exactly one. 
Thus, in the single-interaction case, the distinction from these weighted models is not the number of available interaction matrices, but whether instance information may be encoded in independently tunable coupling strengths. 
The multiplicity-based constructions developed here show that, for several complexity questions, this coefficient freedom can itself be removed at constant locality overhead.

The two-local unweighted regime also contains problems of independent interest. 
In particular, unweighted Quantum Max-Cut is, up to an overall normalization and the max-versus-min convention, a strict fixed-interaction two-local problem generated by repeated copies of the
singlet projector. 
Recent work shows that even constant-factor approximation remains NP-hard for unweighted bounded-degree instances \cite{Piddock2025}. 
Our results do not classify the general two-local SIH problem; the explicit QMA-complete fixed interactions established here begin at locality three.

\paragraph{Hamiltonian simulation with a fixed interaction type.}
A related but stronger notion of universality arises in analogue Hamiltonian simulation, where a simulator is required to reproduce low-energy spectral and state information of a target Hamiltonian.
Zhou and Aharonov show that strongly universal two-dimensional simulators can use only a single type of nearest-neighbor interaction \cite{Zhou2021}. 
Their semi-translation-invariant Hamiltonians have the form
\[
H'=\sum_{\langle i,j\rangle} J_{ij} h_{ij},
\]
where the two-body operator $h$ is fixed but the interaction energies $J_{ij}$ vary spatially.  This is therefore complementary to the SIH restriction: their results provide substantially stronger simulation guarantees while retaining tunable coupling strengths, whereas the present work removes those coupling strengths and studies the resulting complexity-theoretic normal form. 
Correspondingly, the universal reduction of
\cref{thm:universal-single-interaction-reduction}
should be viewed as a ground-energy complexity reduction rather than as a general analogue-simulation theorem.

\paragraph{Sector restrictions and pinning.}
The antiferromagnetic XY model of Childs, Gosset, and Webb provides a particularly close example of an unweighted fixed interaction \cite{Childs2015}. 
Every graph edge carries the same XY interaction, but QMA-hardness is established for energy minimization restricted to a prescribed magnetization, equivalently Hamming-weight, sector. 
Our $n$-dependent construction uses this problem as a source and replaces the global sector restriction by local energetic constraints.

Pinning gives another mechanism for restricting the set of states over which the energy is minimized.
Nagaj et al.\ studied pinned Local Hamiltonian problems and showed, in particular, that pinning a single qubit can yield QMA-complete stoquastic Hamiltonian problems \cite{Nagaj2020}. 
In the present work, pinning is also used as a selector resource, while the stoquastic single-interaction embedding shows that the same phenomenon persists when every local term is an occurrence of one fixed stoquastic interaction.

\paragraph{Translational and geometric restrictions.}
Translationally invariant Hamiltonians impose a different form of homogeneity. 
In the construction of Gottesman and Irani, a fixed local interaction is repeated across a rigid architecture and the system size itself supplies the instance \cite{Gottesman2009}. 
Subsequent work reduced the required qudit dimension while retaining translational invariance \cite{Bausch2016}. 
SIH separates these two notions of uniformity: the local matrix is fixed, but the ordered interaction hypergraph remains instance dependent.

For geometrically constrained hardness, we build on the spatially sparse circuit-to-Hamiltonian construction of Oliveira and Terhal \cite{Oliveira2005} and its stoquastic extension by Waite and Bremner \cite{Waite2025}. 
The explicit layouts used here satisfy our bounded-degree and bounded-interaction-range definition after constant rescaling. 
These source constructions are then combined with local single-interaction auxiliary structures to obtain the geometrically local results of \cref{thm:unpinned-fixed,thm:StoqMA-complete}.

\subsection{Organization}
\label{subsec:organization}

The remainder of the paper is organized as follows.
\Cref{sec:background} reviews the complexity-theoretic and Hamiltonian background used throughout the work, and \cref{sec:sih-framework} defines the SIH and DSIH frameworks and their variants.

\Cref{sec:general_methodology} develops the single-interaction compilation machinery, including the exact finite-alphabet compilers, weighted normalization, singleton unweighting, and the universal reduction. 
\Cref{sec:hardness_single_interaction} then proves QMA-completeness of $\mathrm{SIH}_3$ and geometrically local $\mathrm{SIH}_8$, while \cref{sec:n-dependent-hardness} establishes the corresponding $n$-dependent results.

\Cref{sec:stoquastic-embedding-results} develops the stoquastic single-interaction constructions, including StoqMA-completeness of $\mathrm{Stoq\text{-}SIH}_3$ and QMA-completeness of $1$-Pinned $\mathrm{Stoq\text{-}SIH}_4$.
\Cref{sec:stoq-n-dependent-hardness} proves StoqMA-completeness of geometrically local $\mathrm{Stoq\text{-}DSIH}_9$, and \cref{sec:further-consequences} develops the qPCP, frustration-free, precise, and guided consequences.

Finally, \cref{sec:discussion} discusses implications and open directions. 
The appendices collect the geometrically local source constructions, explicit routing details, and the system-size decoding and uniformity bookkeeping used in the body.

\section{Background and Preliminaries}
\label{sec:background}

\subsection{Notation}

We use the following notation throughout the paper.

\begin{itemize}
\item \textbf{Operator support.}
Subscripts indicate the physical qubits on which an operator acts; for example, $M_{v_1,v_2}$ denotes the operator $M$ acting on qubits $v_1$ and $v_2$. When an interaction acts on several qubits, we also use argument notation such as $h(v_1,\ldots,v_k)$.

\item \textbf{Projectors.}
For a single-qubit state $\Ket{i}$, we write
\[
    \Pi^{\Ket{i}}=\Ket{i}\Bra{i}.
\]
In particular, we frequently use
$\Pi^{\Ket{0}}$, $\Pi^{\Ket{1}}$, and
$\Pi^{\Ket{-}}$.

\item \textbf{Registers.}
Multi-qubit registers are denoted by capital letters such as $W$ for a work register and $A$ for an auxiliary register. 
A subscript identifies a physical qubit within the register; for example, $W_i$ denotes the $i$-th qubit of $W$.

\end{itemize}

Throughout the reductions below, we freely pad verification circuits
with identity operations, or with exactly cancelling pairs of gates from
the relevant fixed gate set, so that the circuit length $T$ dominates
$|x|$ and any fixed polynomial in $|x|$ appearing in the relevant
completeness--soundness parameters. This padding leaves the acceptance
probabilities unchanged. Consequently, inverse-polynomial and
inverse-exponential bounds expressed in terms of $|x|$ remain,
respectively, inverse polynomial and inverse exponential in the final
Hamiltonian system size.

\subsection{Complexity Classes}

A promise problem is a pair $L=(L_{\mathrm{yes}},L_{\mathrm{no}})$ of disjoint sets of strings. 
An algorithm for $L$ is required to distinguish inputs in $L_{\mathrm{yes}}$ from inputs in $L_{\mathrm{no}}$; no condition is imposed on inputs outside $L_{\mathrm{yes}}\cup L_{\mathrm{no}}$.

\paragraph{Bounded-error quantum computation.}
The complexity class $\mathrm{BQP}$ consists of promise problems that can be decided with bounded error by polynomial-size uniform quantum circuits.

\begin{defn}[$\mathrm{BQP}$~\cite{Bernstein1997}]
\label{def:BQP}
A promise problem $L=(L_{\mathrm{yes}},L_{\mathrm{no}})$ is in $\mathrm{BQP}$ if there exists a deterministic polynomial-time classical algorithm $F$ such that, for every input string $x$, the algorithm $F(x)$ outputs a polynomial-size quantum circuit $U_x$ with a designated output qubit satisfying
\[
x\in L_{\mathrm{yes}}
\quad\Longrightarrow\quad
\Pr(U_x\text{ accepts})\ge \frac23,
\]
and
\[
x\in L_{\mathrm{no}}
\quad\Longrightarrow\quad
\Pr(U_x\text{ accepts})\le \frac13.
\]
As usual, these constants can be amplified exponentially close to $1$ and $0$, respectively, with only polynomial overhead.
\end{defn}

The complexity class \textbf{QMA} consists of promise problems for which YES instances admit polynomial-size quantum witnesses that can be verified by polynomial-size quantum circuits with bounded error.
For the circuit-to-Hamiltonian reductions used below, we work with exponentially small completeness and soundness error obtained by standard QMA amplification.

\begin{defn} [$QMA$, adapted from \cite{Kitaev2002}, \cite{Kempe2004}] \label{def:QMA}
A promise problem $L=(L_{\mathrm{yes}},L_{\mathrm{no}})$ is in \textbf{QMA} if there exists a deterministic polynomial-time classical algorithm $F$ and polynomials $n_0, n_1, m, t$ such that for every input string $x\in\left\{ 0,1\right\} ^{*}$, $F(x)$ outputs a quantum circuit $U_x$.
The quantum circuit $U_x$ acts on $n_0(|x|) + n_1(|x|) + m(|x|)$ qubits, consists of $t(|x|)$ gates from a finite, universal quantum gate set, and is followed by a measurement of a designated output qubit in the computational basis $\{\Ket{0},\Ket{1}\}$. The acceptance conditions are:

\begin{itemize}
    \item[] $\forall x \in L_{\mathrm{yes}} \quad \exists \Ket{\xi}\in\left(\mathbb{C}^{2}\right)^{\otimes m(|x|)} \quad \Pr\left(U_x \text{ accepts } \Ket{0}^{\otimes n_{0}}\otimes\Ket{1}^{\otimes n_{1}}\otimes\Ket{\xi} \right)\geq 1-\varepsilon $
    \item[] $\forall x \in L_{\mathrm{no}} \quad \forall \Ket{\xi}\in\left(\mathbb{C}^{2}\right)^{\otimes m(|x|)} \quad \Pr\left(U_x \text{ accepts } \Ket{0}^{\otimes n_{0}}\otimes\Ket{1}^{\otimes n_{1}}\otimes\Ket{\xi} \right)\leq \varepsilon$
\end{itemize}

where $\varepsilon=2^{-\Omega(|x|)}$. The term $\Pr(U_x \text{ accepts } \dots)$ denotes the probability that the measurement of the output qubit yields 1 after applying $U_x$ to the state $\Ket{0}^{\otimes n_{0}}\otimes\Ket{1}^{\otimes n_{1}}\otimes\Ket{\xi}$.
\end{defn}

\begin{rem}
A standard circuit formulation of QMA supplies the classical input $x$ to a uniform verifier and typically initializes ancillas in $\Ket{0}$.
For convenience, we instead hardwire $x$ into the verifier description and allow a fixed computational-basis ancilla string. 
This equivalent formulation avoids requiring the chosen fixed gate set itself to prepare every $\Ket{1}$ ancilla from $\Ket{0}$.
\end{rem}

\paragraph{Perfect completeness.}
The perfect-completeness variant of QMA, denoted $\mathrm{QMA}_1$~\cite{Bravyi2006b}, is defined as follows: YES instances admit a witness that is accepted with probability exactly one, while on NO instances every witness is accepted with probability at most $s<1$, where $1-s$ is at least inverse polynomial in the input size.

Grewal and Rudolph proved that perfect completeness does not change the power of quantum Merlin--Arthur verification.

\begin{thm}[Grewal--Rudolph~\cite{GrewalRudolph2026}]
\label{thm:qma-perfect-completeness}
\[
\mathrm{QMA}=\mathrm{QMA}_1.
\]
Moreover, the perfectly complete verifier may be taken to use only the fixed gate set
\[
\{H,X,\mathrm{Toffoli}\}.
\]
\end{thm}

\paragraph{Exponentially small verification gaps.}
We also use the exponentially precise variant of quantum Merlin--Arthur verification. 
Let $\mathrm{QMA}_{\exp}$ denote the class defined as in $\mathrm{QMA}$, except that the completeness and soundness parameters are only required to satisfy
\[
c-s\ge 2^{-p(n)}
\]
for some polynomial $p$.

Fefferman and Lin proved that allowing an inverse-exponentially small verification gap increases the power of QMA to polynomial space.

\begin{thm}[Fefferman--Lin~\cite{FeffermanLin2016}]
\label{thm:qma-exp-pspace}
\[
\mathrm{QMA}_{\exp}=\mathrm{PSPACE}.
\]
\end{thm}

\paragraph{StoqMA.}
Stoquastic Hamiltonians motivate a corresponding restricted verification model. The complexity class StoqMA captures promise problems verified by quantum circuits whose computational gates are classical reversible operations, supplemented by prescribed $\Ket{0}$, $\Ket{1}$, and $\Ket{+}$ ancillas and a final measurement in the Hadamard basis.

\begin{defn} [$StoqMA$, adapted from \cite{Bravyi2006}] \label{def: StoqMA}
A promise problem $L=(L_{\mathrm{yes}},L_{\mathrm{no}})$ is in \textbf{StoqMA} if there exists a deterministic polynomial-time classical algorithm $F$ and polynomials $n_0, n_1, n_+, m, t$ such that for every input string $x\in\left\{ 0,1\right\} ^{*}$, $F(x)$ outputs a quantum circuit $U_x$. 
The circuit $U_x$ acts on $n_0(|x|) + n_1(|x|) + n_+(|x|) + m(|x|)$ qubits, consists of $t(|x|)$ gates from a finite gate set that is universal for classical reversible computation (i.e., operations restricted to permutations of the computational basis), and is followed by a measurement of a designated output qubit in the Hadamard basis $\{\Ket{+},\Ket{-}\}$. The acceptance conditions are:

\begin{itemize}
    \item[] $\forall x \in L_{\mathrm{yes}} \quad \exists \Ket{\xi}\in\left(\mathbb{C}^{2}\right)^{\otimes m(|x|)} \quad \Pr\left(U_x \text{ accepts } \Ket{0}^{\otimes n_{0}}\otimes\Ket{1}^{\otimes n_{1}}\otimes\Ket{+}^{\otimes n_{+}}\otimes\Ket{\xi} \right)\geq \epsilon_{\mathrm{yes}} $
    \item[] $\forall x \in L_{\mathrm{no}} \quad \forall \Ket{\xi}\in\left(\mathbb{C}^{2}\right)^{\otimes m(|x|)} \quad \Pr\left(U_x \text{ accepts } \Ket{0}^{\otimes n_{0}}\otimes\Ket{1}^{\otimes n_{1}}\otimes\Ket{+}^{\otimes n_{+}}\otimes\Ket{\xi} \right)\leq \epsilon_{\mathrm{no}}$
\end{itemize}

where the threshold probabilities satisfy $0\leq \epsilon_{\mathrm{no}}<\epsilon_{\mathrm{yes}}\leq1$ and have a polynomial separation: $\epsilon_{\mathrm{yes}}-\epsilon_{\mathrm{no}} \geq \frac{1}{\operatorname{poly}(|x|)}$. 
The term $\Pr(U_x \text{ accepts } \dots)$ denotes the probability that the measurement of the output qubit yields outcome $+$ after applying $U_x$ to the state $\Ket{0}^{\otimes n_{0}}\otimes\Ket{1}^{\otimes n_{1}}\otimes\Ket{+}^{\otimes n_{+}}\otimes\Ket{\xi}$.
\end{defn}

\begin{rem}
Unlike QMA, which admits exponential amplification of its completeness and soundness parameters~\cite{Marriott2005}, no general error-reduction procedure is known for StoqMA that preserves the stoquastic-verifier structure~\cite{Aharonov2020}. 
We therefore retain generic completeness and soundness thresholds separated by an inverse-polynomial gap.
\end{rem}

The complexity classes used in this work satisfy the standard containments
\begin{align}
\mathrm{NP}
\subseteq
\mathrm{MA}
\subseteq
\mathrm{StoqMA}
\subseteq
\mathrm{QMA}
\subseteq
\mathrm{PP}
\subseteq
\mathrm{PSPACE},
\end{align}
together with
\[
\mathrm{BQP}\subseteq\mathrm{QMA},
\qquad
\mathrm{QMA}=\mathrm{QMA}_1,
\qquad
\mathrm{QMA}_{\exp}=\mathrm{PSPACE}.
\]
The inclusion $\mathrm{MA}\subseteq\mathrm{StoqMA}\subseteq\mathrm{QMA}$ follows from the verifier definitions~\cite{Bravyi2006}, while $\mathrm{QMA}\subseteq\mathrm{PP}$ follows from \cite{Marriott2005}. 
The equality $\mathrm{QMA}=\mathrm{QMA}_1$ is due to Grewal and Rudolph~\cite{GrewalRudolph2026}, and $\mathrm{QMA}_{\exp}=\mathrm{PSPACE}$ is due to Fefferman and Lin~\cite{FeffermanLin2016}.

\subsection{The Local Hamiltonian Problems}

The $k$-Local Hamiltonian problem is the central decision problem of Hamiltonian complexity. An instance specifies a Hamiltonian $H=\sum_i H_i$ whose terms each act nontrivially on at most $k$ qubits, and asks whether the ground-state energy of $H$ lies below a threshold $a$ or above a threshold $b$, subject to an inverse-polynomial promise gap.

\begin{defn} [$k$-Local Hamiltonian Problem, adapted from \cite{Kitaev2002}, \cite{Kempe2004}]\mbox{}\\
\noindent\textbf{Problem Input:} A set of $m = \operatorname{poly}(n)$ Hermitian matrices $H_1, \dots, H_m$ acting on a system of $n$ qubits. 
Each term $H_i$ acts nontrivially on at most $k$ qubits, satisfies the operator norm bound $\|H_i\| \leq \operatorname{poly}(n)$ and its entries are specified by $\operatorname{poly}(n)$ bits. 
The total Hamiltonian is given by $H = \sum_{i=1}^m H_i$. 
Additionally, the input includes two real numbers $a$ and $b$, specified with $\operatorname{poly}(n)$ bits of precision, such that $b-a \geq 1/\operatorname{poly}(n)$.

\noindent\textbf{Promise:}  One of the following holds:
\begin{itemize}
    \item There exists a quantum state $\Ket{\phi}$ such that its energy with respect to $H$ is at most $a$ (i.e., $\Bra{\phi}H\Ket{\phi} \leq a$).
    \item For every quantum state $\Ket{\phi}$, its energy is at least $b$ (i.e., $\Bra{\phi}H\Ket{\phi} \geq b$).
\end{itemize}

\noindent\textbf{Output:} Determine which of these two possibilities is the case.
\end{defn}

\paragraph{Stoquastic Local Hamiltonian.}
A Hamiltonian is called \emph{stoquastic} in the computational basis if all of its off-diagonal matrix elements are real and non-positive. 
The Stoquastic $k$-Local Hamiltonian problem is obtained by imposing this condition on every local term of a standard $k$-Local Hamiltonian instance.

\begin{defn} [Stoquastic $k$-Local Hamiltonian Problem, adapted from \cite{Bravyi2006a}]
The Stoquastic $k$-Local Hamiltonian Problem is defined exactly as the $k$-Local Hamiltonian Problem, with the additional constraint that each term $H_i$ in the input must be stoquastic in the standard computational basis. 
That is, for every $i \in [m]$ and all standard basis states $\Ket{x} \neq \Ket{y}$, the matrix elements satisfy:
\[ \Bra{x} H_i \Ket{y} \leq 0 \]
and are real-valued.
\end{defn}

\paragraph{Pinned Local Hamiltonian.}
In a pinned Local Hamiltonian problem, the energy minimization is restricted to states in which a specified set of qubits is fixed to a prescribed state. Equivalently, the optimization is performed only over the remaining unpinned qubits.

\begin{defn} [$p$-Pinned $k$-Local Hamiltonian Problem, adapted from \cite{Nagaj2020}]\mbox{}\\
\noindent\textbf{Problem Input:} A set of $m = \operatorname{poly}(n)$ Hermitian matrices $H_1, \dots, H_m$ acting on a system of $n$ qubits. 
Each term $H_i$ acts nontrivially on at most $k$ qubits, satisfies the operator norm bound $\|H_i\| \leq \operatorname{poly}(n)$ and its entries are specified by $\operatorname{poly}(n)$ bits. 
The total Hamiltonian is given by $H = \sum_{i=1}^m H_i$. 
Additionally, the input includes a description of a fixed quantum state $\Ket{\psi}$ on $p$ qubits, where $p=\operatorname{poly}(n)$, and two real numbers $a$ and $b$, specified with $\operatorname{poly}(n)$ bits of precision, such that $b-a \geq 1/\operatorname{poly}(n)$.

\noindent\textbf{Promise:} One of the following holds:
\begin{itemize}
    \item There exists a quantum state $\Ket{\phi}$ on $n-p$ qubits such that the energy of the joint $n$-qubit state $\Ket{\phi} \otimes \Ket{\psi}$ with respect to $H$ is at most $a$ (i.e., $(\Bra{\phi}\otimes\Bra{\psi}) H (\Ket{\phi}\otimes\Ket{\psi}) \leq a$).
    \item For any quantum state $\Ket{\phi}$ on $n-p$ qubits, the energy of the joint state $\Ket{\phi} \otimes \Ket{\psi}$ with respect to $H$ is at least $b$.
\end{itemize}

\noindent\textbf{Output:} Determine which of these two possibilities is the case.
\end{defn}

For every fixed $p\geq 1$ and $k\geq 2$, the $p$-Pinned $k$-Local Hamiltonian problem is QMA-complete. 
QMA-hardness follows from the QMA-completeness of $k$-Local Hamiltonian~\cite{Kempe2004} by adjoining $p$ unused pinned qubits, while containment in QMA is immediate.

\paragraph{Pinned stoquastic Local Hamiltonian.}
Combining the two restrictions gives the pinned stoquastic variant.

\begin{defn} [$p$-Pinned Stoquastic $k$-Local Hamiltonian Problem, adapted from \cite{Nagaj2020}]\mbox{}\\
\noindent The $p$-Pinned Stoquastic $k$-Local Hamiltonian Problem is defined as the $p$-Pinned $k$-Local Hamiltonian Problem, with the additional constraint that each term $H_i$ in the input Hamiltonian must be stoquastic in the standard computational basis.
\end{defn}

However, pinning can substantially change the complexity of stoquastic Local Hamiltonian problems. Nagaj et al.~\cite{Nagaj2020} showed that fixing a single qubit is already sufficient to obtain QMA-completeness.

\begin{thm}[\cite{Nagaj2020}]
The $1$-Pinned Stoquastic $3$-Local Hamiltonian problem is \textbf{QMA}-complete.
\end{thm}

Consequently, the problem remains QMA-complete for every fixed $p\geq1$ and $k\geq3$: additional pinned qubits may be left unused, and a $3$-local Hamiltonian is also $k$-local for every $k\geq3$.

Thus stoquasticity of the local terms does not preclude QMA-hardness when the energy minimization is restricted to a pinned sector. 
This observation motivates the pinned stoquastic single-interaction problems considered later in the paper.

\paragraph{Exponentially precise Local Hamiltonian.}
For the exponentially precise setting, the Local Hamiltonian promise gap is allowed to be inverse exponential rather than inverse polynomial.

\begin{defn}[Precise $k$-Local Hamiltonian Problem, adapted from \cite{FeffermanLin2016}] \label{def:precise-local-hamiltonian} The \emph{Precise $k$-Local Hamiltonian Problem} is defined identically to the $k$-Local Hamiltonian Problem, except that the promise-gap condition is replaced by
\[
b-a \ge 2^{-p(n)}
\]
for some fixed polynomial $p$.
\end{defn}

Fefferman and Lin proved that allowing an inverse-exponentially small promise gap raises the complexity of Local Hamiltonian from QMA to PSPACE.

\begin{thm}[Fefferman--Lin~\cite{FeffermanLin2016}]
\label{thm:precise-local-hamiltonian}
For every fixed $k\ge 3$, the Precise $k$-Local Hamiltonian problem is PSPACE-complete.
\end{thm}

\paragraph{Constant-relative gaps and the Hamiltonian quantum PCP conjecture.}
For a Local Hamiltonian
\[
H=\sum_{i=1}^{M}H_i,
\qquad
0\preceq H_i\preceq I,
\]
the promise gap is said to be \emph{constant relative} when
\[
b-a\ge cM
\]
for some constant $c>0$ independent of the instance. 
Equivalently, the promise gap is a constant fraction of the natural extensive energy
scale set by the number of local terms.

The Hamiltonian quantum PCP conjecture asks whether there exists a constant locality $k$ for which the $k$-Local Hamiltonian problem remains QMA-hard in this constant-relative-gap regime; see, for example, \cite{AharonovAradVidick2013}.

\paragraph{Guided Local Hamiltonians.}
The Guided Local Hamiltonian problem supplements a Local Hamiltonian instance with a classically described state that is promised to have substantial overlap with its ground space. We use the following class of guiding states.

\begin{defn}[Semi-classical state~\cite{Gharibian2021}]
\label{def:semiclassical-subset-state}
An $n$-qubit state $\Ket{u}$ is \emph{semi-classical} if there exists
a nonempty subset
\[
S\subseteq\{0,1\}^{n},
\qquad
|S|=\operatorname{poly}(n),
\]
such that
\[
\Ket{u}
=
\frac{1}{\sqrt{|S|}}
\sum_{x\in S}\Ket{x}.
\]
The classical description of $\Ket{u}$ is the explicit enumeration of the elements of $S$.
\end{defn}

Following the terminology of Cade et al.~\cite{CadeEtAl2023}, we also refer to such states as \emph{semi-classical subset states}.

\begin{defn}[Guided $k$-Local Hamiltonian Problem~\cite{Gharibian2021}]
\label{def:guided-local-hamiltonian}
Let $k\ge 2$, let $a,b\in[0,1]$ satisfy $a<b$, and let
$\delta\in(0,1]$.

An instance of the \emph{Guided $k$-Local Hamiltonian Problem} consists of a $k$-local Hamiltonian $H$ acting on $n$ qubits with
\[
\|H\|\le 1,
\]
together with the classical description of a semi-classical state $\Ket{u}$.

Let $\lambda_H$ denote the ground-state energy of $H$, and let $\Pi_H$ denote the orthogonal projector onto its ground space. 
The instance is promised to satisfy
\[
\|\Pi_H\Ket{u}\|\ge\delta,
\]
and
\[
\lambda_H\le a
\qquad\text{or}\qquad
\lambda_H\ge b.
\]
The task is to decide which of the two cases holds.
\end{defn}

Cade et al.~\cite{CadeEtAl2023} established BQP-hardness for Guided 2-Local Hamiltonian even when the guiding state has overlap $1-1/\operatorname{poly}(n)$ with the ground space.

\begin{cor}[Cade et al.~\cite{CadeEtAl2023}]
\label{cor:guided-local-hamiltonian-hardness}
There exist polynomials $p$ and $q$ and thresholds $a,b\in[0,1]$ satisfying
\[
b-a\ge \frac{1}{p(n)}
\]
such that the Guided $2$-Local Hamiltonian problem is BQP-hard even under the promise
\[
\|\Pi_H\Ket{u}\|
\ge
1-\frac{1}{q(n)},
\]
where $\Ket{u}$ is a semi-classical subset state.
\end{cor}

The construction of Cade et al.~\cite{CadeEtAl2023} may moreover be taken to have a unique ground state separated from the remainder of the spectrum by an inverse-polynomial spectral gap.

For semi-classical guiding states, Guided Local Hamiltonian is in $\mathrm{BQP}$ by quantum phase estimation, as observed by Gharibian and Le Gall~\cite{Gharibian2021}. 
Hence the problem is $\mathrm{BQP}$-complete in the above parameter regime.

\subsubsection{Kitaev's Circuit-to-Hamiltonian Construction}\label{subsubsec: Circuit-to-Hamiltonian}

Kitaev's circuit-to-Hamiltonian construction maps a quantum verification circuit to a local Hamiltonian whose ground-state energy distinguishes accepting from rejecting instances. In particular, it yields the following foundational result~\cite{Kitaev2002}.

\begin{thm} [\cite{Kitaev2002}]
The 5-Local Hamiltonian problem is QMA-complete.
\end{thm}

\paragraph{Construction.}
Let $L\in\mathrm{QMA}$ and let $x$ be an input. 
By \cref{def:QMA}, after standard amplification there is a polynomial-size verification circuit

\[
U_x=U_T\cdots U_1,
\]

where $T=\operatorname{poly}(|x|)$ and each $U_t$ acts on at most two qubits. 
The circuit acts on $ N=n_0+n_1+m $ qubits initialized as $ \Ket{0}^{\otimes n_0} \otimes \Ket{1}^{\otimes n_1} \otimes \Ket{\xi}, $ where $\Ket{\xi}$ is an $m$-qubit witness. 
If $x\in L_{\mathrm{yes}}$, some witness is accepted with probability at least $1-\varepsilon$, whereas if $x\in L_{\mathrm{no}}$, every witness is accepted with probability at most $\varepsilon$, where $\varepsilon=2^{-\Omega(|x|)}$.

The Hamiltonian acts on

\[
\mathcal{H}_{\mathrm{comp}}
\otimes
\mathcal{H}_{\mathrm{clock}},
\]

where the work register $W$ contains the computational qubits and the clock register $C$ contains $T$ qubits encoding the times $t\in\{0,\ldots,T\}$ in unary.

The Hamiltonian is
\begin{align}
    H
    =
    H_{\mathrm{in}}
    +
    H_{\mathrm{out}}
    +
    H_{\mathrm{prop}}
    +
    H_{\mathrm{clock}}.
\end{align}

These terms act as follows:

\paragraph{Initialization.} Ensures that the prescribed ancilla qubits are correctly initialized at time $t=0$.
\begin{align}
    H_{\mathrm{in}} = \sum_{i=1}^{n_0} \Pi^{\Ket{1}}_{W_i} \otimes \Pi^{\Ket{0}}_{C_1} + \sum_{i=n_0 + 1}^{n_0+n_1} \Pi^{\Ket{0}}_{W_i} \otimes \Pi^{\Ket{0}}_{C_1}
\end{align}

\paragraph{Output.} Penalizes the system if the output qubit (designated as qubit 1) indicates rejection at the final time step $t=T$.

\begin{align}
    H_{\mathrm{out}} = \Pi^{\Ket{0}}_{W_1} \otimes\Pi^{\Ket{1}}_{C_T}
\end{align}

\paragraph{Propagation.} Enforces that the state evolves according to the unitary gates $U_t$. It is a sum over all time steps

\begin{align}
    H_{\mathrm{prop}} = \sum_{t=1}^T H_{\mathrm{prop}}(t) \label{eq:hprop}
\end{align}

where for $2\leq t\leq T-1$:

\begin{align}
    H_{\mathrm{prop}}(t) =\frac{1}{2}\Pi^{\Ket{1}}_{C_{t-1}}\otimes \left( I\otimes\Ket{1}\Bra{1}_{C_t} + I\otimes\Ket{0}\Bra{0}_{C_t}-U_{t}\otimes\Ket{1}\Bra{0}_{C_t}-U_{t}^{\dagger}\otimes\Ket{0}\Bra{1}_{C_t} \right) \otimes \Pi^{\Ket{0}}_{C_{t+1}}
\end{align}

and

\begin{align}
    H_{\mathrm{prop}}(1) &=\frac{1}{2}\left( I\otimes\Ket{1}\Bra{1}_{C_1} + I\otimes\Ket{0}\Bra{0}_{C_1}-U_{1}\otimes\Ket{1}\Bra{0}_{C_1}-U_{1}^{\dagger}\otimes\Ket{0}\Bra{1}_{C_1} \right) \otimes \Pi^{\Ket{0}}_{C_{2}}\\
    H_{\mathrm{prop}}(T) &=\frac{1}{2}\Pi^{\Ket{1}}_{C_{T-1}}\otimes \left( I\otimes\Ket{1}\Bra{1}_{C_T} + I\otimes\Ket{0}\Bra{0}_{C_T}-U_{T}\otimes\Ket{1}\Bra{0}_{C_T}-U_{T}^{\dagger}\otimes\Ket{0}\Bra{1}_{C_T} \right)
\end{align}

\paragraph{Clock.} Penalizes clock states outside the valid unary-clock subspace.
\begin{align}
    H_{\mathrm{clock}} = \sum_{t=1}^{T-1} \Pi^{\Ket{0}}_{C_{t}} \otimes \Pi^{\Ket{1}}_{C_{t+1}} \label{eq:hclock}
\end{align}

\paragraph{History state and energy bounds.}

For a fixed witness $\Ket{\xi}$, define

\[
\Ket{\psi_t}
=
U_t\cdots U_1
\Ket{0}^{\otimes n_0}
\otimes
\Ket{1}^{\otimes n_1}
\otimes
\Ket{\xi},
\]

with the convention that the empty product at $t=0$ is the identity.
The corresponding history state is
\begin{equation}
\label{def: history state}
\Ket{\eta}
=
\frac{1}{\sqrt{T+1}}
\sum_{t=0}^{T}
\Ket{\psi_t}\otimes\Ket{t}_C .
\end{equation}
Such history states lie in the null space of $H_{\mathrm{in}} + H_{\mathrm{prop}} + H_{\mathrm{clock}}$.

The following two lemmas give the completeness and soundness bounds for Kitaev's construction.

\begin{lemma} [Completeness, \cite{Kitaev2002}] \label{lemma: kitaev completness}  
If $x\in L_{\mathrm{yes}}$, then there exists a witness $\Ket{\xi}$ whose corresponding history state satisfies
\[
\Bra{\eta}H\Ket{\eta}
\leq
\frac{\varepsilon}{T+1}.
\]

\end{lemma}

\begin{lemma} [Soundness, \cite{Kitaev2002}] \label{lemma: kitaev soundness} 
If $x\in L_{\mathrm{no}}$, then for every $\Ket{\Psi}\in\mathcal{H}$,
\[
\Bra{\Psi}H\Ket{\Psi}
\geq
\frac{c}{T^3}
\left(1-\sqrt{\varepsilon}\right),
\]
for some constant $c>0$ independent of $T$.
\end{lemma}

Since $T=\operatorname{poly}(|x|)$ and $\varepsilon=2^{-\Omega(|x|)}$, these bounds give an inverse-polynomial separation between the YES and NO ground-energy thresholds.

\subsubsection{A Geometrically Local QMA-Hard Source Hamiltonian}
\label{subsubsec: Geometrically Local}

Kitaev's circuit-to-Hamiltonian construction does not in general produce a geometrically local interaction hypergraph: circuit gates may act on distant work qubits, and individual work or clock qubits may participate in a number of Hamiltonian terms that grows with the circuit size.
We therefore use a spatially localized circuit-to-Hamiltonian construction in which both interaction range and vertex degree are bounded independently of the system size.

\begin{defn}[Geometrically Local Hypergraph] \label{def: geometrically local}
An interaction hypergraph $G=(V,E)$ is geometrically local if there exists an injective mapping of the vertices $V$ into a $D$-dimensional regular lattice $\mathbb{Z}^D$, for some fixed constant dimension $D=O(1)$, such that:

\begin{enumerate}
    \item \textbf{Bounded Degree:} Every vertex $v\in V$ participates
    in $O(1)$ hyperedges.

    \item \textbf{Bounded Interaction Range:} For every hyperedge $e\in E$, the lattice distance between any two vertices $u,v\in e$ is bounded by $O(1)$.

\end{enumerate}

All constants implicit in the $O(1)$ bounds are independent of the system size.
\end{defn}

\begin{rem}[Comparison with Spatially Sparse Hamiltonians]
\label{rem:spatial_sparsity_comparison}
Oliveira and Terhal~\cite{Oliveira2005} formulate their geometric restriction in terms of \emph{spatially sparse} interaction hypergraphs. 
We use the slightly stronger metric convention of \cref{def: geometrically local}, which places the qubits on a fixed-dimensional lattice and directly requires bounded degree and bounded interaction range.
\end{rem}

Oliveira and Terhal~\cite{Oliveira2005} proved that the $5$-Local Hamiltonian problem remains QMA-complete for spatially sparse interaction hypergraphs. 
Their explicit two-dimensional layouts satisfy \cref{def: geometrically local} after a constant rescaling of the lattice, giving the following source problem.

\begin{cor}
\label{thm:geo-local-qma}
The $5$-Local Hamiltonian problem remains QMA-complete when restricted to interaction hypergraphs satisfying \cref{def: geometrically local}.
\end{cor}

For the geometric reductions below, we use the standard two-dimensional column-and-SWAP realization underlying the Oliveira--Terhal construction.
The verifier is first compiled to a one-dimensional nearest-neighbor circuit and then spatialized across a two-dimensional array. 
Successive computational rounds occupy adjacent columns, with nearest-neighbor SWAPs transporting the logical state between columns. The unary clock is placed along the same snaking computation path.

After padding the circuit by initial and final identity rounds, every physical work and clock qubit participates in only $O(1)$ Hamiltonian terms, and every local term has constant geometric diameter. 
For each initialized work qubit $W_i$, let $t_i$ denote the first time at which it participates in a nonidentity operation. 
Since all earlier operations on $W_i$ are identities, its initialization condition may be checked locally immediately before time $t_i$ using
\[
\Pi^{\Ket{1}}_{C_{t_i-1}}
\otimes
\Pi^{\Ket{0}}_{C_{t_i}}.
\]

We denote the resulting geometrically local source Hamiltonian by
\[
H_{\mathrm{geo}}
=
H^{\mathrm{geo}}_{\mathrm{in}}
+
H^{\mathrm{geo}}_{\mathrm{out}}
+
H^{\mathrm{geo}}_{\mathrm{prop}}
+
H^{\mathrm{geo}}_{\mathrm{clock}}.
\]
It is $5$-local, has bounded interaction range and bounded degree, and retains the standard inverse-polynomial completeness-soundness separation. 
The explicit spatial layout, localized initialization terms, and the corresponding energy bounds are recorded in \cref{app:geometric-source}.

\subsubsection{The Stoquastic Circuit-to-Hamiltonian Construction} \label{subsection:stoq construction}
The standard circuit-to-Hamiltonian construction of \cref{subsubsec: Circuit-to-Hamiltonian} is not generally stoquastic, because the propagation terms inherit the signs and phases of the verifier gates. 
Bravyi et al.~\cite{Bravyi2006} showed that restricting the verifier to classical reversible gates yields a stoquastic circuit-to-Hamiltonian construction.

\begin{thm}[\cite{Bravyi2006}]
The Stoquastic $6$-Local Hamiltonian problem is StoqMA-complete.
\end{thm}

\paragraph{Construction.}
Let $L\in\mathrm{StoqMA}$ and let $U_x=U_T\cdots U_1$ be a verifier as in \cref{def: StoqMA}. 
The circuit acts on $n_0+n_1+n_++m$ qubits initialized as
$
\Ket{0}^{\otimes n_0}
\otimes
\Ket{1}^{\otimes n_1}
\otimes
\Ket{+}^{\otimes n_+}
\otimes
\Ket{\xi},
$
where the gates $U_t$ are classical reversible gates, hence permutation matrices in the computational basis.

The Hamiltonian is:
\begin{align}
    H^{\mathrm{stoq}} = H_{\mathrm{in}}^{\mathrm{stoq}} + H_{\mathrm{out}}^{\mathrm{stoq}} + H_{\mathrm{prop}}^{\mathrm{stoq}} + H_{\mathrm{clock}}^{\mathrm{stoq}}
\end{align}

These terms enforce the correct execution of the verifier circuit as follows:
\paragraph{Clock.}
The clock Hamiltonian is the same unary-clock penalty as in \cref{eq:hclock}. It is diagonal in the computational basis and is therefore stoquastic.

\paragraph{Propagation.}
The propagation terms have the same form as in \cref{eq:hprop}.
Because each $U_t$ is a permutation matrix, the off-diagonal entries contributed by $-\frac12 U_t$ are non-positive. Hence every propagation term is stoquastic.

\paragraph{Initialization.}
In addition to the computational-basis initialization checks, the $\Ket{+}$ ancillas are enforced by projectors onto $\Ket{-}$:
\begin{align}
    H_{\mathrm{in}}^{\mathrm{stoq}} = \sum_{i=1}^{n_0} \Pi^{\Ket{1}}_{W_i} \otimes \Pi^{\Ket{0}}_{C_1} + \sum_{i=n_0 + 1}^{n_0+n_1} \Pi^{\Ket{0}}_{W_i} \otimes \Pi^{\Ket{0}}_{C_1} + \sum_{i=n_0 +n_1 + 1}^{n_0+n_1+n_+} \Pi^{\Ket{-}}_{W_i} \otimes \Pi^{\Ket{0}}_{C_1} 
\end{align}

\paragraph{Output.}
Since general StoqMA amplification is not available in the form used for QMA, the output penalty is treated as a small perturbation. 
We take $\delta=o(T^{-3})$ and define
\begin{align}
    H_{\mathrm{out}}^{\mathrm{stoq}} = \delta \cdot \left(\Pi^{\Ket{-}}_{W_1} \otimes \Pi^{\Ket{1}}_{C_T} \right)
\end{align}

\paragraph{Stoquastic history state and energy bounds.}
For a fixed witness $\Ket{\xi}$, the corresponding propagation-consistent history state is
\begin{align}
\Ket{\eta_{\mathrm{stoq}}}
=
\frac{1}{\sqrt{T+1}}
\sum_{t=0}^{T}
U_t\cdots U_1
\Ket{0}^{\otimes n_0}
\otimes
\Ket{1}^{\otimes n_1}
\otimes
\Ket{+}^{\otimes n_+}
\otimes
\Ket{\xi}
\otimes
\Ket{t}_C .
\end{align}

To complete the reduction, we establish the energy bounds for the Stoquastic Hamiltonian.

\begin{lemma}[Completeness, \cite{Bravyi2006}]
\label{lemma: bravyi StoqMA completness}
If $x\in L_{\mathrm{yes}}$, then there exists a witness $\Ket{\xi}$ whose corresponding history state satisfies
\[
\Bra{\eta_{\mathrm{stoq}}}
H^{\mathrm{stoq}}
\Ket{\eta_{\mathrm{stoq}}}
\leq
\frac{\delta}{T+1}
\left(1-\epsilon_{\mathrm{yes}}\right).
\]
\end{lemma}

\begin{lemma}[Soundness, adapted from~\cite{Bravyi2006}]
\label{lem:stoq-soundness}
If $x\in L_{\mathrm{no}}$, then

\[
\lambda_{\min}(H^{\mathrm{stoq}})
=
\frac{\delta}{T+1}
\left(
1-p_{\max}
\right)
+
O(\delta^2),
\]

where $p_{\max}\leq\epsilon_{\mathrm{no}}$ is the maximum acceptance probability over all witnesses. 
Since
\[
\epsilon_{\mathrm{yes}}-\epsilon_{\mathrm{no}}
\geq
\frac{1}{\operatorname{poly}(|x|)},
\]
and the unperturbed history-state Hamiltonian has an inverse-polynomial spectral gap, $\delta$ may be chosen sufficiently small that the $O(\delta^2)$ correction is asymptotically smaller than the first-order YES-NO energy separation. 
This yields an inverse-polynomial promise gap.

We do not require an explicit bound on the remainder here. 
In \cref{sec:stoq-dsih-com-sound}, where a quantitative lower bound is needed, we derive one directly using the Projection Lemma.

\end{lemma}

\subsubsection{A Geometrically Local Stoquastic Source Hamiltonian}
\label{subsubsec: Stoq Geometrically Local}

Waite and Bremner~\cite{Waite2025} extend the spatially sparse circuit-to-Hamiltonian construction of Oliveira and Terhal to StoqMA verification circuits and prove that the $6$-Local Stoquastic Hamiltonian problem remains StoqMA-complete on spatially sparse interaction hypergraphs. 
Their explicit construction uses the same two-dimensional column-and-SWAP architecture described in \cref{subsubsec: Geometrically Local}. 
After a constant rescaling of the lattice, every interaction has constant geometric range and every physical qubit participates in only a constant number of terms.
Consequently, their construction also satisfies \cref{def: geometrically local}.

\begin{cor}
\label{thm:geo-stoq}
The $6$-Local Stoquastic Hamiltonian problem remains StoqMA-complete when restricted to the geometrically local interaction hypergraphs of \cref{def: geometrically local}.
\end{cor}

For the reduction in \cref{sec:stoq-n-dependent-hardness}, we use the explicit spatialized StoqMA construction underlying this result. 
It has the same two-dimensional column-and-SWAP architecture as the QMA source above, with the unary clock placed along the snaking computation path.

The initialized physical work qubits are partitioned into sets $I_0$, $I_1$, and $I_+$ according to whether their prescribed states are $\Ket{0}$, $\Ket{1}$, or $\Ket{+}$. 
As in the QMA construction, each initialized qubit is checked immediately before its first nonidentity operation. 
The resulting Hamiltonian has the form
\[
H^{\mathrm{stoq}_g}
=
H^{\mathrm{stoq}_g}_{\mathrm{in}}
+
H^{\mathrm{stoq}_g}_{\mathrm{prop}}
+
H^{\mathrm{stoq}_g}_{\mathrm{clock}}
+
\delta H^{\mathrm{stoq}_g}_{\mathrm{out}},
\]
is $6$-local and stoquastic, and has bounded interaction range and bounded degree.

Writing
\[
H_0
=
H^{\mathrm{stoq}_g}_{\mathrm{in}}
+
H^{\mathrm{stoq}_g}_{\mathrm{prop}}
+
H^{\mathrm{stoq}_g}_{\mathrm{clock}},
\]
the unperturbed history-state Hamiltonian has spectral gap
\[
\Gamma_T=\Omega(T^{-3}).
\]
Choosing $\delta>0$ sufficiently small relative to $\Gamma_T$ yields the inverse-polynomial YES-NO energy separation required below.
The explicit local terms and energy bounds are recorded in \cref{app:geometric-stoq-source}.

\subsubsection{Weighted \texorpdfstring{$S$}{S}-Hamiltonians}
\label{subsubsec:s-hamiltonians}

The $S$-Hamiltonian framework restricts the local interactions of a Hamiltonian to a fixed finite set while allowing independently specified coupling strengths~\cite{Cubitt2013}. 
We use the following explicit weighted formulation.

Let
\[
S=\{s^{(1)},\ldots,s^{(m)}\}
\]
be a fixed finite set of $k$-local Hermitian interactions. 
A weighted $S$-Hamiltonian has the form
\[
H
=
\sum_{i=1}^{M}
\alpha_i s^{(j_i)}_{V_i},
\]
where $M=\operatorname{poly}(n)$, the coefficients $\alpha_i\in\mathbb R$ are polynomially bounded and specified with polynomially many bits, and $V_i$ specifies the qubits on which the interaction is applied.

We write $S^{+}$-Hamiltonian for the restriction in which all coupling strengths are nonnegative,
\[
\alpha_i\geq 0.
\]
When the interaction set is a singleton $S=\{h\}$, we write $\{h\}$-Hamiltonian and $\{h\}^{+}$-Hamiltonian for the unrestricted and positive-weight problems, respectively.

For the positive-weight singleton problems used later, we require the following classification of Piddock and Montanaro.

\begin{thm}[\cite{Piddock2015}]
\label{thm:piddock-montanaro-positive}
Let
\[
h
=
\alpha X\otimes X
+
\beta Y\otimes Y
+
\gamma Z\otimes Z.
\]
Then:
\begin{enumerate}
    \item If
    \[
    \alpha+\beta>0,
    \qquad
    \alpha+\gamma>0,
    \qquad
    \beta+\gamma>0,
    \]
    then $\{h\}^{+}$-Hamiltonian is QMA-complete.

    \item If
    \[
    \alpha=-\beta\neq0,
    \qquad
    \alpha+\gamma>0,
    \qquad
    \beta+\gamma>0,
    \]
    then $\{h\}^{+}$-Hamiltonian is StoqMA-complete.
\end{enumerate}
\end{thm}

\subsubsection{The XY Model} \label{subsec: XY model}

Childs, Gosset, and Webb~\cite{Childs2015} study the antiferromagnetic XY model on a simple graph at fixed magnetization. Unlike general Local Hamiltonian, every edge carries the same two-qubit interaction; the computational problem is defined within a prescribed Hamming-weight sector, with the sector supplied as part of the instance.

The Hamiltonian is defined on a simple graph $G=(V,E)$, by applying the $XY$ interaction on its edges:
\begin{align}
    H_{XY} = \frac{1}{2} \sum_{(u,v) \in E} (X_u X_v + Y_u Y_v)
\end{align}
where $X$ and $Y$ are the Pauli matrices. Equivalently, in the
computational basis,
\begin{align}
    H_{XY} = \sum_{(u,v) \in E} \left( \Ket{01}\Bra{10}_{u,v} + \Ket{10}\Bra{01}_{u,v} \right)
\end{align}
The XY Hamiltonian preserves Hamming weight. Writing
\begin{align}
    \hat{N}
    =
    \sum_{v\in V}\frac{I-Z_v}{2}
    =
    \sum_{v\in V}\Pi^{\Ket{1}}_v,
    \qquad
    [H_{XY},\hat{N}]=0,
\end{align}
the Hilbert space decomposes into invariant subspaces $\mathcal{S}_k$ of fixed Hamming weight,
\begin{align}
    \mathcal{S}_k = \text{span} \left\{ \Ket{z} : z \in \{0,1\}^{|V|}, \text{wt}(z) = k \right\}
\end{align}
where the total magnetization corresponds to $M=|V|-2k$.

\begin{defn}[XY Hamiltonian Problem at Fixed Hamming Weight, adapted from \cite{Childs2015}] \label{def: XY Hamiltonian problem} \mbox{}\\
    \noindent \textbf{Problem Input:} A simple graph $G=(V,E)$, an integer $k \leq |V|$ specifying the Hamming weight, and two real numbers $a, b$, specified with $\operatorname{poly}(|V|)$ precision, such that $b - a \geq 1/\operatorname{poly}(|V|)$.
    
    \noindent \textbf{Promise:} Let $H_{XY}$ be the Hamiltonian resulting from applying the XY interaction to all edges of $G$. Restricting $H_{XY}$ to the subspace $\mathcal{S}_k$, one of the following holds:
    \begin{itemize}
        \item There exists a quantum state $\Ket{\psi} \in \mathcal{S}_k$ such that $\Bra{\psi} H_{XY} \Ket{\psi} \leq a$.
        \item For all quantum states $\Ket{\psi} \in \mathcal{S}_k$, $\Bra{\psi} H_{XY} \Ket{\psi} \geq b$.
    \end{itemize}
    
    \noindent \textbf{Output:} Determine which of the two cases holds.
\end{defn}

\begin{thm}[\cite{Childs2015}] \label{thm: XY hamiltonian QMA complete}
The XY Hamiltonian problem on a simple graph at fixed Hamming weight is QMA-complete.
\end{thm}

\subsection{Embedding a Complex Hamiltonian into a Stoquastic One} \label{sec: Embedding Hamiltonian}

We use a two-step embedding to convert a general complex Hamiltonian into a stoquastic Hamiltonian on an enlarged Hilbert space. 
First, realification replaces complex matrix entries by a real representation using one auxiliary qubit. 
A second auxiliary qubit is then used to replace positive off-diagonal entries while recovering the original interaction in a pinned $\Ket{-}$ sector. 
Each step increases locality by one.

\subsubsection{Realification of Complex Hamiltonians}

\begin{defn}[Realification Map, adapted from
\cite{Rudolph2002, McKague2010, McKague2011}]
\label{def: real to complex iso}
Let $\mathcal{H}$ be a complex Hilbert space with a fixed orthonormal basis, and let $\mathcal{L}$ be a two-dimensional auxiliary space with basis $\{\Ket{0},\Ket{1}\}$.

For a vector $\Ket{\psi}\in\mathcal{H}$, define
\begin{align}
\Phi(\Ket{\psi})
=
\operatorname{Re}(\Ket{\psi})\otimes\Ket{0}_{\mathcal L}
+
\operatorname{Im}(\Ket{\psi})\otimes\Ket{1}_{\mathcal L},
\end{align}
where the real and imaginary parts are taken componentwise in the fixed basis, and we set

\[
\Phi(\Bra{\psi})
=
\Phi(\Ket{\psi})^\dagger .
\]

Define

\[
J_{\mathcal L}
=
\Ket{1}\!\Bra{0}_{\mathcal L}
-
\Ket{0}\!\Bra{1}_{\mathcal L}.
\]

For a linear operator $M$ on $\mathcal{H}$, define
\begin{align}
\Phi(M)
=
\operatorname{Re}(M)\otimes I_{\mathcal L}
+
\operatorname{Im}(M)\otimes J_{\mathcal L}.
\end{align}
The operator map is an injective real-algebra homomorphism.
\end{defn}

The following identities will be used repeatedly.

\begin{lemma}[\cite{McKague2011}]
\label{lemma: real to complex iso prop}
Let $M,N$ be linear operators on a Hilbert space $\mathcal H$ and
$\Ket{\varphi}, \Ket{\psi} \in \mathcal{H}$. Then:
\begin{align}
    \Phi\left(MN\right)
    &= \Phi\left(M\right)\Phi\left(N\right), \\
    \Phi\left(M\Ket{\psi}\right)
    &= \Phi\left(M\right)\Phi\left(\Ket{\psi}\right), \\
    \Phi\left(\Bra{\psi}\right)\Phi\left(M\right)\Phi\left(\Ket{\varphi}\right)
    &= \operatorname{Re}\left(\Bra{\psi}M\Ket{\varphi}\right).
    \label{eq:realification-matrix-element}
\end{align}
\end{lemma}

\begin{lemma}[Spectrum under realification]
\label{lem:realification-spectrum}
Let $H$ be a Hermitian operator on an $N$-dimensional complex Hilbert space. 
Then $\Phi(H)$ is a real symmetric $2N\times 2N$ matrix. 
When regarded as a Hermitian operator on the complex Hilbert space $\mathbb{C}^{2N}$, its spectrum is the spectrum of $H$ with every eigenvalue appearing with twice its original multiplicity.
In particular,
\[
\lambda_{\min}(\Phi(H))
=
\lambda_{\min}(H).
\]
\end{lemma}

\begin{proof}
Write an eigenvector of $H$ as
\[
\Ket{\psi}
=
\Ket{x}+i\Ket{y},
\]
where $\Ket{x}$ and $\Ket{y}$ are real vectors, and suppose
\[
H\Ket{\psi}
=
\lambda\Ket{\psi}.
\]
Since $H$ is Hermitian, $\lambda\in\mathbb{R}$.

By \cref{lemma: real to complex iso prop},
\[
\Phi(H)\Phi(\Ket{\psi})
=
\Phi(H\Ket{\psi})
=
\lambda\Phi(\Ket{\psi}).
\]
Thus
\[
\Ket{x}\otimes\Ket{0}
+
\Ket{y}\otimes\Ket{1}
\]
is a real eigenvector of $\Phi(H)$ with eigenvalue $\lambda$.

The vector corresponding to $i\Ket{\psi}$ is
\[
\Phi(i\Ket{\psi})
=
-\Ket{y}\otimes\Ket{0}
+
\Ket{x}\otimes\Ket{1},
\]
which is linearly independent over $\mathbb{R}$ from $\Phi(\Ket{\psi})$ whenever $\Ket{\psi}\neq0$, and has the same eigenvalue.

Hence every $d$-dimensional complex eigenspace of $H$ gives a $2d$-dimensional real eigenspace of $\Phi(H)$. Since the total real dimension is $2N$, these eigenspaces exhaust the space. 
Viewing the same real symmetric matrix as an operator over $\mathbb{C}^{2N}$ does not change its eigenvalues or their multiplicities. 
Therefore every eigenvalue of $H$ occurs twice in $\Phi(H)$.
\end{proof}

\subsubsection{Stoquastic Sign Embedding}

We next convert a real Hamiltonian into a stoquastic one using a single auxiliary sign qubit. The underlying construction originates with Janzing and Wocjan~\cite{Janzing2006}, was applied to stoquastic Hamiltonians by Jordan, Gosset, and Love~\cite{Jordan2009}, and is used here in the pinned formulation of Nagaj et al.~\cite{Nagaj2020}.

Let

\[
R=\Phi(h_A)
\]

be the realification of a Hermitian operator $h_A$. Decompose

\[
R=R^-+R^+,
\]

where $R^-$ contains the diagonal entries and all non-positive off-diagonal entries of $R$, with the strictly positive off-diagonal entries set to zero, while $R^+$ contains precisely the strictly positive off-diagonal entries of $R$.

Introduce an auxiliary sign qubit $c$ and define
\begin{equation}
\label{eq: make stoquastic}
\mathcal{S}(R)
=
R^-\otimes I_c-
R^+\otimes X_c.
\end{equation}
The operator $\mathcal{S}(R)$ is stoquastic in the computational basis.
Indeed, the off-diagonal entries contributed by $R^-\otimes I_c$ are non-positive, and those contributed by $-R^+\otimes X_c$ are likewise non-positive.

The subspace in which the sign qubit is fixed to
$
\Ket{-}_c
=
\frac{\Ket{0}-\Ket{1}}{\sqrt{2}}
$
is invariant under $\mathcal{S}(R)$. 
Since

\[
X_c\Ket{-}_c=-\Ket{-}_c,
\]

for every state $\Ket{\psi}$ on the original system together with the phase qubit,
\begin{align}
\mathcal{S}(R)
\left(
\Ket{\psi}\otimes\Ket{-}_c
\right)
&=
\left(
R^-\otimes I_c
-
R^+\otimes X_c
\right)
\left(
\Ket{\psi}\otimes\Ket{-}_c
\right)
\nonumber\\
&=
\left(
R^-+R^+
\right)
\Ket{\psi}
\otimes
\Ket{-}_c
\nonumber\\
&=
R\Ket{\psi}
\otimes
\Ket{-}_c.
\end{align}

Equivalently, if

\[
V_-:\Ket{\psi}\longmapsto
\Ket{\psi}\otimes\Ket{-}_c
\]

denotes the corresponding isometry, then
\begin{equation}
\label{eq:stoq-sign-restriction}
V_-^\dagger
\mathcal{S}(R)
V_-
=
R.
\end{equation}

Thus the restriction of $\mathcal{S}(R)$ to the pinned $\Ket{-}_c$ sector coincides with $R$. 
In particular,

\[
V_-^\dagger
\mathcal{S}\!\left(\Phi(h_A)\right)
V_-
=
\Phi(h_A).
\]

By \cref{lem:realification-spectrum}, this restricted operator has the same spectrum as $h_A$, with every eigenvalue appearing twice.

\section{A Framework for Single-Interaction Hamiltonians}
\label{sec:sih-framework}

We now formalize the single-interaction framework. 
In every instance, each hyperedge carries an unweighted copy of the same $k$-qubit interaction. 
In $\mathrm{SIH}_k$, this interaction is fixed independently of the instance and the system size; in $\mathrm{DSIH}_k$, it belongs to a uniformly computable family $\{h^{(n)}\}_{n\geq1}$ indexed only by the final physical system size.

\subsection{The SIH and DSIH Framework}
\label{subsec:sih-dsih-framework}

\begin{defn}[Single-Interaction $k$-Local Hamiltonian, denoted $\mathrm{SIH}_k$] \label{def: SIH} Fix a $k$-qubit Hermitian matrix
\[
h \in \mathrm{Herm}\!\left((\mathbb{C}^2)^{\otimes k}\right).
\]
The interaction $h$ is fixed as part of the definition of the promise problem and is independent of both the input instance and the system size $n$. 
We require the entries of $h$ to be polynomial-time computable. 
More precisely, there exists a deterministic algorithm which, given a precision parameter $q$, outputs rational approximations to all entries of $h$ with additive error at most $2^{-q}$ in time polynomial in $q$.

\textbf{Problem Input:}
The system size $n$, encoded in unary, a set $L$ of $M=\operatorname{poly}(n)$ pairwise distinct ordered $k$-tuples
\[
e=(v_1,\ldots,v_k)
\]

of physical qubits, where $v_1,\ldots,v_k$ are pairwise distinct within each tuple, and two real numbers $a,b$ specified with $\operatorname{poly}(n)$ bits of precision such that

\[
b-a\geq\frac{1}{\operatorname{poly}(n)}.
\]

\textbf{Promise:}
For each $e\in L$, let $h_e$ denote the application of $h$ to the ordered tuple $e$, acting as the identity on all remaining qubits. Define
\[
H=\sum_{e\in L} h_e.
\]
Exactly one of the following holds:
\begin{itemize}
    \item There exists an $n$-qubit state $\Ket{\phi}$ such that
    \[
    \bra{\phi}H\Ket{\phi}\leq a.
    \]
    \item For every $n$-qubit state $\Ket{\phi}$,
    \[
    \bra{\phi}H\Ket{\phi}\geq b.
    \]
\end{itemize}

\textbf{Output:}
Determine which of the two cases holds.
\end{defn}

\begin{defn}[$n$-Dependent Single-Interaction $k$-Local Hamiltonian,
denoted $\mathrm{DSIH}_k$]
\label{def: DSIH}

The problem $\mathrm{DSIH}_k$ is defined as in \cref{def: SIH}, except that the fixed interaction $h$ is replaced by a uniform family of $k$-qubit Hermitian matrices
\[
\{h^{(n)}\}_{n\geq1}.
\]

For each system size $n$, the Hamiltonian associated with an interaction set $L$ is
\[
H^{(n)}
=
\sum_{e\in L} h^{(n)}_e.
\]

The family $\{h^{(n)}\}_{n\geq1}$ is fixed as part of the definition of the promise problem.  For an instance of final physical size $n$, the common interaction is $h^{(n)}$; no other instance-specific parameter may enter the interaction matrix. 
We require a deterministic algorithm which, given $1^n$ and a precision parameter $q$, outputs rational approximations to all entries of $h^{(n)}$ with additive error at most $2^{-q}$ in time polynomial in $n$ and $q$. 
We additionally require
\[
\|h^{(n)}\|
\leq
\operatorname{poly}(n).
\]

All remaining input and promise conditions are those of \cref{def: SIH}, including the pairwise-distinct ordered interaction tuples and the inverse-polynomial promise gap.
\end{defn}

\paragraph{Final-system-size convention.}
Throughout the paper, $n$ in every $n$-dependent single-interaction problem denotes the total number of physical qubits in the \emph{final} Hamiltonian instance, including all auxiliary and pinned qubits. 
Thus $h^{(n)}$ may depend only on this final system size and not directly on any other instance-specific quantity.

\subsection{Pinned and Stoquastic Variants}

We next define stoquastic and pinned restrictions of the SIH and DSIH frameworks. 
Stoquastic variants require the common interaction to be stoquastic in the computational basis.

\begin{defn}[Stoquastic Single-Interaction $k$-Local Hamiltonian, denoted $\mathrm{Stoq\text{-}SIH}_k$]
The $\mathrm{Stoq\text{-}SIH}_k$ problem is the restriction of $\mathrm{SIH}_k$ to the case in which the fixed interaction $h$ is stoquastic in the computational basis; that is, $h$ is real and
\[
\bra{x}h\Ket{y} \leq 0
\]
for all distinct computational-basis states $\Ket{x}$ and $\Ket{y}$.
\end{defn}

\begin{defn}[Stoquastic $n$-Dependent Single-Interaction $k$-Local Hamiltonian, denoted $\mathrm{Stoq\text{-}DSIH}_k$]
The $\mathrm{Stoq\text{-}DSIH}_k$ problem is the restriction of $\mathrm{DSIH}_k$ to uniform families $\{h^{(n)}\}_{n\geq 1}$ for which every $h^{(n)}$ is stoquastic in the computational basis.
\end{defn}

Pinned variants restrict the energy minimization to states with a prescribed state on a constant-size register.

\begin{defn}[$p$-Pinned Single-Interaction $k$-Local Hamiltonian, denoted $p$-$\mathrm{Pinned\text{-}SIH}_k$]
\label{def:pinned-sih}

Let $p=O(1)$ be fixed. The problem is defined as $\mathrm{SIH}_k$ with the following additional input.

\textbf{Additional Input:}
An ordered set
\[
A=(a_1,\ldots,a_p)
\]
of $p$ distinct physical qubits, called the pinned register, together with a classical description of a polynomial-size quantum circuit $C_{\mathrm{pin}}$ over a fixed finite gate set acting on $p$ qubits.
The pinned state is
\[
\Ket{\psi_{\mathrm{pin}}}_A
=
C_{\mathrm{pin}}\Ket{0^p}_A.
\]

Let $W$ denote the remaining $n-p$ qubits and define the isometry
\[
V_{\mathrm{pin}}
:
\mathcal H_W
\longrightarrow
\mathcal H_W\otimes\mathcal H_A,
\qquad
V_{\mathrm{pin}}\Ket{\phi}
=
\Ket{\phi}_W\otimes\Ket{\psi_{\mathrm{pin}}}_A.
\]
For the SIH Hamiltonian $H$, define its compression to the pinned sector by
\[
H_{\mathrm{pin}}
=
V_{\mathrm{pin}}^\dagger H V_{\mathrm{pin}}.
\]

The promise is that exactly one of
\[
\lambda_{\min}(H_{\mathrm{pin}})\leq a
\qquad\text{or}\qquad
\lambda_{\min}(H_{\mathrm{pin}})\geq b
\]
holds. The task is to determine which case holds. All other input and
promise requirements are those of \cref{def: SIH}.
\end{defn}

\begin{defn}[$p$-Pinned $n$-Dependent Single-Interaction
$k$-Local Hamiltonian, denoted
$p$-$\mathrm{Pinned\text{-}DSIH}_k$]
\label{def:pinned-dsih}

The problem is defined identically to
\cref{def:pinned-sih}, with $\mathrm{SIH}_k$ replaced by
$\mathrm{DSIH}_k$. In particular, for
\[
H^{(n)}
=
\sum_{e\in L}h^{(n)}_e,
\]
the relevant compressed Hamiltonian is
\[
H_{\mathrm{pin}}^{(n)}
=
V_{\mathrm{pin}}^\dagger
H^{(n)}
V_{\mathrm{pin}},
\]
and the promise is that exactly one of
\[
\lambda_{\min}\!\left(H_{\mathrm{pin}}^{(n)}\right)\leq a
\qquad\text{or}\qquad
\lambda_{\min}\!\left(H_{\mathrm{pin}}^{(n)}\right)\geq b
\]
holds.

\end{defn}

Finally, we combine the two restrictions to obtain the pinned stoquastic
variants.

\begin{defn}[$p$-Pinned Stoquastic Single-Interaction $k$-Local
Hamiltonian, denoted $p$-$\mathrm{Pinned\text{-}Stoq\text{-}SIH}_k$]
The $p$-$\mathrm{Pinned\text{-}Stoq\text{-}SIH}_k$ problem is the
restriction of $p$-$\mathrm{Pinned\text{-}SIH}_k$ to the case in which
the fixed interaction matrix $h$ is stoquastic in the computational basis.
\end{defn}

\begin{defn}[$p$-Pinned Stoquastic $\mathrm{DSIH}_k$, denoted
$p$-$\mathrm{Pinned\text{-}Stoq\text{-}DSIH}_k$]
This problem is the restriction of the pinned DSIH problem of
\cref{def:pinned-dsih} to uniform families
$\{h^{(n)}\}_{n\geq1}$ for which every $h^{(n)}$ is stoquastic in
the computational basis.
\end{defn}

\subsection{Additional Promise Variants}
\label{subsec:additional-sih-variants}

We also use several variants that retain the same interaction
uniformity while modifying the promise or supplying additional
low-energy information.  We define only those needed later.

\begin{defn}[Frustration-Free Single-Interaction Hamiltonian,
denoted $\mathrm{FF\text{-}SIH}_k$]
\label{def:ff-sih}

The problem $\mathrm{FF\text{-}SIH}_k$ is defined as in
\cref{def: SIH} for a fixed positive-semidefinite interaction
\[
h\succeq0,
\]
except that the lower threshold is fixed to $a=0$ and need not be
supplied as part of the input. The remaining threshold satisfies
\[
b\geq\frac{1}{\operatorname{poly}(n)}.
\]
The promise is that exactly one of
\[
\lambda_{\min}(H)=0
\qquad\text{or}\qquad
\lambda_{\min}(H)\geq b
\]
holds. The task is to distinguish these two cases.
\end{defn}

Since every local term is positive semidefinite, the YES condition is
equivalent to the existence of a common zero-energy state,
\[
h_e\Ket{\psi}=0
\qquad
\text{for every }e\in L.
\]
Thus the YES instances are precisely the frustration-free instances.

\begin{defn}[Precise Single-Interaction Hamiltonian, denoted $\mathrm{Precise\text{-}SIH}_k$]
\label{def:precise-sih}
The problem $\mathrm{Precise\text{-}SIH}_k$ is defined identically to
$\mathrm{SIH}_k$, except that the inverse-polynomial promise-gap
condition is replaced by
\[
b-a \ge 2^{-p(n)}
\]
for some fixed polynomial $p$.

\end{defn}

\begin{defn}[Guided Single-Interaction $k$-Local Hamiltonian,
denoted $\mathrm{Guided\text{-}SIH}_k$]
\label{def:guided-sih}
Fix a $k$-qubit Hermitian interaction $h$ as in
\cref{def: SIH}.

\textbf{Problem Input:}
An $\mathrm{SIH}_k$ instance $(n,L,a,b)$, together with a classical description
of an $n$-qubit semi-classical state $\Ket{u}$ as in \cref{def:semiclassical-subset-state} and
a number $\delta\in(0,1]$ specified with $\operatorname{poly}(n)$ bits of precision.

We require
\[
b-a\ge\frac{1}{\operatorname{poly}(n)}
\qquad\text{and}\qquad
\delta\ge\frac{1}{\operatorname{poly}(n)}.
\]

Let
\[
H=\sum_{e\in L}h_e,
\]
and let $\Pi_H$ denote the orthogonal projector onto the ground space
of $H$.

\textbf{Promise:}
The guiding state satisfies
\[
\|\Pi_H\Ket{u}\|\ge\delta,
\]
and exactly one of the following holds:
\begin{itemize}
    \item[\textup{YES:}]
    $\lambda_{\min}(H)\le a$;
    \item[\textup{NO:}]
    $\lambda_{\min}(H)\ge b$.
\end{itemize}

\textbf{Output:}
Determine which of the two cases holds.
\end{defn}

\subsection{Basic Containment Results}
\label{subsec:basic-containment}

We conclude the framework section by recording the basic upper bounds
used later.  Ordinary SIH and DSIH variants are contained in QMA,
unpinned stoquastic variants in StoqMA, precise variants in PSPACE, and
Guided-SIH in BQP; pinned variants remain in QMA\@. For $n$-dependent
problems, these containments rely on the uniform computability
requirement of \cref{def: DSIH}.

\begin{lemma}[Containment of Unpinned Variants]
\label{lem:unpinned-containment}
For constant $k$,
\[
\mathrm{SIH}_k,\mathrm{DSIH}_k \in \mathrm{QMA}.
\]
\end{lemma}

\begin{proof}
By \cref{def: SIH}, the fixed interaction $h$ can be approximated to
inverse-exponential precision in polynomial time. Since an instance
contains only $\operatorname{poly}(n)$ applications of $h$, the total
operator-norm approximation error can be made smaller than any fixed
fraction of the inverse-polynomial promise gap.

For $\mathrm{DSIH}_k$, the same conclusion follows from the uniform
computability requirement of \cref{def: DSIH}, which allows
$h^{(n)}$ to be generated to the required precision in polynomial time.

Thus both problems reduce directly to efficiently specified instances of
the standard $k$-Local Hamiltonian problem and are contained in QMA.
\end{proof}

\begin{lemma}[Containment of Stoquastic Variants]
\label{lem:stoquastic-containment}
For constant $k$,
\[
\mathrm{Stoq\text{-}SIH}_k,
\mathrm{Stoq\text{-}DSIH}_k
\in
\mathrm{StoqMA}.
\]
\end{lemma}

\begin{proof}
By \cref{def: SIH,def: DSIH}, the relevant fixed or uniformly
$n$-dependent interaction can be computed to sufficient precision in
polynomial time. The approximation may be chosen stoquastic by replacing
any positive off-diagonal approximation by zero; since the exact
off-diagonal entries are non-positive, this changes the error by at most
the chosen approximation precision.

Hence every instance is an efficiently specified Stoquastic
$k$-Local Hamiltonian instance. Since Stoquastic Local Hamiltonian with
inverse-polynomial promise gap is contained in StoqMA
\cite{Bravyi2006}, the claim follows.
\end{proof}

\begin{lemma}[Containment of Pinned Variants]
\label{lem:pinned-containment}
For every fixed constant $p$ and constant $k$, the problems
\[
p\text{-}\mathrm{Pinned\text{-}SIH}_k,\qquad
p\text{-}\mathrm{Pinned\text{-}DSIH}_k,
\]
\[
p\text{-}\mathrm{Pinned\text{-}Stoq\text{-}SIH}_k,\qquad
p\text{-}\mathrm{Pinned\text{-}Stoq\text{-}DSIH}_k
\]
are contained in QMA.
\end{lemma}

\begin{proof}
Let $A$ be the constant-size pinned register and $W$ the remaining
qubits. A QMA verifier receives a witness on $W$, prepares the prescribed
state $\Ket{\psi_{\mathrm{pin}}}_A$ using its supplied preparation
circuit, and applies the standard Local Hamiltonian energy test to the
joint state
\[
\Ket{\phi}_W\otimes\Ket{\psi_{\mathrm{pin}}}_A.
\]
The YES and NO energy bounds are exactly those of the compressed
Hamiltonian defined in \cref{def:pinned-sih}. The same argument applies
to the $n$-dependent and stoquastic restrictions. Hence all four pinned
variants lie in QMA.
\end{proof}

\begin{lemma}[Containment of Frustration-Free SIH]
\label{lem:ff-sih-containment}
For constant $k$,
\[
\mathrm{FF\text{-}SIH}_k\in\mathrm{QMA}.
\]
\end{lemma}

\begin{proof}
An $\mathrm{FF\text{-}SIH}_k$ instance is a special case of an
$\mathrm{SIH}_k$ instance with positive-semidefinite fixed interaction
$h$, YES threshold $a=0$, and inverse-polynomial promise gap.
The claim therefore follows directly from
\cref{lem:unpinned-containment}.
\end{proof}

\begin{lemma}[Containment of Precise Variants]
\label{lem:precise-sih-containment}
For constant $k$,

\[
\mathrm{Precise\text{-}SIH}_k\in\mathrm{PSPACE}.
\]

\end{lemma}

\begin{proof}
Let

\[
H=\sum_{e\in L}h_e
\]

be a Precise-$\mathrm{SIH}_k$ instance with
$M:=|L|=\operatorname{poly}(n)$ and promise gap

\[
b-a\ge 2^{-p(n)}
\]

for some fixed polynomial $p$.

The only additional issue relative to standard Precise Local
Hamiltonian is that the entries of the fixed interaction $h$ are
efficiently computable rather than supplied explicitly with the
instance.
By \cref{def: SIH}, we may approximate $h$ in polynomial time to
operator-norm error at most

\[
\frac{b-a}{4M}.
\]

Indeed, this requires only $\operatorname{poly}(n)$ bits of precision.
The resulting Hamiltonian $\widetilde H$ satisfies
\[
\|H-\widetilde H\|
\le
\frac{b-a}{4}.
\]
Hence, by Weyl's inequality, setting
\[
\widetilde a:=a+\frac{b-a}{4},
\qquad
\widetilde b:=b-\frac{b-a}{4},
\]
preserves the YES and NO cases, while
\[
\widetilde b-\widetilde a
=
\frac{b-a}{2}
\ge
2^{-\operatorname{poly}(n)}.
\]

Thus Precise-$\mathrm{SIH}_k$ reduces to an efficiently specified
constant-locality Precise Local Hamiltonian instance, and hence lies in
$\mathrm{QMA}_{\exp}=\mathrm{PSPACE}$ by
\cref{thm:qma-exp-pspace}.

The $n$-dependent case follows identically from the uniform
computability requirement of \cref{def: DSIH}; pinning is handled as
in \cref{lem:pinned-containment}; and stoquasticity is merely an
additional restriction on the allowed instances.
\end{proof}

\begin{lemma}[Containment of Guided SIH]
\label{lem:guided-sih-containment}
For constant $k$,
\[
\mathrm{Guided\text{-}SIH}_k\in\mathrm{BQP}.
\]
\end{lemma}

\begin{proof}
Let
\[
H=\sum_{e\in L}h_e
\]
be a Guided-$\mathrm{SIH}_k$ instance. Since $h$ is fixed and
$|L|=\operatorname{poly}(n)$,
\[
\|H\|\le |L|\,\|h\|=\operatorname{poly}(n).
\]
Moreover, the semi-classical guiding state $\Ket{u}$ is efficiently
preparable~\cite{Gharibian2021}.

After rescaling $H$ by a known polynomial bound on its norm, standard
Hamiltonian simulation and phase estimation can estimate its energy to
precision smaller than a constant fraction of the inverse-polynomial
promise gap. The guiding-state promise
\[
\|\Pi_H\Ket{u}\|\geq\delta
\]
gives ground-space measurement probability at least $\delta^2$.
Since $\delta\geq 1/\operatorname{poly(n)}$, polynomially many repetitions suffice to
distinguish the YES and NO cases with bounded error.

Hence
\[
\mathrm{Guided\text{-}SIH}_k\in\mathrm{BQP}.
\]
\end{proof}

\section{General Single-Interaction Compilers}
\label{sec:general_methodology}

We develop the single-interaction compilation framework in two stages.
First, \cref{sec:pinned-compiler,sec:unpinned-compiler} give exact
pinned and unpinned compilers for Hamiltonians whose local terms are
drawn from a fixed finite interaction alphabet.  We then treat weighted
and arbitrary local Hamiltonians separately in
\cref{sec:universal-finite-alphabet-normalization}, where coupling
strengths and arbitrary local matrices are reduced to a fixed finite
alphabet before the exact compiler is applied.  Thus the
finite-alphabet compilation itself is exact; approximation enters only
in the preceding normalization step.

Let
\[
\mathcal{A}
=
\{h^{(1)},h^{(2)},\ldots,h^{(m)}\}
\]
be a finite collection of $k$-local Hermitian matrices, and consider
\[
K
=
\sum_{i=1}^{M} h^{(j_i)}_{V_i},
\qquad
j_i\in\{1,\ldots,m\},
\]
where $V_i$ is an ordered $k$-tuple of pairwise distinct working qubits.
Interactions that differ by a scalar coefficient are regarded as
distinct elements of $\mathcal A$.

For the constant-locality constructions considered here,
$\mathcal A$ is fixed independently of the target instance and
$m=O(1)$.  The compiled interaction may depend on $\mathcal A$, but not
on the particular target Hamiltonian.  The pinned compiler requires the
pairs $(V_i,j_i)$ to be pairwise distinct: two different alphabet
elements may act on the same working tuple, but the same interaction
may not occur twice on that tuple.  The unpinned compiler removes this
restriction by assigning a dedicated auxiliary block to every
occurrence.

Both constructions produce a Hamiltonian $\widetilde K$ consisting
entirely of applications of one interaction $\widetilde h$.  Auxiliary
qubits act as selectors for the members of $\mathcal A$.  In the pinned
compiler their selector states are fixed by the pinning resource,
whereas in the unpinned compiler the valid selector configuration is
enforced energetically.  The specialized constructions in later
sections use the same selector principle while exploiting additional
structure to reduce the locality overhead.

\subsection{The Pinned Single-Interaction Reduction}
\label{sec:pinned-compiler}

The pinned compiler uses a constant-size pinned auxiliary register to
select among the interactions in $\mathcal A$.  Each
$h^{(j)}\in\mathcal A$ is assigned a tag, and the common interaction is
constructed so that compression to the corresponding pinned selector
state yields exactly $h^{(j)}$ on the working register.

We first give a binary selector and then a base-$3$ Pauli selector,
which reduces the selector locality at the cost of additional pinned
resource qubits per selector digit.

\subsubsection{The Binary Selector}

For each interaction $h^{(j)}\in\mathcal A$, choose a distinct binary
string
\[
x^{(j)}
=
\left(
x^{(j)}_1,\ldots,x^{(j)}_\ell
\right),
\qquad
\ell=\left\lceil\log_2m\right\rceil.
\]
Using the computational-basis projectors
\[
\Pi^0:=\Pi^{\Ket{0}},
\qquad
\Pi^1:=\Pi^{\Ket{1}},
\]
define the $(k+\ell)$-local interaction
\[
\widetilde h_{\mathrm B}
(v_1,\ldots,v_k,u_1,\ldots,u_\ell)
=
\sum_{j=1}^{m}
h^{(j)}_{v_1,\ldots,v_k}
\otimes
\Pi^{x^{(j)}_1}_{u_1}
\otimes\cdots\otimes
\Pi^{x^{(j)}_\ell}_{u_\ell}.
\]

The pinned auxiliary register contains $\ell$ pairs, each fixed to
$\Ket{01}$. At each selector position, routing the first or second qubit of the
corresponding pair supplies $\Ket{0}$ or $\Ket{1}$, respectively. Hence any tag $\Ket{x^{(j)}}$ can be routed to the selector
slots using the same $2\ell$ pinned qubits.

By orthogonality of the computational-basis projectors, compression to
that tag selects exactly the desired interaction:
\[
\left(
I\otimes\Bra{x^{(j)}}
\right)
\widetilde h_{\mathrm B}
\left(
I\otimes\Ket{x^{(j)}}
\right)
=
h^{(j)}.
\]

\subsubsection{The Base-3 Pauli Selector}

The binary construction uses
$\lceil\log_2m\rceil$ selector qubits.  A base-$3$ construction reduces
the selector length to $\lceil\log_3m\rceil$ by using the three Pauli
operators and their positive eigenstates, at the cost of three pinned
resource qubits per selector digit.

Define
\begin{align}
    \Ket{s_1}
    &= \Ket{+}
    =\frac{1}{\sqrt{2}}(\Ket{0}+\Ket{1}),
    &
    P^{(1)}&=X,
    \nonumber\\
    \Ket{s_2}
    &= \Ket{+i}
    =\frac{1}{\sqrt{2}}(\Ket{0}+i\Ket{1}),
    &
    P^{(2)}&=Y,
    \nonumber\\
    \Ket{s_3}
    &= \Ket{0},
    &
    P^{(3)}&=Z.
\end{align}
These states satisfy
\[
\Bra{s_q}P^{(r)}\Ket{s_q}
=
\delta_{qr},
\qquad
q,r\in\{1,2,3\}.
\]

We instead encode each interaction index
$j\in\{1,\dots,m\}$ by a base-$3$ string
$x^{(j)}$ of length
\[
\ell=\left\lceil\log_3 m\right\rceil,
\]
where $x^{(j)}_r\in\{1,2,3\}$. Replacing the computational-basis
projectors by the Pauli selector operators gives
\[
\widetilde h_{\mathrm{P}}
(v_1,\dots,v_k,u_1,\dots,u_\ell)
=
\sum_{j=1}^{m}
h^{(j)}_{v_1,\dots,v_k}
\otimes
P^{(x^{(j)}_1)}_{u_1}
\otimes\cdots\otimes
P^{(x^{(j)}_\ell)}_{u_\ell}.
\]

To route an arbitrary base-$3$ tag, the pinned auxiliary register
contains $\ell$ triplets, each fixed to
\[
\Ket{s_1}\otimes\Ket{s_2}\otimes\Ket{s_3}.
\]
For the $r$-th selector digit, we route the qubit in state
$\Ket{s_{x^{(j)}_r}}$ into the corresponding selector slot. Thus the
routed selector state is
\[
\Ket{x^{(j)}}
=
\bigotimes_{r=1}^{\ell}
\Ket{s_{x^{(j)}_r}}.
\]
Using
\[
\Bra{s_q}P^{(r)}\Ket{s_q}
=
\delta_{qr},
\]
compression to this tag gives
\[
\left(
I\otimes\Bra{x^{(j)}}
\right)
\widetilde h_{\mathrm P}
\left(
I\otimes\Ket{x^{(j)}}
\right)
=
h^{(j)}.
\]

\subsubsection{Global Construction and Resource Bounds}

Let $\widetilde h$ denote either of the two selector interactions above.
For each target term
\[
h^{(j_i)}_{V_i},
\]
route an $\ell$-tuple $U_i$ of pinned selector qubits whose joint state
is the corresponding tag $\Ket{x^{(j_i)}}$, and define
\[
\widetilde K
=
\sum_{i=1}^{M}
\widetilde h_{V_i,U_i}.
\]

Let $\Ket{\Psi_{\mathrm{ancilla}}}$ denote the global pinned auxiliary
state and define
\[
V\Ket{\psi}
=
\Ket{\psi}_W
\otimes
\Ket{\Psi_{\mathrm{ancilla}}}_A.
\]
The selector identities above give
\[
V^\dagger\widetilde K V
=
\sum_{i=1}^{M}
h^{(j_i)}_{V_i}
=
K.
\]
Thus the compression of $\widetilde K$ to the pinned sector reproduces
the finite-alphabet Hamiltonian exactly.

\begin{thm}[Pinned Single-Interaction Compiler]
\label{thm:pinned-resource-bounds}
Let
\[
\mathcal{A}
=
\{h^{(1)},\ldots,h^{(m)}\}
\]
be a fixed finite alphabet of $k$-local Hermitian interactions, and let
\[
K=\sum_{i=1}^{M} h^{(j_i)}_{V_i}
\]
be a Hamiltonian whose local terms are drawn from $\mathcal{A}$. Assume that the pairs $(V_i,j_i)$ are pairwise distinct; equivalently,
no interaction $h^{(j)}$ occurs more than once on the same ordered
working tuple.

There exists a single interaction matrix $\widetilde h$ and a pinned
auxiliary register such that, if

\[
V\Ket{\psi}
=
\Ket{\psi}_W\otimes\Ket{\Psi_{\mathrm{ancilla}}}_A,
\]

then the resulting Hamiltonian $\widetilde K$ satisfies

\[
V^\dagger \widetilde K V = K.
\]

Moreover, if the entries of every interaction in $\mathcal A$ are
polynomial-time computable, then the entries of $\widetilde h$ are
polynomial-time computable as well.

For the binary selector, the required resources are
\[
k' = k+\left\lceil\log_2 m\right\rceil,
\qquad
p=2\left\lceil\log_2m\right\rceil.
\]
For the base-$3$ Pauli selector, they are
\[
k' = k+\left\lceil\log_3m\right\rceil,
\qquad
p=3\left\lceil\log_3m\right\rceil.
\]

In particular, if $k$ and $m$ are constants, then both the locality
overhead and the number of pinned qubits are constant.
\end{thm}

\begin{proof}
The selector identities established above give
\[
V^\dagger\widetilde K V=K.
\]

For the binary construction,
\[
\ell=\left\lceil\log_2m\right\rceil,
\]
so the common interaction acts on $k+\ell$ qubits, while the pinned
resource pool contains two qubits per selector digit. Hence
\[
k'
=
k+\left\lceil\log_2m\right\rceil,
\qquad
p
=
2\left\lceil\log_2m\right\rceil.
\]

For the base-$3$ construction,
\[
\ell=\left\lceil\log_3m\right\rceil,
\]
and each selector digit requires one pinned qubit in each of the three
states $\Ket{s_1},\Ket{s_2},\Ket{s_3}$. Therefore
\[
k'
=
k+\left\lceil\log_3m\right\rceil,
\qquad
p
=
3\left\lceil\log_3m\right\rceil.
\]

It remains only to verify the ordered-hyperedge condition. Consider two
target terms indexed by $i\neq i'$. If $V_i\neq V_{i'}$, their compiled
tuples already differ in the working-qubit positions. If
$V_i=V_{i'}$, the hypothesis that $(V_i,j_i)$ and
$(V_{i'},j_{i'})$ are distinct implies $j_i\neq j_{i'}$. Their selector
tags therefore differ in at least one position, where a different
pinned physical qubit is routed into the compiled tuple. Hence all
compiled ordered tuples are pairwise distinct. Each tuple also contains
pairwise distinct physical qubits because the working tuple has distinct
entries, different selector positions use different auxiliary
pairs or triplets, and the working and auxiliary registers are disjoint.

Finally, the entries of $\widetilde h$ are obtained from the finitely
many matrices in $\mathcal A$ by finite sums and tensor products with
fixed projectors or Pauli matrices. Polynomial-time computability of
$\mathcal A$ therefore implies polynomial-time computability of
$\widetilde h$.
\end{proof}

\subsection{The Unpinned Single-Interaction Reduction}
\label{sec:unpinned-compiler}

We now remove the pinning resource from the binary-selector compiler.
Rather than restricting the auxiliary register to the desired selector
state, we enforce that state energetically. Since the selector and check
operators are diagonal in the computational basis, the resulting
Hamiltonian decomposes into invariant auxiliary sectors, allowing the
valid and invalid configurations to be analyzed separately.

\subsubsection{Penalty Construction and Dedicated Auxiliary Blocks}

We use the binary selector with
\[
\ell=\lceil\log_2 m\rceil.
\]
The base-$3$ Pauli selector does not provide the computational-basis
sector decomposition required below, since its selector operators are
not diagonal in that basis.

Each target occurrence receives its own dedicated block of
$\ell+1$ auxiliary pairs.  This prevents the degree of any auxiliary
qubit from growing with the number of target terms and also
distinguishes repeated occurrences of the same interaction on the same
working tuple. 
The complete auxiliary register contains
\[
M(\ell+1)
\]
pairs, each intended to lie in the state $\Ket{01}$.

To enforce this configuration, define the two-qubit check operator
\[
h_{\mathrm{check}}(u_{\ell+1},u_{\ell+2})
=
I_{u_{\ell+1},u_{\ell+2}}
-
\Pi^{\Ket{0}}_{u_{\ell+1}}
\otimes
\Pi^{\Ket{1}}_{u_{\ell+2}}.
\]
Thus $h_{\mathrm{check}}$ projects onto the orthogonal complement of
the valid pair state $\Ket{01}$.  We incorporate this check into the
common interaction by defining the modified interaction matrix $\widetilde h_{\mathrm{unp}}$ acting on $k$ working qubits, $\ell$ selector qubits, and 2 check qubits as:

\begin{align}
    \widetilde h_{\mathrm{unp}}(v_1,\dots,v_k,u_1,\dots,u_{\ell+2})
    &= \frac{1}{\ell+1}
    \sum_{j=1}^m
    \left(
    h^{(j)}_{v_1,\dots,v_k}
    \otimes
    \Pi^{x^{(j)}_1}_{u_1}
    \otimes \cdots \otimes
    \Pi^{x^{(j)}_\ell}_{u_\ell}
    \otimes I_{u_{\ell+1},u_{\ell+2}}
    \right)
    \nonumber\\
    &\quad
    +\Delta
    \left(
    I_{v_1,\dots,v_k,u_1,\dots,u_\ell}
    \otimes
    h_{\mathrm{check}}(u_{\ell+1},u_{\ell+2})
    \right).
    \label{eq:unpinned-single-interaction}
\end{align}

Let
\[
J:=\max_{j\in[m]}\|h^{(j)}\|.
\]
Fix a parameter $\gamma>0$ and choose the check penalty
\[
\Delta = 2J+\gamma.
\]
The factor $1/(\ell+1)$ bounds the total computational contribution of
a target block by $J$ in every fixed auxiliary sector.  Choosing
$\Delta=2J+\gamma$ then leaves an energy separation of at least
$\gamma$ whenever that block contains an invalid auxiliary pair.

\subsubsection{Exact Spectral Preservation}

For each target occurrence $i$, let its dedicated auxiliary block
contain $\ell+1$ ordered pairs.  We apply
$\widetilde h_{\mathrm{unp}}$ once for each pair:
\[
\widetilde K
=
\sum_{i=1}^{M}
\sum_{r=1}^{\ell+1}
\left(\widetilde h_{\mathrm{unp}}\right)_{V_i,U_i^{(r)}}.
\]
In the $r$-th application, the $r$-th pair occupies the two check
positions and the remaining $\ell$ pairs encode the tag of
$h^{(j_i)}$.  Thus every pair is checked exactly once, while in the
valid auxiliary configuration all $\ell+1$ applications select the
same target interaction.

For each computational-basis configuration
\[
z\in\{0,1\}^{2M(\ell+1)}
\]
of the auxiliary register, define
\[
\mathcal H_z
=
\mathcal H_W
\otimes
\operatorname{span}\{\Ket z_A\}.
\]
Since every operator acting on the auxiliary register is diagonal in
the computational basis, each $\mathcal H_z$ is invariant under
$\widetilde K$. Let
\[
\widetilde K(z)
=
\left(I_W\otimes\Bra z_A\right)
\widetilde K
\left(I_W\otimes\Ket z_A\right)
\]
denote the induced operator on the working register. The unique valid
configuration is
\[
z_{\mathrm{val}}
=
(01)^{M(\ell+1)};
\]
all other configurations are called invalid.

\begin{thm}[Exact Low-Energy Preservation of the Single-Interaction Compiler]
\label{thm:unpinned-spectral-preservation}
Let
\[
\mathcal{A}
=
\{h^{(1)},\ldots,h^{(m)}\}
\]
be a fixed finite alphabet of $k$-local Hermitian interactions, let
\[
\ell=\left\lceil\log_2m\right\rceil,
\qquad
J=\max_{j\in[m]}\|h^{(j)}\|,
\]
and consider
\[
K=\sum_{i=1}^{M} h^{(j_i)}_{V_i}.
\]
No distinctness assumption is imposed on the pairs $(V_i,j_i)$:
the same interaction may occur more than once on the same ordered
working tuple.
Fix $\gamma>0$, set
\[
\Delta=2J+\gamma,
\]
and construct $\widetilde K$ using the interaction in
\cref{eq:unpinned-single-interaction}.

Let
\[
\Ket{\Omega}_A
=
\Ket{01}^{\otimes M(\ell+1)}
\]
denote the valid auxiliary configuration, and define the isometry
\[
V:\mathcal{H}_W\longrightarrow\mathcal{H}_W\otimes\mathcal{H}_A,
\qquad
V\Ket{\psi}
=
\Ket{\psi}\otimes\Ket{\Omega}_A.
\]
Let
\[
E_0=\lambda_{\min}(K).
\]

Then the following properties hold:
\begin{enumerate}
    \item The valid subspace $\operatorname{im}(V)$ is invariant under
    $\widetilde K$, and
    \[
    V^\dagger \widetilde K V=K.
    \]

    \item For every invalid computational-basis auxiliary sector $z$,
    \[
    \widetilde K(z)\succeq K+\gamma I,
    \]
    where $\widetilde K(z)$ denotes the operator induced on the working
    register in sector $z$.

    \item The ground space is preserved exactly:
    \[
    \operatorname{Ground}(\widetilde K)
    =
    V\,\operatorname{Ground}(K).
    \]
    In particular, the ground-state degeneracy is unchanged.

    \item Every eigenvalue of $K$ strictly below $E_0+\gamma$ appears in
    $\widetilde K$ with exactly the same multiplicity, and
    $\widetilde K$ has no additional eigenvalues below this cutoff.

    \item Consequently,
    \[
    \lambda_{\min}(\widetilde K)=\lambda_{\min}(K),
    \]
    so the same Local Hamiltonian promise thresholds may be used before
    and after compilation.

    \item If $K$ has spectral gap $\delta>0$ above its ground space and
    $\widetilde\delta$ denotes the corresponding gap of $\widetilde K$, then
    \[
    \min\{\delta,\gamma\}
    \leq
    \widetilde\delta
    \leq
    \delta.
    \]
    In particular, if $\gamma\geq\delta$, then
    \[
    \widetilde\delta=\delta.
    \]
\end{enumerate}

The constructed interaction has locality
\[
k'
=
k+\ell+2
=
k+\left\lceil\log_2m\right\rceil+2.
\]
Hence the locality overhead is constant whenever $k$ and $m$ are
constant.

The construction introduces $2M(\ell+1)$ auxiliary qubits. Hence, if the
target Hamiltonian acts on $N$ qubits, the compiled Hamiltonian acts on
\[
N'
=
N+2M(\ell+1)
\]
qubits. 
In particular, for constant $m$ and $M=\operatorname{poly}(N)$, the construction has polynomial size. 
Moreover, all ordered interaction tuples in the compiled Hamiltonian are pairwise distinct, even when identical pairs $(V_i,j_i)$ occur multiple times in $K$. 

\end{thm}

\begin{lemma}[Exact Embedding of the Valid Sector]
\label{lem:unpinned-valid-sector}
The valid auxiliary subspace is invariant under $\widetilde K$, and
\[
V^\dagger\widetilde K V=K.
\]
\end{lemma}

\begin{proof}
All operators acting on the auxiliary register are diagonal in the
computational basis, so every computational-basis auxiliary sector is
invariant under $\widetilde K$.

In the valid configuration $\Ket{\Omega}_A$, every check contribution
vanishes.
For each target interaction $h^{(j_i)}_{V_i}$, the selector qubits match the
assigned binary tag in each of the $\ell+1$ applications belonging to the
$i$-th target block. Hence each such application contributes
\[
\frac{1}{\ell+1}h^{(j_i)}_{V_i}.
\]
Summing the $\ell+1$ applications gives
\[
h^{(j_i)}_{V_i}.
\]
Summing over all target blocks therefore yields
\[
V^\dagger\widetilde KV
=
\sum_{i=1}^{M}h^{(j_i)}_{V_i}
=
K.
\]
\end{proof}

\begin{lemma}[Separation of Invalid Sectors]
\label{lem:unpinned-invalid-sector}
Let $z$ be a computational-basis configuration of the auxiliary register,
and let $r(z)$ denote the number of target blocks containing at least one
invalid auxiliary pair. Then
\[
\widetilde K(z)
\succeq
K+r(z)\gamma I.
\]
In particular, if $z$ is an invalid auxiliary sector, equivalently $z\neq z_{\mathrm{val}}$, then $r(z)\geq 1$ and hence
\[
\widetilde K(z)
\succeq
K+\gamma I.
\]
\end{lemma}

\begin{proof}
Because every operator acting on the auxiliary register is diagonal in the
computational basis,
\[
\widetilde K
=
\bigoplus_z \widetilde K(z).
\]

Fix an auxiliary sector $z$. For target block $i$, let $z_i$ denote its
local auxiliary configuration and write
\[
\widetilde K_i(z_i)
=
C_i(z_i)+t_i(z_i)\Delta I,
\]
where $C_i(z_i)$ is the computational contribution and
$t_i(z_i)$ is the number of invalid auxiliary pairs in that block.

For fixed $z_i$, each of the $\ell+1$ applications in block $i$
contributes either zero or
\[
\frac{1}{\ell+1}h^{(j)}_{V_i}
\]
for some $j\in[m]$.
Therefore
\[
\|C_i(z_i)\|
\leq
J.
\]
Since also
\[
\|h^{(j_i)}_{V_i}\|
\leq J,
\]
we obtain
\[
C_i(z_i)-h^{(j_i)}_{V_i}
\succeq
-2JI.
\]

If block $i$ is invalid, then $t_i(z_i)\geq1$. Hence
\[
\begin{aligned}
\widetilde K_i(z_i)-h^{(j_i)}_{V_i}
&=
C_i(z_i)-h^{(j_i)}_{V_i}
+t_i(z_i)\Delta I\\
&\succeq
-2JI+(2J+\gamma)I\\
&=
\gamma I.
\end{aligned}
\]

If block $i$ is valid, then $t_i(z_i)=0$ and the routing construction gives
exactly
\[
\widetilde K_i(z_i)
=
h^{(j_i)}_{V_i}.
\]

Summing these operator inequalities over all target blocks gives
\[
\widetilde K(z)-K
\succeq
r(z)\gamma I,
\]
which proves the claim.
\end{proof}

\begin{proof}[Proof of \cref{thm:unpinned-spectral-preservation}]
By \cref{lem:unpinned-valid-sector}, the restriction of $\widetilde K$ to the valid auxiliary
sector is exactly isometric to $K$.  Hence every eigenvalue of $K$
appears in $\widetilde K$ with its original multiplicity.

By \cref{lem:unpinned-invalid-sector}, every invalid auxiliary sector obeys
\[
\widetilde K(z)
\succeq
K+\gamma I.
\]
In particular,
\[
\lambda_{\min}\!\left(\widetilde K(z)\right)
\geq
E_0+\gamma.
\]
Thus no invalid sector contains an eigenvalue strictly below
$E_0+\gamma$.

It follows that the spectrum of $\widetilde K$ below this cutoff consists
exactly of the corresponding eigenvalues of $K$ in the valid sector, with
the same multiplicities. In particular,
\[
\operatorname{Ground}(\widetilde K)
=
V\,\operatorname{Ground}(K),
\]
and therefore
\[
\lambda_{\min}(\widetilde K)
=
\lambda_{\min}(K).
\]

Now suppose that $K$ has spectral gap $\delta>0$ above its ground space.
The valid sector contains an excited eigenvalue at energy
\[
E_0+\delta,
\]
so
\[
\widetilde\delta\leq\delta.
\]
On the other hand, every excited state in the valid sector has energy at
least $E_0+\delta$, while every invalid sector has energy at least
$E_0+\gamma$. Hence
\[
\widetilde\delta
\geq
\min\{\delta,\gamma\}.
\]
If $\gamma\geq\delta$, both bounds coincide and
\[
\widetilde\delta=\delta.
\]

It remains to verify that all compiled ordered hyperedges are valid and pairwise distinct. 
Each occurrence $i$ is assigned its own dedicated auxiliary block, and the blocks for
different occurrences are disjoint. Hence applications associated with
different occurrences have different auxiliary entries, even if the
corresponding pairs $(V_i,j_i)$ are identical.

For a fixed occurrence $i$, the $\ell+1$ applications use different
auxiliary pairs in the two check positions. Therefore these applications
also give pairwise distinct ordered tuples. Thus every application of
$\widetilde h_{\mathrm{unp}}$ occurs on a distinct ordered hyperedge.
Each such tuple also has pairwise distinct entries: the working tuple
contains distinct qubits, the selector qubits are drawn from distinct
auxiliary pairs, the two check qubits form another auxiliary pair, and
the working and auxiliary registers are disjoint.

The locality statement follows directly from the construction in
\cref{eq:unpinned-single-interaction}.
\end{proof}

\subsubsection{Preserved Structural Properties}

\begin{lemma}[Preservation of Geometric Locality]
\label{lem:geometric-locality-preservation}
Suppose the target interaction hypergraph is geometrically local according
to \cref{def: geometrically local}, the interaction alphabet has constant
size $m=O(1)$, and every ordered working tuple occurs in at most $O(1)$
target-interaction occurrences, counting multiplicity. Then the unpinned
single-interaction construction preserves geometric locality and bounded
degree up to constant overhead.
\end{lemma}

\begin{proof}
Let
\[
\ell=\left\lceil\log_2m\right\rceil.
\]
Each target interaction is assigned a dedicated auxiliary block consisting
of $\ell+1$ pairs of qubits. Since $m=O(1)$, every such block contains only
$O(1)$ auxiliary qubits.

Consider a target interaction supported on the tuple $V_i$. By geometric
locality, all computational qubits in $V_i$ lie within a region of constant
diameter. Since the underlying lattice has fixed dimension and the original
interaction hypergraph has bounded degree, only $O(1)$ interaction sites
occur within any constant-radius neighborhood. 
After a constant-factor rescaling of the lattice, the dedicated auxiliary blocks can therefore be placed within constant distance of their corresponding computational supports.

Every application of the constructed single interaction associated with
$V_i$ acts only on the qubits of $V_i$ and on qubits from its dedicated
nearby auxiliary block. Hence every resulting hyperedge has constant
geometric diameter.

Furthermore, each auxiliary qubit participates only in the $O(1)$
applications associated with its own block. Since the target interaction
hypergraph has bounded degree and each ordered working tuple occurs with
only $O(1)$ multiplicity, every computational qubit participates in only
$O(1)$ target occurrences. Each target occurrence is replaced by
$\ell+1=O(1)$ applications of the compiled interaction. Hence every
computational and auxiliary qubit has bounded degree in the compiled
interaction hypergraph.

Therefore the resulting interaction hypergraph is geometrically local.
\end{proof}

\paragraph{Uniform $n$-dependent alphabets.}
The pinned and unpinned finite-alphabet compilers apply pointwise to a
constant-size family
\[
\mathcal A^{(n)}
=
\{h^{(1,n)},\ldots,h^{(m,n)}\},
\qquad m=O(1),
\]
provided that every parameter entering the compiled interaction is
uniformly computable from the final physical system size $n$.  In
particular, the matrices in $\mathcal A^{(n)}$ and, for the unpinned
compiler, the penalty strength must be computable in polynomial time
from $1^n$ and the requested precision.

The final-system-size convention is essential: a source-system size or
other instance-specific quantity may enter the construction only if it
can itself be recovered uniformly from the final physical system size,
including all auxiliary and pinned qubits. \Cref{sec:n-dependent-hardness}
implements this bookkeeping explicitly.

\paragraph{Stoquasticity.}
The binary pinned and unpinned compilers preserve stoquasticity.
If every interaction in the source alphabet is stoquastic, tensoring
with the nonnegative diagonal selector projectors preserves
non-positive off-diagonal entries, while the check operator is
diagonal.  Hence the resulting single interaction is stoquastic.

\subsection{Weighted and Universal Finite-Alphabet Normalization}
\label{sec:universal-finite-alphabet-normalization}

The exact compilers of the preceding sections assume that the source
Hamiltonian is already expressed as an unweighted sum over a fixed
finite interaction alphabet.  We now explain how to reach that setting
from weighted and, ultimately, arbitrary fixed-locality Hamiltonians.

The reductions in this section have two conceptually distinct stages.
First, scalar coefficients are rounded to a common inverse-polynomial
scale and represented by multiplicity; this normalization step is
approximate.  The resulting occurrence Hamiltonian is then converted
to strict single-interaction form by an exact compiler.  For a
positive-weight singleton source, a specialized exact occurrence
compiler improves the locality overhead by one qubit.

\subsubsection{Weighted finite alphabets}

As reviewed in \cref{subsubsec:s-hamiltonians}, an
$S$-Hamiltonian allows independently specified scalar weights
multiplying interactions from a fixed finite set $S$.  We first remove
these weights, to inverse-polynomial accuracy, by rounding them to a
common scale and representing the resulting integer coefficients by
multiplicity.

\begin{lemma}[Weighted finite-alphabet normalization]
\label{lem:weighted-finite-alphabet-normalization}
Let
\[
S=\{s^{(1)},\ldots,s^{(m)}\}
\]
be a fixed finite set of $k$-local Hermitian interactions, and define
\[
J=\max_{j\in[m]}\|s^{(j)}\|.
\]
Consider a Hamiltonian
\[
H
=
\sum_{i=1}^{M}
\alpha_i s^{(j_i)}_{V_i},
\]
where $M=\operatorname{poly}(n)$, the coefficients $\alpha_i$ are
specified with $\operatorname{poly}(n)$ bits, and
\[
|\alpha_i|\le B
\]
for some $B=\operatorname{poly}(n)$.

For every $\varepsilon>0$, define
\[
Q
=
\max\left\{
1,
\left\lceil
\frac{MJ}{2\varepsilon}
\right\rceil
\right\}.
\]
Then there is a Hamiltonian $K$, written as a sum of occurrences from
the fixed alphabet
\[
S\cup(-S),
\]
such that
\[
\left\|
H-\frac{1}{Q}K
\right\|
\le \varepsilon.
\]

If every $\alpha_i\ge 0$, then $K$ may be chosen using interactions from
$S$ alone.

Moreover, if $1/\varepsilon= O( \operatorname{poly}(n))$, then the number of
occurrences in $K$ and the time required to construct $K$ are
polynomial in $n$.
\end{lemma}

\begin{proof}
For every $i$, let
\[
z_i
=
\operatorname{round}(Q\alpha_i),
\]
so that
\[
\left|
\alpha_i-\frac{z_i}{Q}
\right|
\le
\frac{1}{2Q}.
\]
Represent the integer coefficient $z_i$ by multiplicity:
\[
K
=
\sum_{i=1}^{M}
\sum_{r=1}^{|z_i|}
\bigl(
\operatorname{sgn}(z_i)s^{(j_i)}
\bigr)_{V_i},
\]
where the inner sum is absent when $z_i=0$.

Since $\|s^{(j_i)}\|\le J$,
\[
\begin{aligned}
\left\|
H-\frac{1}{Q}K
\right\|
&\le
\sum_{i=1}^{M}
\left|
\alpha_i-\frac{z_i}{Q}
\right|
\|s^{(j_i)}\| \\
&\le
\frac{MJ}{2Q}
\le
\varepsilon.
\end{aligned}
\]

If every $\alpha_i\ge0$, then every $z_i\ge0$, so only interactions from
$S$ occur.

Finally,
\[
|z_i|
\le
QB+\frac12,
\]
and hence the total number $R$ of occurrences satisfies
\[
R
\le
M\left(QB+\frac12\right).
\]
Thus $R=\operatorname{poly}(n)$ whenever $M$, $B$, and
$1/\varepsilon$ are polynomially bounded. The rounded coefficients and
the complete occurrence list can likewise be computed in polynomial
time.
\end{proof}

\begin{cor}[Generic weighted single-interaction compilation]
\label{cor:weighted-finite-alphabet-sih}
Let
\[
S=\{s^{(1)},\ldots,s^{(m)}\}
\]
be a fixed finite set of $k$-local Hermitian interactions, and consider
Hamiltonians
\[
H
=
\sum_{i=1}^{M}
\alpha_i s^{(j_i)}_{V_i}
\]
with polynomially bounded, polynomial-bit coupling strengths.

For arbitrary signed coefficients, there is a polynomial-time
promise-preserving reduction, to inverse-polynomial accuracy and up to a
known global energy scale, to
\[
\mathrm{SIH}_{
k+\lceil\log_2(2m)\rceil+2
}.
\]

If all coefficients are nonnegative, the locality improves to
\[
\mathrm{SIH}_{
k+\lceil\log_2m\rceil+2
}.
\]

In both cases the final single interaction is fixed independently of the
source instance, the system size, the rounding scale $Q$, and the chosen
inverse-polynomial accuracy.
\end{cor}

\begin{proof}
Apply \cref{lem:weighted-finite-alphabet-normalization} to obtain the rounded
finite-alphabet Hamiltonian $K$.  For arbitrary signed weights, $K$
uses the fixed alphabet
\[
    S\cup(-S),
\]
whose cardinality is at most $2m$.  Applying the exact unpinned
compiler of \cref{thm:unpinned-spectral-preservation} therefore gives locality at
most
\[
    k+\lceil\log_2(2m)\rceil+2.
\]
For nonnegative weights, $K$ uses $S$ itself, giving locality
\[
    k+\lceil\log_2m\rceil+2.
\]

Let $\widetilde K$ denote the single-interaction Hamiltonian produced
by the unpinned compiler, and let $V$ be its encoding isometry. Then
\[
V^\dagger \widetilde K V = K.
\]

The compilation of $K$ is exact; the only approximation is the
preceding rounding of the source weights. Hence
\[
\left\|
H-\frac{1}{Q}V^\dagger\widetilde K V
\right\|
\le\varepsilon.
\]
If $(H,a,b)$ is a Local Hamiltonian instance and
\[
0<\varepsilon<\frac{b-a}{4},
\]
one may therefore use the thresholds
\[
\widetilde a=Q(a+\varepsilon),
\qquad
\widetilde b=Q(b-\varepsilon).
\]

Repeated rounded occurrences cause no conflict with the strict SIH
definition: by
\cref{thm:unpinned-spectral-preservation}, every occurrence receives a
dedicated auxiliary block, so all final ordered interaction tuples are
pairwise distinct.

Finally, the alphabet supplied to the unpinned compiler is fixed before
the instance is given. Hence its selector structure and penalty strength
are fixed as well, and the resulting interaction does not depend on
$Q$.
\end{proof}

\subsubsection{Sharper singleton unweighting}

When the weighted source uses a single fixed interaction with
nonnegative coefficients, the generic weighted construction can be
sharpened.  After the weights have been converted to multiplicities,
a particularly simple exact occurrence compiler increases the
locality by only one.

\begin{lemma}[Positive-weight singleton unweighting]
\label{lem:positive-singleton-unweighting}
Fix a $k$-qubit Hermitian interaction $h$. Choose a fixed rational
constant
\[
c>\lambda_{\max}(h)
\]
and define
\[
\delta
=
c-\lambda_{\max}(h)>0.
\]
Let
\[
g_h
=
h\otimes\Pi^{\Ket{1}}
+
cI\otimes\Pi^{\Ket{0}}.
\]

Consider any occurrence Hamiltonian
\[
K
=
\sum_{t=1}^{R}h_{V_t},
\]
where repeated occurrences on the same ordered working tuple are
allowed. Introduce one fresh auxiliary qubit $a_t$ for every occurrence
and define
\[
\widetilde K
=
\sum_{t=1}^{R}
(g_h)_{V_t,a_t}.
\]

Let
\[
\Ket{\Omega}_A
=
\Ket{1}^{\otimes R}
\]
and define
\[
V\Ket{\psi}
=
\Ket{\psi}\otimes\Ket{\Omega}_A.
\]
Then
\[
V^\dagger\widetilde K V=K.
\]

Moreover, for every computational-basis auxiliary sector
$z\in\{0,1\}^{R}$, if
\[
r(z)
=
|\{t:z_t=0\}|,
\]
then
\[
\widetilde K(z)
\succeq
K+r(z)\delta I.
\]
Consequently,
\[
\operatorname{Ground}(\widetilde K)
=
V\,\operatorname{Ground}(K),
\qquad
\lambda_{\min}(\widetilde K)
=
\lambda_{\min}(K).
\]

All ordered interaction tuples in $\widetilde K$ are pairwise distinct.
Thus $g_h$ defines a strict
\[
\mathrm{SIH}_{k+1}
\]
interaction.

If $h$ is stoquastic in the computational basis, then $g_h$ is
stoquastic as well.
\end{lemma}

\begin{proof}
Because the auxiliary operators are diagonal, $\widetilde K$ decomposes
into computational-basis sectors of the auxiliary register.

In the all-$1$ sector, every occurrence contributes $h_{V_t}$, so
\[
V^\dagger\widetilde K V
=
K.
\]

If $z_t=0$, the $t$-th occurrence contributes $cI$ instead of
$h_{V_t}$. Hence
\[
\widetilde K(z)
=
K
+
\sum_{t:z_t=0}
\left(
cI-h_{V_t}
\right).
\]
Since
\[
cI-h_{V_t}
\succeq
\left(
c-\lambda_{\max}(h)
\right)I
=
\delta I,
\]
we obtain
\[
\widetilde K(z)
\succeq
K+r(z)\delta I.
\]
Thus every invalid auxiliary sector is separated from the
corresponding valid-sector Hamiltonian by at least the fixed constant
$\delta$.

Finally, occurrence $t$ uses its own fresh auxiliary qubit $a_t$.
Therefore distinct occurrences give distinct ordered
$(k+1)$-tuples even when their working tuples $V_t$ coincide.

If $h$ is stoquastic, the only off-diagonal matrix elements of $g_h$
occur in the auxiliary-$1$ block and are exactly those of $h$.
Therefore $g_h$ is stoquastic.
\end{proof}

\begin{cor}[Positive-weight singleton Hamiltonians]
\label{cor:positive-weight-singleton-sih}
Fix a $k$-qubit Hermitian interaction $h$. A positive-weight singleton
Hamiltonian
\[
H
=
\sum_{i=1}^{M}
\alpha_i h_{V_i},
\qquad
\alpha_i\ge0,
\]
with polynomially bounded, polynomial-bit weights admits a
polynomial-time promise-preserving reduction, to inverse-polynomial
accuracy and up to a known global energy scale, to
\[
\mathrm{SIH}_{k+1}.
\]

Equivalently,
\[
\{h\}^{+}\text{-Hamiltonian}
\longrightarrow
\mathrm{SIH}_{k+1}.
\]
The final interaction is fixed independently of the rounding scale $Q$
and of the source instance.
\end{cor}

\begin{proof}
Apply
\cref{lem:weighted-finite-alphabet-normalization}
with $S=\{h\}$ and nonnegative coefficients. This gives
\[
\left\|
H-\frac1Q K
\right\|
\le\varepsilon,
\]
where
\[
K=\sum_{t=1}^{R}h_{V_t}
\]
is represented entirely by multiplicity.

Apply
\cref{lem:positive-singleton-unweighting}
to $K$. For its encoding isometry $V$,
\[
V^\dagger\widetilde K V=K.
\]
Thus the singleton occurrence compilation is exact; as in
\cref{cor:weighted-finite-alphabet-sih}, the only approximation comes
from rounding the original coupling strengths. Hence
\[
\left\|
H-\frac1Q V^\dagger\widetilde K V
\right\|
\le\varepsilon.
\]
For a source promise $(a,b)$, choosing
$\varepsilon<(b-a)/4$ gives output thresholds
\[
\widetilde a=Q(a+\varepsilon),
\qquad
\widetilde b=Q(b-\varepsilon).
\]
The interaction $g_h$ depends only on the fixed matrix $h$ and the fixed
constant $c$, not on $Q$.
\end{proof}

\subsubsection{Universal normalization of arbitrary local terms}

We now remove the remaining assumption that the source terms are drawn
from a fixed interaction set.  Because $k$ is constant, every local
term can be expanded in the fixed $k$-qubit Pauli basis.  Rounding all
Pauli coefficients to a common scale and representing the resulting
integers by multiplicity reduces an arbitrary $k$-local Hamiltonian to
an unweighted occurrence Hamiltonian over a fixed signed Pauli
alphabet.

\begin{lemma}[Universal finite-alphabet normalization]
\label{lem:universal-finite-alphabet-normalization}
Fix $k=O(1)$ and let
\[
H=\sum_{i=1}^{M} H_i
\]
be a $k$-local Hamiltonian on $n$ qubits, where
$M=\operatorname{poly}(n)$ and
\[
\|H_i\|\le B
\]
for some $B=\operatorname{poly}(n)$.

Define the fixed signed Pauli alphabet
\[
\mathcal A_k^\star
=
\left\{
\pm P:
P\in\{I,X,Y,Z\}^{\otimes k}
\right\}.
\]
In particular,
\[
|\mathcal A_k^\star|=2\cdot 4^k.
\]

For every $\varepsilon>0$, there is an integer
\[
Q
=
\max\left\{
1,
\left\lceil
\frac{M4^k}{2\varepsilon}
\right\rceil
\right\}
\]
and a finite-alphabet Hamiltonian $K$, written as a sum of occurrences
of interactions from $\mathcal A_k^\star$, such that
\[
\left\|
H-\frac{1}{Q}K
\right\|
\le \varepsilon.
\]

Moreover, if $1/\varepsilon= O( \operatorname{poly}(n))$ and
$\varepsilon$ is specified with $\operatorname{poly}(n)$ bits, then
$Q$, the number of occurrences in $K$, and the time required to
construct $K$ are all polynomial in $n$.
\end{lemma}

\begin{proof}

Since $k$ is fixed, we may assume without loss of generality that $n\geq k$, adjoining at most $k-n=O(1)$ isolated qubits if necessary and continuing to denote the resulting system size by $n$.

A $k$-local term may act nontrivially on fewer than $k$ qubits. To express every $H_i$ in the same $k$-qubit Pauli basis, choose an ordered $k$-tuple $V_i$ of distinct physical qubits containing the support of $H_i$, and regard $H_i$ as an operator on $V_i$ by tensoring with the identity on the additional qubits. This enlargement does not change the action of $H_i$ on the full system and is used only to obtain a uniform $k$-qubit representation of all local terms.

Let
\[
\mathcal P_k=\{I,X,Y,Z\}^{\otimes k}.
\]
Using the Pauli basis, write
\[
H_i
=
\sum_{P\in\mathcal P_k}
\alpha_{i,P} P,
\qquad
\alpha_{i,P}
=
\frac{1}{2^k}\operatorname{Tr}(P H_i).
\]
Since $H_i$ and every $P\in\mathcal P_k$ are Hermitian,
$\alpha_{i,P}\in\mathbb R$.

Fix $\varepsilon>0$, and set
\[
Q
=
\max\left\{
1,
\left\lceil
\frac{M4^k}{2\varepsilon}
\right\rceil
\right\}.
\]
For every pair $(i,P)$, choose the nearest integer
\[
z_{i,P}
=
\operatorname{round}(Q\alpha_{i,P}),
\]
so that
\[
\left|
\alpha_{i,P}
-
\frac{z_{i,P}}{Q}
\right|
\le
\frac{1}{2Q}.
\]

We represent each integer coefficient $z_{i,P}$ by multiplicity. If $z_{i,P}>0$, include
$z_{i,P}$ separate occurrences of $+P$ on $V_i$; if
$z_{i,P}<0$, include $|z_{i,P}|$ separate occurrences of
$-P$ on $V_i$; and if $z_{i,P}=0$, include no occurrence.
Equivalently,
\[
K
=
\sum_{i=1}^{M}
\sum_{P\in\mathcal P_k}
\sum_{r=1}^{|z_{i,P}|}
\bigl(\operatorname{sgn}(z_{i,P})P\bigr)_{V_i},
\]
with the convention that the innermost sum is absent when
$z_{i,P}=0$.

The identity Pauli $I^{\otimes k}$ is included in
$\mathcal P_k$, so this construction retains the scalar component
of every local term.

Since every Pauli operator has operator norm one,
\[
\begin{aligned}
\left\|
H-\frac{1}{Q}K
\right\|
&=
\left\|
\sum_{i=1}^{M}
\sum_{P\in\mathcal P_k}
\left(
\alpha_{i,P}
-
\frac{z_{i,P}}{Q}
\right)
P_{V_i}
\right\| \\
&\le
\sum_{i=1}^{M}
\sum_{P\in\mathcal P_k}
\left|
\alpha_{i,P}
-
\frac{z_{i,P}}{Q}
\right| \\
&\le
\frac{M4^k}{2Q}
\le
\varepsilon.
\end{aligned}
\]

It remains to bound the size of the occurrence list.
For every Pauli coefficient,
\[
|\alpha_{i,P}|
=
\frac{1}{2^k}
\left|
\operatorname{Tr}(P H_i)
\right|
\le
\|H_i\|
\le B.
\]
Hence
\[
|z_{i,P}|
\le
QB+\frac{1}{2}.
\]
Therefore the total number $R$ of occurrences in $K$ satisfies
\begin{align}
R
=
\sum_{i=1}^{M}
\sum_{P\in\mathcal P_k}
|z_{i,P}|
\le
M4^k
\left(
QB+\frac{1}{2}
\right).
\end{align}
For fixed $k$, if $M$, $B$, and $1/\varepsilon$ are polynomially
bounded in $n$, then both $Q$ and $R$ are polynomially bounded.

Finally, because $k$ is fixed, each coefficient
$\alpha_{i,P}$ is obtained from a constant-dimensional local matrix
using a constant number of arithmetic operations on polynomial-bit
input data. Thus the coefficients, the integers $z_{i,P}$, and the
complete occurrence list defining $K$ can all be computed in
polynomial time.
\end{proof}

\begin{thm}[Universal single-interaction reduction]
\label{thm:universal-single-interaction-reduction}
Fix $k=O(1)$. Let
\[
H=\sum_{i=1}^{M} H_i
\]
be a $k$-local Hamiltonian on $n$ qubits satisfying the assumptions of
\cref{lem:universal-finite-alphabet-normalization}, and let $\varepsilon>0$ be specified with $\operatorname{poly}(n)$
bits, with $1/\varepsilon= O( \operatorname{poly}(n))$.

Then there is a polynomial-time construction of a $(3k+3)$-local
single-interaction Hamiltonian $\widetilde H$, together with an integer
$Q=\operatorname{poly}(n)$ and an isometry $V$, such that
\[
\left\|
H-\frac{1}{Q}V^\dagger \widetilde H V
\right\|
\le \varepsilon.
\]

For every fixed $k$, the single interaction used in
$\widetilde H$ is fixed independently of $H$, $n$, $Q$, and
$\varepsilon$.

Moreover, if $(H,a,b)$ is a $k$-Local Hamiltonian instance and
\[
0<\varepsilon<\frac{b-a}{4},
\]
then the construction is promise preserving with output thresholds
\[
\widetilde a=Q(a+\varepsilon),
\qquad
\widetilde b=Q(b-\varepsilon).
\]

\end{thm}

\begin{proof}
First apply \cref{lem:universal-finite-alphabet-normalization} to
replace $H$, to accuracy $\varepsilon$, by a multiplicity Hamiltonian
$K$ over the fixed signed Pauli alphabet
\[
\left\|
H-\frac{1}{Q}K
\right\|
\le\varepsilon,
\]
where $K$ is written as a polynomial number of occurrences from the
fixed signed Pauli alphabet
\[
\mathcal A_k^\star
=
\left\{
\pm P:
P\in\{I,X,Y,Z\}^{\otimes k}
\right\}.
\]

Since
\[
|\mathcal A_k^\star|
=
2\cdot4^k
=
2^{2k+1},
\]
the binary selector length in the unpinned compiler is
\[
\ell
=
\left\lceil
\log_2|\mathcal A_k^\star|
\right\rceil
=
2k+1.
\]
Every element of $\mathcal A_k^\star$ has operator norm one, so in
\cref{thm:unpinned-spectral-preservation} we may take
\[
J=1
\]
and, for example, fix
\[
\gamma=1,
\qquad
\Delta=2J+\gamma=3.
\]
Thus the resulting single interaction depends only on $k$.

We now apply the exact compiler of
\cref{thm:unpinned-spectral-preservation} to $K$.
\[
V^\dagger\widetilde H V=K.
\]
The compiler permits repeated occurrences of the same signed Pauli
operator on the same ordered working tuple, since each occurrence is
assigned its own auxiliary block. Hence the multiplicities introduced
in the normalization lemma are compatible with the pairwise-distinct
ordered-hyperedge requirement.

The locality of the compiled interaction is
\[
k+\ell+2
=
k+(2k+1)+2
=
3k+3.
\]
Combining the exact compiler identity with the normalization bound gives
\[
\left\|
H-\frac{1}{Q}V^\dagger\widetilde H V
\right\|
=
\left\|
H-\frac{1}{Q}K
\right\|
\le\varepsilon.
\]

Let $R$ denote the number of occurrences in $K$. Since
$\ell=2k+1$, the unpinned compiler introduces $2(\ell+1)=4k+4$
auxiliary qubits and $\ell+1=2k+2$ interaction applications per
occurrence. Hence
\[
n_{\mathrm{out}}
=
n+(4k+4)R,
\]
and the total number of applications is
\[
R(2k+2).
\]
By \cref{lem:universal-finite-alphabet-normalization},
$R=\operatorname{poly}(n)$ under the stated assumptions.  Hence both the final
system size and the number of interaction applications are polynomial
in the source size.

Finally, Weyl's inequality and the exact ground-energy preservation of
the unpinned compiler give
\[
\left|
\lambda_{\min}(H)
-
\frac{1}{Q}
\lambda_{\min}(\widetilde H)
\right|
\leq
\varepsilon.
\]
Therefore, for a Local Hamiltonian instance $(H,a,b)$ with
\[
0<\varepsilon<\frac{b-a}{4},
\]
the compiled instance may use thresholds
\[
\widetilde a
=
Q(a+\varepsilon),
\qquad
\widetilde b
=
Q(b-\varepsilon),
\]
with
\[
\widetilde b-\widetilde a
=
Q\bigl(b-a-2\varepsilon\bigr)
>0.
\]

\end{proof}

\begin{rem}
The two stages of \cref{thm:universal-single-interaction-reduction}
have different preservation properties.  The approximation occurs
only when the arbitrary source Hamiltonian $H$ is replaced by the
finite-alphabet Hamiltonian $K$; the subsequent compilation of $K$ to
a single interaction is exact.

Consequently, exact properties of $K$ preserved by the
finite-alphabet compiler need not be preserved relative to the
original Hamiltonian $H$.  In particular, the normalization step need
not preserve ground-space degeneracy, frustration freeness, or the
exact spectral gap.  Representing coefficients by multiplicity may
increase the degree of working qubits by a polynomial factor, so
bounded degree and geometric locality are not automatically
preserved.  Likewise, a general Pauli expansion need not preserve
stoquasticity or positivity term by term.

Thus the universal construction is a complexity-theoretic reduction
with a known global energy scale $Q$, rather than an exact
structure-preserving simulation of $H$.
\end{rem}

\begin{cor}
\label{cor:universal-sih-qma-completeness}
For every fixed $k\ge 2$, there exists a fixed $(3k+3)$-local
Hermitian interaction $h_k$ such that the corresponding
$\mathrm{SIH}_{3k+3}$ problem is QMA-complete.
In particular, for $k=2$ there exists a fixed $9$-local interaction
giving a QMA-complete $\mathrm{SIH}_9$ problem.
\end{cor}

\begin{proof}
For every fixed $k\ge 2$, $k$-Local Hamiltonian is QMA-hard.
Given an instance with promise gap
\[
b-a\ge \frac{1}{\operatorname{poly}(n)},
\]
choose, for example,
\[
\varepsilon=\frac{b-a}{4}.
\]

Then $\varepsilon$ has a polynomial-bit description and
\[
\frac{1}{\varepsilon}
=
\operatorname{poly}(n),
\]
so \cref{thm:universal-single-interaction-reduction} gives a
polynomial-time promise-preserving reduction to
$\mathrm{SIH}_{3k+3}$.

The output promise gap satisfies
\[
\widetilde b-\widetilde a
=
\frac{Q(b-a)}{2}
\ge
\frac{b-a}{2}
\ge
\frac{1}{\operatorname{poly}(n)}.
\]
Since the output system size satisfies
$n_{\mathrm{out}}=\operatorname{poly}(n)$, this is also an
inverse-polynomial promise gap in $n_{\mathrm{out}}$.

QMA containment follows from
\cref{lem:unpinned-containment}.

\end{proof}

\section{Hardness of Constant Single-Interaction Hamiltonians}
\label{sec:hardness_single_interaction}

This section turns from the general single-interaction normal form of
\cref{sec:general_methodology} to sharper fixed-interaction hardness
results. The universal reduction applies to arbitrary fixed-locality
Hamiltonians, but its normalization is not locality-optimal and need
not preserve bounded degree. We therefore develop two complementary
specialized constructions.

First, we combine the positive-weight singleton unweighting of
\cref{sec:universal-finite-alphabet-normalization} with a QMA-complete
positive-weight two-qubit interaction to obtain an unpinned
$\mathrm{SIH}_3$ problem. We then treat geometric locality separately,
combining a spatially local fixed-gate verifier with local energetic
auxiliary checks to obtain geometrically local $\mathrm{SIH}_8$.
The first construction gives the sharper locality bound, while the
second preserves both bounded interaction range and bounded vertex
degree.

\subsection{The \texorpdfstring{$\mathrm{SIH}_3$}{SIH3} Construction}
\label{subsec:sih3-construction}

We obtain the three-local result by applying the positive-weight
singleton construction of
\cref{cor:positive-weight-singleton-sih} to a QMA-complete
two-qubit interaction.

Consider the antiferromagnetic Heisenberg interaction
\[
h_{\mathrm H}
=
X\otimes X
+
Y\otimes Y
+
Z\otimes Z.
\]
Since $h_{\mathrm H}$ has
\[
\alpha=\beta=\gamma=1,
\]
\cref{thm:piddock-montanaro-positive} implies that
$\{h_{\mathrm H}\}^{+}$-Hamiltonian is QMA-complete.

\begin{thm}
\label{thm:sih3-qma-complete}
There exists a fixed $3$-local Hermitian interaction
$g_{\mathrm H}$ such that the corresponding $\mathrm{SIH}_3$
problem is QMA-complete. In particular, one may take
\[
g_{\mathrm H}
=
h_{\mathrm H}\otimes\Pi^{\Ket{1}}
+
2I\otimes\Pi^{\Ket{0}}.
\]
\end{thm}

\begin{proof}
The Heisenberg interaction satisfies
\[
\lambda_{\max}(h_{\mathrm H})=1,
\]
so the choice $c=2$ satisfies the hypothesis of
\cref{lem:positive-singleton-unweighting}.

Applying \cref{cor:positive-weight-singleton-sih} therefore gives the
polynomial-time promise-preserving reduction
\[
\{h_{\mathrm H}\}^{+}\text{-Hamiltonian}
\longrightarrow
\mathrm{SIH}_3
\]
with fixed interaction
\[
g_{\mathrm H}
=
h_{\mathrm H}\otimes\Pi^{\Ket{1}}
+
2I\otimes\Pi^{\Ket{0}}.
\]
Since the source problem is QMA-complete by
\cref{thm:piddock-montanaro-positive}, this proves QMA-hardness.
Containment in QMA follows from \cref{lem:unpinned-containment}.
\end{proof}

\subsection{The Geometrically Local \texorpdfstring{$\mathrm{SIH}_8$}{SIH8} Construction}
\label{subsec:unpinned_geometric}

We now turn to the single-interaction setting under strict spatial
constraints.

\begin{thm}
\label{thm:unpinned-fixed}
There exists a fixed $8$-local Hermitian interaction $h$ for which
geometrically local $\mathrm{SIH}_8$ is QMA-complete.
\end{thm}

For hardness, we start from the geometrically local $5$-Local
Hamiltonian construction of
\cref{subsubsec: Geometrically Local}.  We first compile the verifier
to a fixed computational gate and then spatialize it using native
SWAP and identity operations.  The resulting geometrically local
Hamiltonian is encoded by one fixed $8$-local interaction, with the
valid auxiliary configuration enforced energetically.

\subsubsection{Fixed-Gate Verifier Compilation}
\label{subsubsec:fixed-gate-verifier-compilation}

Let $L$ be a promise problem in QMA\@. By \cref{def:QMA}, there exists a
deterministic polynomial-time algorithm that, given
$x\in\{0,1\}^*$, outputs a quantum verification circuit $U_x$.
After amplification, we may assume that $U_x$ contains
$T=\operatorname{poly}(|x|)$ two-qubit gates and acts on
\[
N=n_0+n_1+m=\operatorname{poly}(|x|)
\]
qubits, initialized as
\[
\Ket{0}^{\otimes n_0}
\otimes
\Ket{1}^{\otimes n_1}
\otimes
\Ket{\xi},
\]
where $\Ket{\xi}$ is an $m$-qubit witness. For
$\varepsilon=2^{-\Omega(|x|)}$, if $x\in L_{\mathrm{yes}}$ then some
witness is accepted with probability at least $1-\varepsilon$, whereas
if $x\in L_{\mathrm{no}}$ then every witness is accepted with
probability at most $\varepsilon$.

Before spatialization, we compile the verifier so that every
computational gate is an ordered application of a single fixed
two-qubit gate $\mathbb{B}$.  We take $\mathbb{B}$ to be the Barenco
gate~\cite{Barenco1995}, using Barenco's parameter ordering
$\mathbb{B}(\phi,\alpha,\theta)$:
\begin{align*}
\mathbb{B}(\phi,\alpha,\theta)=
\begin{pmatrix}
1 & 0 & 0 & 0\\
0 & 1 & 0 & 0\\
0 & 0 & e^{i\alpha}\cos\theta
& -ie^{i(\alpha-\phi)}\sin\theta \\
0 & 0 & -ie^{i(\alpha+\phi)}\sin\theta
& e^{i\alpha}\cos\theta
\end{pmatrix}.
\end{align*}

We fix
\[
\phi=\frac{\pi}{2},
\qquad
\alpha=\frac{\pi}{4},
\qquad
\theta=\frac{\pi}{\sqrt{2}},
\]
and write
\[
\mathbb{B}
=
\mathbb{B}\left(
\frac{\pi}{2},
\frac{\pi}{4},
\frac{\pi}{\sqrt{2}}
\right).
\]

Our fixed gate $\mathbb{B}$ belongs to the special universal family
$A(\pi/2,\pi/4,\theta)$ identified in~\cite{Barenco1995}; here
$\theta/\pi=1/\sqrt{2}$ is irrational.  Hence ordered applications of
$\mathbb{B}$ are universal for quantum computation.

For efficient compilation, consider a fixed three-qubit register and
the finite gate set obtained by applying $\mathbb{B}$ to its ordered
pairs of qubits.  After phase-normalizing these gates to determinant
one, the inverse-free Solovay--Kitaev theorem
of~\cite{BoulandGiurgicaTiron2021} applies.  Every gate from a fixed
universal verifier gate set can therefore be approximated to error
$\delta$ using $\operatorname{polylog}(1/\delta)$ ordered applications
of $\mathbb{B}$, with the same asymptotic classical compilation time.

Choosing $\delta$ exponentially small in $|x|$ preserves the
exponentially small completeness and soundness error of the amplified
verifier while increasing its size by only a polynomial factor.

\subsubsection{Spatially Local Fixed-Gate Circuit}

We now spatialize the fixed-gate verifier using the construction of
\cref{subsubsec: Geometrically Local} and
\cref{app:geometric-source}.  Spatialization introduces exact SWAP
operations to transport the logical register between neighboring
columns and identity operations to advance the local time cursor.
Thus the resulting polynomial-length spatial circuit
\[
U_1,\ldots,U_T
\]
uses the fixed operation set
\[
\{\mathbb{B},\mathrm{SWAP},I\},
\]
and every physical work qubit participates in only $O(1)$ circuit
operations.

We keep $\mathrm{SWAP}$ and $I$ as native operations rather than
compiling them into sequences of $\mathbb{B}$ gates.  This preserves
the constant participation bound needed for bounded degree in the
circuit-to-Hamiltonian construction.

Finally, by padding with an initial and a final identity round, we may
assume without loss of generality that the first and last computational
rounds consist entirely of identity operations. The explicit lattice
layout, operation count, and work-register bookkeeping are given in
\cref{app:geometric-source}.

\subsubsection{The Single-Interaction Matrix}

We encode $H^{\mathrm{geo}}$ using a fixed $8$-local Hermitian
interaction $h$. Its computational branches implement Barenco,
SWAP, and identity propagation together with the initialization,
output, and clock penalties, while an additional two-qubit check term
enforces the valid auxiliary configuration.

The interaction acts on eight qubits $(v_{1},\ldots,v_{8})$ and is defined by
\begin{align} \label{eq:fixed-h-def}
    h(v_{1},\dots,v_{8}) &= \frac{1}{4} \Big( h_{\mathrm{prop}_1}(v_{1},v_{2},v_{3}) \otimes \Pi_{v_{4}}^{\Ket{0}} \otimes \Pi_{v_{5}}^{\Ket{1}}\otimes \Pi_{v_{6}}^{\Ket{1}} \nonumber \\
    &\quad + h_{\mathrm{prop}_2}(v_{1},v_{2},v_{3}) \otimes \Pi_{v_{4}}^{\Ket{1}} \otimes \Pi_{v_{5}}^{\Ket{1}}\otimes \Pi_{v_{6}}^{\Ket{0}} \nonumber \\
    &\quad + h_{\mathrm{prop}_3}(v_{1},v_{2}) \otimes \Pi_{v_{3}}^{\Ket{1}} \otimes \Pi_{v_{4}}^{\Ket{0}} \otimes \Pi_{v_{5}}^{\Ket{0}}\otimes \Pi_{v_{6}}^{\Ket{0}} \nonumber \\
    &\quad + \Pi_{v_{1}}^{\Ket{1}} \otimes\Pi_{v_{2}}^{\Ket{0}} \otimes\Pi_{v_{3}}^{\Ket{0}} \otimes \Pi_{v_{4}}^{\Ket{0}} \otimes \Pi_{v_{5}}^{\Ket{1}} \otimes \Pi_{v_{6}}^{\Ket{0}}  \nonumber \Big) \otimes I_{v_7,v_8} \\
    &\quad + I_{v_1,\dots,v_6} \otimes h_{\text{check}}(v_7,v_8)
\end{align}

where $h_{\mathrm{prop}_1}$ is the propagation check for a Barenco gate on qubits $v_{1},v_{2}$ with $v_{3}$ acting as the clock transition qubit:
\begin{align} \label{eq:fixed-h-prop-barenco}
    h_{\mathrm{prop}_1}(v_{1},v_{2},v_{3}) = \frac{1}{2}\left(I_{v_1,v_2} \otimes \Ket{1}\Bra{1}_{v_3} + I_{v_1,v_2} \otimes \Ket{0}\Bra{0}_{v_3} - \mathbb{B}_{v_{1},v_{2}} \otimes \Ket{1}\Bra{0}_{v_3} - \mathbb{B}_{v_{1},v_{2}}^{\dagger} \otimes \Ket{0}\Bra{1}_{v_3}\right)
\end{align} 

$h_{\mathrm{prop}_2}$ is the propagation check for a SWAP gate on qubits $v_{1},v_{2}$ with $v_{3}$ acting as the clock transition qubit:
\begin{align} \label{eq:fixed-h-prop-swap}
    h_{\mathrm{prop}_2}(v_{1},v_{2},v_{3}) = \frac{1}{2}\left(I_{v_1,v_2} \otimes \Ket{1}\Bra{1}_{v_3} + I_{v_1,v_2} \otimes \Ket{0}\Bra{0}_{v_3} - \mathrm{SWAP}_{v_{1},v_{2}} \otimes \Ket{1}\Bra{0}_{v_3} - \mathrm{SWAP}_{v_{1},v_{2}}^{\dagger} \otimes \Ket{0}\Bra{1}_{v_3}\right) 
\end{align}

$h_{\mathrm{prop}_3}$ is the propagation check for an Identity gate on qubit $v_{1}$ with $v_{2}$ acting as the clock transition qubit:
\begin{align} \label{eq:fixed-h-prop-identity}
    h_{\mathrm{prop}_3}(v_{1},v_{2}) = \frac{1}{2}\left(I_{v_1} \otimes \Ket{1}\Bra{1}_{v_2} + I_{v_1} \otimes \Ket{0}\Bra{0}_{v_2} - I_{v_{1}} \otimes \Ket{1}\Bra{0}_{v_2} - I_{v_{1}} \otimes \Ket{0}\Bra{1}_{v_2}\right)
\end{align}

and $h_{\text{check}}(v_7,v_8)$ is a local penalty term that introduces an energy cost for any state orthogonal to $\Ket{01}$:
\begin{align} \label{eq:fixed-h-check}
    h_{\text{check}}(v_{7},v_{8}) &= I_{v_7,v_8} - \Pi_{v_{7}}^{\Ket{0}} \otimes \Pi_{v_{8}}^{\Ket{1}} 
\end{align}

The four computational branches of $h$ are designed so that selector
projectors already present in the desired Hamiltonian terms can be shared
with the branch-selection mechanism. Thus the tensor-factor positions of $h$ do not have fixed
\emph{selector} and \emph{computational} roles.

On the valid auxiliary configuration, even-indexed auxiliary qubits are
in $\Ket{0}$ and odd-indexed auxiliary qubits are in $\Ket{1}$. The four
branches are selected as follows.

For Barenco propagation, an odd auxiliary qubit is routed to $v_6$.
Since the first branch contains the factor $\Pi^{\Ket{1}}_{v_6}$ while
each of the other three branches contains
$\Pi^{\Ket{0}}_{v_6}$, this annihilates all branches except the first.
The projectors on $v_4$ and $v_5$ are then left available to implement
the neighboring clock projectors appearing in the propagation term.

For SWAP propagation, an odd auxiliary qubit is routed to $v_4$.
Only the second branch contains $\Pi^{\Ket{1}}_{v_4}$; the remaining
three contain $\Pi^{\Ket{0}}_{v_4}$. Hence the second branch is selected,
while the $v_5$ and $v_6$ positions remain available for the neighboring
clock projectors.

For identity propagation, an even auxiliary qubit is routed to $v_5$.
The first, second, and fourth branches contain
$\Pi^{\Ket{1}}_{v_5}$ and are therefore annihilated, leaving the third
branch. The remaining projector positions in that branch can again be
used as part of the clock constraint.

Finally, the fourth branch is used for the initialization, output, and
clock penalties. In these applications the auxiliary qubits are routed
so that each of the first three branches is annihilated, while some of
the projector positions in the fourth branch are occupied by work or
clock qubits and therefore become part of the desired penalty operator.

This sharing of selector and computational projector positions is what
allows the specialized construction to achieve lower locality than the
general compiler of \cref{sec:general_methodology}.

\subsubsection{Auxiliary-Block Construction and Spatial Embedding}

We use the clock placement of \cref{app:geometric-source}, so that the
clock qubit $C_t$ lies within constant distance of the work qubits
involved in the operation $U_t$.

For each time step $t$, we place in the same local neighborhood a
dedicated auxiliary block
\[
A^{(t)}
=
\left(
A^{(t)}_0,\ldots,A^{(t)}_7
\right),
\]
consisting of four ordered pairs. Its valid state is
\[
\Ket{\Omega_t}_A
=
\Ket{01}_{A^{(t)}_0,A^{(t)}_1}
\Ket{01}_{A^{(t)}_2,A^{(t)}_3}
\Ket{01}_{A^{(t)}_4,A^{(t)}_5}
\Ket{01}_{A^{(t)}_6,A^{(t)}_7}.
\]
The valid state of the complete auxiliary register is therefore
\[
\Ket{\Omega}_A
=
\bigotimes_{t=1}^{T}\Ket{\Omega_t}_A.
\]

The prefactor $1/4$ in the computational part of $h$ is compensated by
using four applications of $h$ for every target term. These applications
are indexed by $j\in\{0,1,2,3\}$. In the valid auxiliary state, all four
applications select the same computational branch of $h$, so each
contributes exactly one quarter of the desired target operator. Their
sum therefore reproduces that operator with coefficient one.

At the same time, the check positions $(v_7,v_8)$ are occupied by
\[
\left(
A^{(t)}_{2j+6},
A^{(t)}_{2j+7}
\right),
\]
where auxiliary indices are interpreted modulo $8$. As $j$ ranges from
$0$ to $3$, these positions contain, respectively,
\[
(A^{(t)}_6,A^{(t)}_7),\quad
(A^{(t)}_0,A^{(t)}_1),\quad
(A^{(t)}_2,A^{(t)}_3),\quad
(A^{(t)}_4,A^{(t)}_5).
\]
Hence the same four applications that restore the coefficient of the
computational term also route every auxiliary pair through
$h_{\mathrm{check}}$ exactly once.

We define the valid auxiliary subspace by
\[
\mathcal{H}_{\mathrm{val}}
=
\mathcal{H}_{W,C}
\otimes
\operatorname{span}\{\Ket{\Omega}_A\},
\]

and denote its orthogonal complement by
\[
\mathcal{H}_{\mathrm{inv}}
=
\mathcal{H}_{\mathrm{val}}^\perp.
\]

Let
\[
V_{\mathrm{val}}:\mathcal{H}_{W,C}\longrightarrow
\mathcal{H}_{\mathrm{val}}
\]
be the isometry
\[
V_{\mathrm{val}}\Ket{\psi}
=
\Ket{\psi}_{W,C}\otimes\Ket{\Omega}_A.
\]

With this local layout and auxiliary structure established, the total
unpinned Hamiltonian is
\begin{align}
    \widetilde{H}^{\mathrm{geo}}
    &=
    \widetilde{H}^{\mathrm{geo}}_{\mathrm{in}}
    +
    \widetilde{H}^{\mathrm{geo}}_{\mathrm{out}}
    +
    \widetilde{H}^{\mathrm{geo}}_{\mathrm{prop}}
    +
    \widetilde{H}^{\mathrm{geo}}_{\mathrm{clock}}.
\end{align}

Let $I_0$ and $I_1$ denote the initialized physical work qubits whose
prescribed states are $\Ket{0}$ and $\Ket{1}$, respectively. For
$W_i\in I_0\cup I_1$, let $t_i$ be its first nonidentity time step,
with $\mathrm{SWAP}$ counted as nonidentity.

By the identity padding and localized-initialization argument of
\cref{app:geometric-source},
\[
2\leq t_i\leq T-1,
\]
and the local clock projector
\[
\Pi^{\Ket{1}}_{C_{t_i-1}}
\otimes
\Pi^{\Ket{0}}_{C_{t_i}}
\]
selects the computational snapshot at which the prescribed initial
state of $W_i$ may be checked.

Each target term of $H^{\mathrm{geo}}$ is implemented by four
applications of $h$.  The appropriate work, clock, and auxiliary
qubits are routed into its tensor-factor positions so that, on the
valid auxiliary state, all four applications select the same
computational branch.  Their computational contributions sum to the
target operator, while the four auxiliary pairs are simultaneously
cycled through the check positions.

The branch assignment is
\[
\begin{array}{c|c}
\text{target term} & \text{branch of }h\\
\hline
\mathbb{B}\text{-propagation} & h_{\mathrm{prop}_1}\\
\mathrm{SWAP}\text{-propagation} & h_{\mathrm{prop}_2}\\
I\text{-propagation} & h_{\mathrm{prop}_3}\\
\text{initialization, output, and clock penalties} & \text{fourth branch}
\end{array}
\]

For initialization terms, the check is placed at the local clock
transition immediately preceding the first nonidentity operation on the
corresponding work qubit, as described above. The output term is placed
at the final clock position, and the clock and propagation terms are
routed within the corresponding constant-radius time-step neighborhood.

The explicit ordered $8$-tuples implementing each of these terms,
including the endpoint propagation terms, are given in
\cref{app:geometric-sih-routing}.

\begin{lemma}[Valid and invalid auxiliary subspaces]
\label{lem:geometric-sih-auxiliary-subspaces}
The restriction of $\widetilde{H}^{\mathrm{geo}}$ to the valid auxiliary
subspace is exactly the target geometrically local Hamiltonian:
\[
V_{\mathrm{val}}^\dagger
\widetilde{H}^{\mathrm{geo}}
V_{\mathrm{val}}
=
H^{\mathrm{geo}}.
\]

Moreover,
\[
\left.
\widetilde{H}^{\mathrm{geo}}
\right|_{\mathcal{H}_{\mathrm{inv}}}
\succeq I_{\mathcal{H}_{\mathrm{inv}}}.
\]
\end{lemma}

\begin{proof}

On the valid auxiliary state $\Ket{\Omega}_A$, the routing construction
of the preceding subsection selects the intended computational branch
for every target term, while all unwanted branches are annihilated by
their selector projectors. The four applications associated with each
target term contribute four copies of one quarter of the desired
operator and simultaneously route all four auxiliary pairs through
$h_{\mathrm{check}}$, whose contribution vanishes on $\Ket{01}$.
Therefore
\[
V_{\mathrm{val}}^\dagger
\widetilde{H}^{\mathrm{geo}}
V_{\mathrm{val}}
=
H^{\mathrm{geo}}.
\]

For the invalid subspace, define the auxiliary penalty operator, extended
trivially over the work and clock registers, by
\[
H_A
=
I_{W,C}
\otimes
\sum_{t=1}^{T}
\sum_{r=0}^{3}
\left(
I-
\Pi^{\Ket{0}}_{A^{(t)}_{2r}}
\otimes
\Pi^{\Ket{1}}_{A^{(t)}_{2r+1}}
\right).
\]

Each $h_{\mathrm{prop}_i}$ is positive semidefinite, as is every
computational-basis projector appearing in \cref{eq:fixed-h-def}.
Together with $h_{\mathrm{check}}\succeq0$, this implies that every
application of $h$ dominates its $h_{\mathrm{check}}$ contribution.

In particular, the four applications implementing the propagation term
at time $t$ route all four pairs of $A^{(t)}$ through
$h_{\mathrm{check}}$. Therefore
\[
\widetilde{H}^{\mathrm{geo}}
\succeq
H_A.
\]

The kernel of $H_A$ is precisely
$\mathcal{H}_{\mathrm{val}}$, while every state orthogonal to that
subspace has $H_A$-energy at least $1$. Equivalently,
\[
\left.H_A\right|_{\mathcal{H}_{\mathrm{inv}}}
\succeq
I_{\mathcal{H}_{\mathrm{inv}}}.
\]
Thus
\[
\left.
\widetilde{H}^{\mathrm{geo}}
\right|_{\mathcal{H}_{\mathrm{inv}}}
\succeq
I_{\mathcal{H}_{\mathrm{inv}}},
\]
as claimed.
\end{proof}

\begin{cor}
\label{cor:geometric-sih-low-energy-spectrum}
The spectrum of $\widetilde{H}^{\mathrm{geo}}$ below energy $1$ is
exactly the spectrum of $H^{\mathrm{geo}}$ below energy $1$, including
multiplicities. In particular,
\[
\lambda_{\min}\!\left(\widetilde{H}^{\mathrm{geo}}\right)
=
\lambda_{\min}\!\left(H^{\mathrm{geo}}\right).
\]
\end{cor}

\begin{proof}
By \cref{lem:geometric-sih-auxiliary-subspaces},
$\widetilde{H}^{\mathrm{geo}}$ decomposes into the invariant valid and
invalid auxiliary subspaces. On the valid subspace it is exactly
$H^{\mathrm{geo}}$, whereas every eigenvalue in the invalid subspace is
at least $1$. Hence the spectra below $1$ coincide.

Moreover, for any witness state, the corresponding
propagation-consistent history state with correctly initialized
computational-basis ancillas is annihilated by the initialization,
propagation, and clock terms.  The output projector acts only at the
final clock time, so
\[
\Bra{\eta}H^{\mathrm{geo}}\Ket{\eta}
\leq
\frac{1}{T+1}
<1.
\]
Hence
\[
\lambda_{\min}\!\left(H^{\mathrm{geo}}\right)<1.
\]
Since every eigenvalue arising from
$\mathcal H_{\mathrm{inv}}$ is at least $1$, the ground energy of
$\widetilde H^{\mathrm{geo}}$ lies in the valid sector, proving
\[
\lambda_{\min}\!\left(\widetilde H^{\mathrm{geo}}\right)
=
\lambda_{\min}\!\left(H^{\mathrm{geo}}\right).
\]

\end{proof}

\subsubsection{Correctness and QMA-Completeness}

\begin{proof}[Proof of \cref{thm:unpinned-fixed}]
By \cref{lem:geometric-sih-auxiliary-subspaces}, the invariant valid
auxiliary subspace is unitarily equivalent, via $V_{\mathrm{val}}$, to
the target Hamiltonian $H^{\mathrm{geo}}$.

If $x\in L_{\mathrm{yes}}$, then by
\cref{lemma: OT completness} there exists a history state
$\Ket{\eta}$ such that
\[
\bra{\eta}H^{\mathrm{geo}}\Ket{\eta}
\leq
\frac{\varepsilon}{T+1}.
\]
Therefore
\[
\Ket{\eta}\otimes\Ket{\Omega}_A
\]
has the same energy under $\widetilde{H}^{\mathrm{geo}}$.

If $x\in L_{\mathrm{no}}$, then every state in the valid auxiliary
subspace has energy at least
\[
\frac{c}{T^3}
\left(
1-\varepsilon-\sqrt{\varepsilon}
\right)
\]
by \cref{lemma: OT soundness}. Every state in the invalid auxiliary subspace, on the other hand, has
energy at least $1$ by
\cref{lem:geometric-sih-auxiliary-subspaces}.

Hence
\[
\lambda_{\min}(\widetilde{H}^{\mathrm{geo}})
\geq
\min\left\{
\frac{c}{T^3}
\left(
1-\varepsilon-\sqrt{\varepsilon}
\right),
1
\right\}.
\]
Since $T=\operatorname{poly}(|x|)$ and
$\varepsilon=2^{-\Omega(|x|)}$, this lower bound is inverse polynomial
in $|x|$.

Thus the reduction preserves an inverse-polynomial promise gap and is
QMA-hard.

Geometric locality follows from the local routing construction.  Each
constant-size auxiliary block $A^{(t)}$ is placed within constant
distance of the work and clock qubits used at time $t$, so every
application of $h$ has constant geometric diameter.  Since each source
term is replaced by only four applications of $h$, the source
Hamiltonian has bounded degree, and every auxiliary block participates
in only $O(1)$ applications, the compiled Hamiltonian also has bounded
degree.

The spatial circuit and $H^{\mathrm{geo}}$ have polynomial size, so
the number of physical qubits and interaction applications remains
polynomial in $|x|$.  The explicit routing in
\cref{app:geometric-sih-routing} verifies that all ordered
$8$-tuples are pairwise distinct and contain eight distinct physical
qubits.

Finally, $h$ is a fixed $8$-qubit Hermitian matrix, independent of the
input and system size, with efficiently computable entries.  Thus the
reduction produces a valid geometrically local $\mathrm{SIH}_8$
instance and proves QMA-hardness.  Containment in QMA follows from
\cref{lem:unpinned-containment}.
\end{proof}

\section{Hardness of \texorpdfstring{$n$}{n}-Dependent Single-Interaction Hamiltonians}
\label{sec:n-dependent-hardness}

This section proves QMA-completeness results for $n$-dependent
single-interaction Hamiltonians, where the common interaction may depend
uniformly on the final physical system size $n$.  We begin with the
QMA-complete XY Hamiltonian problem at fixed Hamming weight and replace
its global sector restriction by a $2$-local energetic penalty.  We
then encode the resulting Hamiltonian first as a $2$-Pinned
$\mathrm{DSIH}_3$ instance and subsequently, using the exact unpinned
compiler, as an unpinned $\mathrm{DSIH}_5$ instance.

Throughout this section, $N$ denotes the number of graph qubits in the
source XY instance, whereas $n$ denotes the total number of physical
qubits in the final DSIH instance, including auxiliary, pinned, and
dummy qubits.

\subsection{Localizing the XY Hamming-Weight Constraint}
\label{subsec:xy-local-penalty}

We begin with the QMA-complete XY Hamiltonian problem at fixed Hamming
weight from \cref{def: XY Hamiltonian problem,thm: XY hamiltonian QMA complete}. Let
\[
G=(V,E),
\qquad
N=|V|,
\]
and let $k$ be the specified Hamming weight. Recall that
\[
H_{\mathrm{XY}}
=
\sum_{(u,v)\in E}
\left(
\Ket{01}\bra{10}_{u,v}
+
\Ket{10}\bra{01}_{u,v}
\right),
\]
and that
\[
[H_{\mathrm{XY}},\widehat N]=0,
\qquad
\widehat N
=
\sum_{u\in V}\Pi^{\Ket{1}}_u.
\]
Hence the Hilbert space decomposes into invariant Hamming-weight sectors
$\mathcal{S}_w$.

We replace the global restriction to $\mathcal{S}_k$ by an energetic
penalty that can be written using only $2$-local terms. Define
\begin{equation}
\label{eq:Hint}
H_{\mathrm{int}}
=
\frac{H_{\mathrm{XY}}+|E|I}{2}
+
\Lambda
\left(
\widehat N-kI
\right)^2,
\end{equation}
where
\begin{equation}
\label{eq:xy-penalty-strength}
\Lambda=2|E|.
\end{equation}

Each local XY interaction has operator norm $1$, and therefore
\[
\|H_{\mathrm{XY}}\|
\leq |E|.
\]
Consequently,
\[
\frac{H_{\mathrm{XY}}+|E|I}{2}\succeq0.
\]

To express the penalty using at most $2$-local terms, observe that
\[
\widehat N^2
=
\sum_u \Pi^{\Ket{1}}_u
+
\sum_{u\neq v}
\Pi^{\Ket{1}}_u\otimes\Pi^{\Ket{1}}_v,
\]
where throughout this section $ \sum_{u\neq v} $ denotes a sum over \emph{ordered} pairs of distinct vertices. Hence
\begin{equation}
\label{eq:xy-local-penalty}
\begin{aligned}
H_{\mathrm{penalty}}
&=
\Lambda(\widehat N-kI)^2\\
&=
\Lambda\left(
\sum_{u\neq v}
\Pi^{\Ket{1}}_u\otimes\Pi^{\Ket{1}}_v
+
(1-2k)
\sum_u\Pi^{\Ket{1}}_u
+
k^2I
\right).
\end{aligned}
\end{equation}
Thus $H_{\mathrm{int}}$ is $2$-local.

The shift and rescaling of the XY term induce the promise thresholds
\begin{equation}
\label{eq:xy-shifted-thresholds}
a'
=
\frac{a+|E|}{2},
\qquad
b'
=
\frac{b+|E|}{2}.
\end{equation}
In particular,
\[
b'-a'
=
\frac{b-a}{2},
\]
so the inverse-polynomial promise gap is preserved up to a constant factor.

\begin{lemma}[Local penalty for the Hamming-weight constraint]
\label{lem:xy-local-penalty}
The Local Hamiltonian problem defined by
$H_{\mathrm{int}}$ in \cref{eq:Hint}, with thresholds
$(a',b')$ from \cref{eq:xy-shifted-thresholds}, is QMA-complete.
\end{lemma}

\begin{proof}
Containment in QMA follows because $H_{\mathrm{int}}$ is a $2$-local
Hamiltonian with polynomially bounded coefficients and an
inverse-polynomial promise gap.

For QMA-hardness, we reduce from the XY Hamiltonian problem of
\cref{thm: XY hamiltonian QMA complete}. Since
\[
[H_{\mathrm{XY}},\widehat N]=0,
\]
the Hamiltonian $H_{\mathrm{int}}$ is block diagonal with respect to the
Hamming-weight sectors $\mathcal{S}_w$.

Suppose first that the XY instance is a YES instance. Then there exists
$\Ket{\psi}\in\mathcal{S}_k$ such that
\[
\bra{\psi}H_{\mathrm{XY}}\Ket{\psi}\leq a.
\]
On $\mathcal{S}_k$ the penalty vanishes, and therefore
\[
\bra{\psi}H_{\mathrm{int}}\Ket{\psi}
=
\frac{
\bra{\psi}H_{\mathrm{XY}}\Ket{\psi}+|E|
}{2}
\leq
\frac{a+|E|}{2}
=
a'.
\]
Hence the transformed instance is a YES instance.

Now suppose that the XY instance is a NO instance. Consider first a state
$\Ket{\psi}\in\mathcal{S}_k$. The penalty again vanishes, while the
NO-instance promise gives
\[
\bra{\psi}H_{\mathrm{XY}}\Ket{\psi}\geq b.
\]
Thus
\[
\bra{\psi}H_{\mathrm{int}}\Ket{\psi}
\geq
\frac{b+|E|}{2}
=
b'.
\]

Next consider a Hamming-weight sector $\mathcal{S}_w$ with $w\neq k$.
Since $w$ and $k$ are integers,
\[
(w-k)^2\geq1.
\]
Moreover,
\[
\frac{H_{\mathrm{XY}}+|E|I}{2}\succeq0.
\]
Therefore every normalized
$\Ket{\psi}\in\mathcal{S}_w$ satisfies
\[
\bra{\psi}H_{\mathrm{int}}\Ket{\psi}
\geq
\Lambda
=
2|E|.
\]

For a NO instance we necessarily have
\[
b\leq |E|,
\]
because $\|H_{\mathrm{XY}}\|\leq|E|$ and the spectrum of the restriction
of $H_{\mathrm{XY}}$ to $\mathcal{S}_k$ is contained in
$[-|E|,|E|]$. Hence
\[
b'
=
\frac{b+|E|}{2}
\leq
|E|
\leq
2|E|
=
\Lambda.
\]
Thus every state outside the valid Hamming-weight sector also has energy
at least $b'$.

Since $H_{\mathrm{int}}$ is block diagonal in the Hamming-weight
decomposition, the same lower bound holds for arbitrary superpositions
of different sectors. Therefore every NO instance has ground-state energy
at least $b'$, completing the reduction.
\end{proof}

\subsection{The \texorpdfstring{$2$}{2}-Pinned \texorpdfstring{$\mathrm{DSIH}_3$}{DSIH3} Construction}
\label{subsec: pinned n-dependent fixed}

The pinned construction is the $m=2$, $k=2$ binary specialization of
the compiler in \cref{thm:pinned-resource-bounds}, applied pointwise to
the $n$-dependent two-element interaction alphabet underlying
$H_{\mathrm{int}}$.  It therefore requires two pinned selector qubits
and has locality
\[
2+\left\lceil\log_2 2\right\rceil=3.
\]
We give the resulting interaction explicitly because the additional
issue in the DSIH setting is uniformity: every coefficient appearing
in the interaction must be recoverable from the \emph{final} physical
system size $n$.

\begin{thm}
\label{thm:pinned-n-dependent}
There exists a uniform family of $3$-local Hermitian matrices
$\{h^{(n)}\}_{n\geq1}$ for which the
$2$-Pinned $\mathrm{DSIH}_3$ problem is QMA-complete.
\end{thm}

\subsubsection{Interaction Family and System-Size Encoding}
\label{subsec: pinned n-dep-family}

The intermediate Hamiltonian $H_{\mathrm{int}}$ of \cref{lem:xy-local-penalty} is built from two $2$-local interaction types. For a system size $n$ encoding parameters $(N,\Lambda,k)$, define
\[
g_0^{(n)}
=
\frac12\bigl(
\ket{01}\!\bra{10}
+
\ket{10}\!\bra{01}
+
I
\bigr),
\]
and
\[
g_1^{(n)}
=
\Lambda\left(
\Pi_{\ket{1}}\otimes\Pi_{\ket{1}}
+
\frac{1-2k}{N-1}\Pi_{\ket{1}}\otimes I
+
\frac{k^2}{N(N-1)}I
\right).
\]
Then
\[
H_{\mathrm{int}}
=
\sum_{(u,v)\in E} g^{(n)}_{0,u,v}
+
\sum_{u\neq v} g^{(n)}_{1,u,v}.
\]

We apply the binary pinned compiler of
\cref{thm:pinned-resource-bounds} pointwise to the two-element alphabet
\[
\mathcal A^{(n)}
=
\{g_0^{(n)},g_1^{(n)}\}.
\]
Since $|\mathcal A^{(n)}|=2$, the selector length is
\[
\ell=\lceil\log_2 2\rceil=1.
\]
Introduce a two-qubit pinned register
\[
A=(A_1,A_2),
\qquad
\ket{\psi_{\mathrm{pin}}}_A=\ket{01}_{A_1A_2}.
\]
Routing $A_1$ to the selector position selects $g_0^{(n)}$, while routing
$A_2$ selects $g_1^{(n)}$. Hence the resulting three-local interaction is

\begin{align}\label{eq:n-fixed-h-def}
    h^{(n)}(v_1,v_2,v_3)
    =
    g^{(n)}_{0,v_1,v_2}\otimes\Pi^{\ket{0}}_{v_3}
    +
    g^{(n)}_{1,v_1,v_2}\otimes\Pi^{\ket{1}}_{v_3}.
\end{align}

Equivalently,
\begin{align*}
    h^{(n)}(v_1,v_2,v_3)
    = &
    \frac12
    \left(
    \ket{01}\!\bra{10}_{v_1,v_2}
    +
    \ket{10}\!\bra{01}_{v_1,v_2}
    +
    I_{v_1,v_2}
    \right)
    \otimes\Pi^{\ket{0}}_{v_3}\\
    +&
    \Lambda\left(
    \Pi_{\ket{1}}{}_{v_1}\otimes\Pi_{\ket{1}}{}_{v_2}
    +
    \frac{1-2k}{N-1}
    \Pi_{\ket{1}}{}_{v_1}\otimes I_{v_2}
    +
    \frac{k^2}{N(N-1)}I_{v_1,v_2}
    \right)
    \otimes\Pi^{\ket{1}}_{v_3}.
\end{align*}

The matrix $h^{(n)}$ depends on the source parameters
$N$, $\Lambda=2|E|$, and $k$.  For this to define a valid DSIH family,
these parameters must be recoverable uniformly from the final physical
system size $n$, rather than supplied independently with the instance. Since the source graph is simple,
\[
0\leq k\leq N,
\qquad
\Lambda=2|E|\leq N(N-1).
\]

We therefore encode $(N,\Lambda,k)$ into the final system size by setting

\[
n=N^4+\Lambda(N+1)+k+2.
\]

\begin{lemma}[Uniform decoding of the system-size encoding]
\label{lem:DSIH-size-decoding}
Let
\[
n=N^4+\Lambda(N+1)+k+2,
\]
where $N\geq2$, $0\leq k\leq N$, and
\[
\Lambda=2|E|\leq N(N-1).
\]
Then the tuple $(N,\Lambda,k)$ is uniquely determined by $n$ and can be
computed in time polynomial in $n$.
\end{lemma}

The proof is given in \cref{app:dsih-size-decoding}.

For an arbitrary system size $n$, apply the decoding procedure of
\cref{lem:DSIH-size-decoding}. If $n$ does not encode a valid tuple
$(N,\Lambda,k)$ satisfying the conditions of that lemma, define
\[
h^{(n)}=0.
\]
Otherwise define $h^{(n)}$ by \cref{eq:n-fixed-h-def} using the decoded
parameters.

By \cref{lem:DSIH-size-decoding}, the parameters are recoverable from
$n$ in polynomial time. Since $h^{(n)}$ acts on only three qubits, its
entries can then be computed to additive error at most $2^{-q}$ in time
polynomial in $n$ and $q$, and all coefficients in
\cref{eq:n-fixed-h-def} are polynomially bounded in $n$. Hence
$\{h^{(n)}\}_{n\geq1}$ is a uniform family in the sense of
\cref{def: DSIH}.

The active construction uses only the $N$ graph qubits together with the
two pinned qubits. We therefore add
\[
n_{\mathrm{dummy}}=n-(N+2)
\]
isolated qubits so that the total physical system size is exactly the
encoded value $n$. No interaction tuple contains a dummy qubit.  The dummy register
therefore changes only spectral multiplicities: it leaves the set of
eigenvalues, and in particular the ground-state energy and promise
gap, unchanged.

\subsubsection{Hamiltonian Construction}

Fix an arbitrary deterministic orientation of every edge
$\{u,v\}\in E$ and denote the resulting ordered edge by $(u,v)$.
Because the $XY$ interaction is symmetric under exchanging its two
qubits, this choice does not change $H_{\mathrm{XY}}$.

The compiled Hamiltonian is
\[
H^{(n)}
=
H^{(n)}_{\mathrm{XY}}
+
H^{(n)}_{\mathrm{penalty}}.
\]

\paragraph{XY terms.}
For every oriented graph edge $(u,v)\in E$, route the pinned qubit
$A_1$ into the selector slot:
\[
H^{(n)}_{\mathrm{XY}}
=
\sum_{(u,v)\in E}
h^{(n)}(u,v,A_1).
\]

\paragraph{Penalty terms.}
For every ordered pair of distinct graph qubits, route $A_2$ into the
selector slot:
\[
H^{(n)}_{\mathrm{penalty}}
=
\sum_{u\neq v}
h^{(n)}(u,v,A_2).
\]

Every interaction tuple used in the construction contains three pairwise
distinct physical qubits. The $XY$ terms use tuples $(u,v,A_1)$ with
$u\neq v$, while the penalty terms use tuples $(u,v,A_2)$ with $u\neq v$.
Within each family the tuples are pairwise distinct because their ordered
graph-qubit pairs are distinct, and an $XY$ tuple cannot coincide with a
penalty tuple because their third entries are respectively $A_1$ and
$A_2$. Hence the construction satisfies the ordered-tuple requirements
of \cref{def: DSIH}.

\subsubsection{Correctness and QMA-Completeness}
\label{subsubsec: pinned n dependent fixed comp sound}

\begin{proof}[Proof of \cref{thm:pinned-n-dependent}]
Fix a valid encoded system size $n$ and let $(N,\Lambda,k)$ be the
parameters decoded according to \cref{lem:DSIH-size-decoding}.
The two interactions $g_0^{(n)}$ and $g_1^{(n)}$ form the alphabet
$\mathcal A^{(n)}$ used in the binary pinned compiler above.

The target interaction occurrences satisfy the distinctness hypothesis
of \cref{thm:pinned-resource-bounds}, as verified in the preceding
subsection. Applying that theorem for this fixed $n$, and using the fact
that the added dummy qubits are isolated, gives
\[
\lambda_{\min}\!\left(
V_{\mathrm{pin}}^\dagger H^{(n)}V_{\mathrm{pin}}
\right)
=
\lambda_{\min}(H_{\mathrm{int}}).
\]
Thus the thresholds $(a',b')$ of \cref{lem:xy-local-penalty}
are preserved exactly.
QMA-hardness therefore follows from that lemma.

The uniformity analysis above established that the interaction family
$\{h^{(n)}\}_{n\geq1}$ is uniformly computable from the final system
size and that the construction uses exactly $n$ physical qubits.
Hence the reduction produces valid $2$-Pinned DSIH$_3$ instances.
Containment in QMA follows from \cref{lem:pinned-containment}.
Therefore $2$-Pinned DSIH$_3$ is QMA-complete.
\end{proof}

\subsection{The Unpinned \texorpdfstring{$\mathrm{DSIH}_5$}{DSIH5} Construction}
\label{subsec:unpinned-dsih5}

We now remove the pinning resource by applying the exact unpinned finite-alphabet compiler of \cref{thm:unpinned-spectral-preservation} pointwise to the two interaction types appearing in $H_{\mathrm{int}}$.

\begin{thm}
\label{thm:unpinned-n-dependent}
There exists a uniform family of $5$-local Hermitian matrices $\{\widehat h^{(n)}\}_{n\geq1}$ such that $\mathrm{DSIH}_5$ is QMA-complete.
\end{thm}

\subsubsection{Compiler Instantiation and Uniformity}

We reduce from the QMA-complete family of intermediate Hamiltonians
$H_{\mathrm{int}}$ established in \cref{lem:xy-local-penalty}.
For a valid system-size encoding $n$, let $(N,\Lambda,k)$ be the parameters
decoded according to \cref{lem:DSIH-size-decoding}.

The Hamiltonian $H_{\mathrm{int}}$ is built from two $2$-local interaction
types. Define
\[
g^{(n)}_0
=
\frac{1}{2}
\left(
\Ket{01}\bra{10}
+
\Ket{10}\bra{01}
+
I
\right),
\]
and
\[
g^{(n)}_1
=
\Lambda
\left(
\Pi^{\Ket{1}}\otimes\Pi^{\Ket{1}}
+
\frac{1-2k}{N-1}
\Pi^{\Ket{1}}\otimes I
+
\frac{k^2}{N(N-1)}I
\right).
\]
Then
\[
H_{\mathrm{int}}
=
\sum_{(u,v)\in E}
g^{(n)}_{0,u,v}
+
\sum_{u\neq v}
g^{(n)}_{1,u,v}.
\]

We now instantiate the unpinned single-interaction compiler on the
two-element alphabet

\[
\mathcal A^{(n)}
=
\left\{
g_0^{(n)},g_1^{(n)}
\right\}.
\]

Since the alphabet has two elements, the binary compiler has selector
length
\[
\ell=\left\lceil\log_2 2\right\rceil=1.
\]
With separation parameter $\gamma=1$, set
\[
J_n
=
\max\left\{
\|g_0^{(n)}\|,
\|g_1^{(n)}\|
\right\},
\qquad
\Delta_n=2J_n+1.
\]
The resulting $5$-local interaction is

\begin{align}
\widehat h^{(n)}
(v_1,v_2,u,c_1,c_2)
&=
\frac{1}{2}
\left(
g^{(n)}_{0,v_1,v_2}\otimes\Pi^{\Ket{0}}_u
+
g^{(n)}_{1,v_1,v_2}\otimes\Pi^{\Ket{1}}_u
\right)
\otimes I_{c_1,c_2}
\nonumber\\
&\quad
+
\Delta_n
I_{v_1,v_2,u}
\otimes
\left(
I_{c_1,c_2}
-
\Pi^{\Ket{0}}_{c_1}
\otimes
\Pi^{\Ket{1}}_{c_2}
\right).
\label{eq:n-fixed-h5-def}
\end{align}

\begin{lemma}[Compatibility with the final system size]
\label{lem:unpinned-dsih-system-size}
For every valid encoded instance, the unpinned compiler construction can be
realized on exactly the same total system size
\[
n=N^4+\Lambda(N+1)+k+2
\]
used in \cref{lem:DSIH-size-decoding}.
\end{lemma}

The proof is given in \cref{app:unpinned-dsih-system-size}.

The intermediate Hamiltonian contains
\[
M
=
|E|+N(N-1)
=
O(N^2)
\]
interaction occurrences.  Since the binary compiler uses
$\ell+1=2$ applications of $\widehat h^{(n)}$ per occurrence, the
compiled Hamiltonian contains
\[
2M
=
2|E|+2N(N-1)
=
O(N^2)
\]
applications.
By \cref{lem:unpinned-dsih-system-size}, these interactions and their
dedicated ancillas fit within a system of exactly $n$ physical qubits.

For arbitrary $n$, define the family
$\{\widehat h^{(n)}\}_{n\geq1}$ uniformly using the decoding procedure
of \cref{lem:DSIH-size-decoding}.  If $n$ does not encode valid
parameters, set
\[
\widehat h^{(n)}=0.
\]
Otherwise compute $g_0^{(n)}$, $g_1^{(n)}$, $J_n$, and $\Delta_n$ from
the decoded tuple $(N,\Lambda,k)$ and define $\widehat h^{(n)}$ by
\cref{eq:n-fixed-h5-def}.

Both $g^{(n)}_0$ and $g^{(n)}_1$ act on only two qubits, and
$g^{(n)}_1$ is diagonal.  Their operator norms, and hence $\Delta_n,$
can therefore be computed to the requested precision in time
polynomial in $n$ and the precision parameter $q$.

All coefficients are polynomially bounded in $n$, so
$\{\widehat h^{(n)}\}_{n\geq1}$ is a uniform family in the sense of
\cref{def: DSIH}.

\subsubsection{Correctness and QMA-Completeness}

\begin{proof}[Proof of \cref{thm:unpinned-n-dependent}]
For each valid encoded system size $n$, the matrices
$g_0^{(n)}$ and $g_1^{(n)}$ form a fixed two-element interaction alphabet.
We may therefore apply \cref{thm:unpinned-spectral-preservation} separately
for each fixed $n$. This gives the exact spectral-preservation statement
for the corresponding compiled Hamiltonian. The preceding subsection
shows moreover that the compiled interaction $\widehat h^{(n)}$ can be
computed uniformly from $n$ in polynomial time.

Since
\[
\Delta_n=2J_n+1,
\]
the compiler theorem applies with
\[
\gamma=1.
\]
Therefore
\[
\lambda_{\min}\!\left(\widehat H^{(n)}\right)
=
\lambda_{\min}(H_{\mathrm{int}}),
\]
and the ground space is preserved exactly in the valid auxiliary sector.
In fact, the complete spectrum below
\[
\lambda_{\min}(H_{\mathrm{int}})+1
\]
is preserved with multiplicity.

By \cref{lem:xy-local-penalty}, the Local Hamiltonian problem defined by
$H_{\mathrm{int}}$ with thresholds $(a',b')$ is QMA-hard. Since the
compiler preserves the ground-state energy exactly, the same thresholds
$(a',b')$ define the YES and NO cases of the resulting
$\mathrm{DSIH}_5$ instance.

The reduction has polynomial overhead. By
\cref{lem:unpinned-dsih-system-size}, $n=O(N^4),$ and the number of applications of the $5$-local interaction is $2M=O(N^2).$
The family $\{\widehat h^{(n)}\}$ is uniformly computable from the final
system size $n$, as shown above.

Because every target term receives its own four-qubit auxiliary block, applications associated with distinct target terms differ in their auxiliary entries, even when the underlying work-qubit support is the same. For a fixed target term, the two compiler applications use different auxiliary pairs in the check positions and are therefore distinct. Each individual \(5\)-tuple consists of two distinct graph qubits and three distinct qubits from the dedicated auxiliary block. Hence all output interaction tuples satisfy the requirements of \cref{def: DSIH}.

Thus the reduction is polynomial time and promise preserving, proving
QMA-hardness of $\mathrm{DSIH}_5$. QMA containment follows from
\cref{lem:unpinned-containment}. Therefore $\mathrm{DSIH}_5$ is
QMA-complete.

\end{proof}

\section{Stoquastic Single-Interaction Hamiltonians}
\label{sec:stoquastic-embedding-results}

This section develops two complementary stoquastic
single-interaction constructions.  First, applying the positive-weight
singleton unweighting directly to a StoqMA-complete stoquastic
two-qubit interaction yields an
$\mathrm{Stoq\text{-}SIH}_3$ problem.  We then formulate
realification and the stoquastic sign embedding within the
single-interaction framework.  The general complex construction
increases locality by two and introduces one pinned qubit, while a
real source interaction requires only the sign embedding and therefore
increases locality by one.  Applying this sharper real-source
construction to the QMA-complete $\mathrm{SIH}_3$ interaction yields
QMA-completeness of $1$-Pinned $\mathrm{Stoq\text{-}SIH}_4$.

\subsection{The \texorpdfstring{$\mathrm{Stoq\text{-}SIH}_3$}{Stoq-SIH3} Construction}
\label{subsec:stoq-sih3-construction}

We first obtain an unpinned three-local stoquastic hardness result from
a positive-weight two-qubit source.  Consider
\[
h_{\mathrm S}
=
-X\otimes X+Y\otimes Y+2Z\otimes Z.
\]
For the parametrization
\[
h_{\mathrm S}
=
\alpha X\otimes X
+
\beta Y\otimes Y
+
\gamma Z\otimes Z,
\]
we have
\[
\alpha=-1,
\qquad
\beta=1,
\qquad
\gamma=2.
\]
Hence
\[
\alpha=-\beta\neq 0,
\qquad
\alpha+\gamma=1>0,
\qquad
\beta+\gamma=3>0.
\]
By \cref{thm:piddock-montanaro-positive},
the positive-weight singleton problem
$\{h_{\mathrm S}\}^{+}$-Hamiltonian is StoqMA-complete.

\begin{thm}
\label{thm:stoq-sih3-stoqma-complete}
There exists a fixed $3$-local stoquastic Hermitian interaction
$g_{\mathrm S}$ such that the corresponding
$\mathrm{Stoq\text{-}SIH}_3$ problem is StoqMA-complete. In particular,
one may take
\[
g_{\mathrm S}
=
h_{\mathrm S}\otimes\Pi^{\Ket{1}}
+
5I\otimes\Pi^{\Ket{0}}.
\]
\end{thm}

\begin{proof}
The interaction $h_{\mathrm S}$ has
\[
\lambda_{\max}(h_{\mathrm S})=4,
\]
so the choice $c=5$ satisfies the hypothesis of
\cref{lem:positive-singleton-unweighting}.  Since
$\{h_{\mathrm S}\}^{+}$-Hamiltonian is StoqMA-complete by
\cref{thm:piddock-montanaro-positive},
\cref{cor:positive-weight-singleton-sih} gives the
promise-preserving reduction
\[
\{h_{\mathrm S}\}^{+}\text{-Hamiltonian}
\longrightarrow
\mathrm{Stoq\text{-}SIH}_3
\]
with fixed interaction
\[
g_{\mathrm S}
=
h_{\mathrm S}\otimes\Pi^{\Ket{1}}
+
5I\otimes\Pi^{\Ket{0}}.
\]
Because $h_{\mathrm S}$ is stoquastic,
\cref{lem:positive-singleton-unweighting} implies that
$g_{\mathrm S}$ is stoquastic as well.  This proves StoqMA-hardness,
while containment follows from
\cref{lem:stoquastic-containment}.
\end{proof}

\subsection{The Stoquastic Single-Interaction Embedding}

We now lift the realification and stoquastic sign embedding of
\cref{sec: Embedding Hamiltonian} to the single-interaction framework.
For a real matrix $R$, write
\[
R=R^{-}+R^{+},
\]
where $R^{-}$ contains the diagonal and non-positive off-diagonal
entries of $R$, and $R^{+}$ contains its strictly positive
off-diagonal entries.

For a $k$-local Hermitian interaction $h$, let
\[
R(h):=\Phi(h)
\]
and define
\begin{equation}
\label{eq:single-interaction-stoq-embedding}
\mathcal{S}(h)
=
R(h)^{-}\otimes I_c
-
R(h)^{+}\otimes X_c.
\end{equation}
The interaction $\mathcal S(h)$ is $(k+2)$-local: one additional
qubit is introduced by realification and one by the sign embedding.

\begin{thm}[Stoquastic embedding of a single-interaction problem]
\label{thm:stoquastic-single-interaction-embedding}
Let $H$ be an instance of a $p$-Pinned $\mathrm{SIH}_k$ problem whose
pinned state $\Ket{\psi_{\mathrm{pin}}}$ has real amplitudes in the
computational basis; $p=0$ denotes the unpinned case.

There is a polynomial-time reduction to a
$(p+1)$-Pinned $\mathrm{Stoq\text{-}SIH}_{k+2}$ instance.  The
original pinned register remains fixed to
$\Ket{\psi_{\mathrm{pin}}}$, the new sign qubit is pinned to
$\Ket{-}$, and the realification qubit remains unpinned.  The
restricted spectrum of the output Hamiltonian is exactly the spectrum
of the original pinned problem, with every eigenvalue appearing twice.

The same statement holds for $\mathrm{DSIH}_k$, yielding
$(p+1)$-Pinned $\mathrm{Stoq\text{-}DSIH}_{k+2}$.
\end{thm}

\begin{proof}
Let the original Hamiltonian be
\[
H=\sum_{e\in L}h_e.
\]
Introduce one global phase qubit $b$ and one global sign qubit $c$.
For every ordered interaction tuple $e$, append $b$ and $c$ and apply
the same interaction $\mathcal{S}(h)$ from
\cref{eq:single-interaction-stoq-embedding}. Thus
\[
H_{\mathrm{stoq}}
=
\sum_{e\in L}
\mathcal{S}(h)_{e,b,c}.
\]

Because $b$ and $c$ are new physical qubits, every extended tuple
contains pairwise distinct entries.  Distinct source tuples remain
distinct after appending the same ordered pair $(b,c)$, so the
construction satisfies the ordered-hyperedge requirements of
\cref{def: SIH}; the same argument applies to DSIH.

By construction, every off-diagonal matrix element of $S(h)$ is
non-positive, so $H_{\mathrm{stoq}}$ is stoquastic.

We now identify the restricted Hamiltonian in two compression steps.

\paragraph{First compression: the sign qubit.}
Let
\[
V_-:\mathcal H_{W,A,b}\longrightarrow
\mathcal H_{W,A,b,c},
\qquad
V_-\ket{\phi}
=
\ket{\phi}\otimes\ket{-}_c .
\]
Since
\[
X_c\ket{-}_c=-\ket{-}_c,
\]
compression of a single interaction occurrence gives
\[
V_-^\dagger
S(h)_{e,b,c}
V_-
=
R(h)^-_{e,b}+R(h)^+_{e,b}
=
R(h)_{e,b}
=
\Phi(h_e).
\]
Summing over all interaction occurrences therefore yields
\[
V_-^\dagger H_{\mathrm{stoq}}V_-
=
\sum_{e\in L}\Phi(h_e)
=
\Phi(H),
\]
where the final equality follows from the linearity of $\Phi$.
Thus pinning the new sign qubit to $\ket{-}$ removes the stoquastic
sign embedding exactly and leaves the realified source Hamiltonian.

\paragraph{Second compression: the original pinned register.}
Let $A$ denote the original pinned register, taken to be trivial when
$p=0$, and let $W$ denote the unpinned source register. Define
\[
Q=I_W\otimes\bra{\psi_{\mathrm{pin}}}_A,
\qquad
H_{\mathrm{eff}}=QHQ^\dagger,
\]
with $Q=I_W$ in the unpinned case.
Since $Q$ and $Q^\dagger$ are real matrices,
\[
Q\,\operatorname{Re}(H)\,Q^\dagger
=
\operatorname{Re}(H_{\mathrm{eff}})
\]
and
\[
Q\,\operatorname{Im}(H)\,Q^\dagger
=
\operatorname{Im}(H_{\mathrm{eff}}).
\]
Consequently,
\[
(Q\otimes I_b)\Phi(H)(Q^\dagger\otimes I_b)
=
\operatorname{Re}(H_{\mathrm{eff}})\otimes I_b
+
\operatorname{Im}(H_{\mathrm{eff}})\otimes J_b
=
\Phi(H_{\mathrm{eff}}).
\]

Combining the two compression steps gives
\[
(Q\otimes I_b)
V_-^\dagger H_{\mathrm{stoq}}V_-
(Q^\dagger\otimes I_b)
=
\Phi(H_{\mathrm{eff}}).
\]
Hence, after pinning both the original pinned register to
$\ket{\psi_{\mathrm{pin}}}$ and the new sign qubit to $\ket{-}$, while
leaving the realification qubit $b$ unpinned, the restricted output
Hamiltonian is exactly $\Phi(H_{\mathrm{eff}})$.

By \cref{lem:realification-spectrum}, its spectrum is the spectrum of
$H_{\mathrm{eff}}$ with every eigenvalue appearing twice. In particular,
the ground-state energy and the YES and NO thresholds are preserved
exactly.

Thus the transformation introduces two physical auxiliary qubits, but
only the sign qubit is pinned.  The locality increases from $k$ to
$k+2$, the number of pinned qubits from $p$ to $p+1$, and the
restricted spectrum is preserved up to the exact multiplicity-two
realification.

For the $n$-dependent case, let $n$ be the final size of the source
DSIH instance.  The phase and sign qubits increase the final physical
size to
\[
n'=n+2.
\]
Define
\[
\widehat h^{(n')}
=
\mathcal S\!\left(h^{(n'-2)}\right)
\]
for $n'\geq3$, and set $\widehat h^{(n')}=0$ for smaller values.

Given $1^{n'}$ and a requested precision, the entries of
$h^{(n'-2)}$ can be computed uniformly. Taking real and imaginary parts
and applying the entrywise positive- and negative-part maps can likewise
be performed to the required precision in polynomial time. Because the
matrix dimension is constant, these entrywise operations increase the
operator norm by at most a constant factor. Hence
\[
\|\widehat h^{(n')}\|
\leq
\operatorname{poly}(n').
\]
\end{proof}

For real source interactions, the realification qubit is unnecessary.
This removes one unit of locality overhead and is the specialization
used below.

\begin{cor}
\label{cor:real-stoquastic-embedding}
Suppose the interaction $h$ of a $p$-Pinned $\mathrm{SIH}_k$ instance
is real in the computational basis.  Then there is a polynomial-time
reduction
\[
p\text{-Pinned }\mathrm{SIH}_k
\longrightarrow
(p+1)\text{-Pinned }\mathrm{Stoq\text{-}SIH}_{k+1}
\]
that preserves the restricted spectrum exactly.

The same statement holds for DSIH.
\end{cor}

\begin{proof}
Since $h$ is already real, apply only the sign embedding
\[
h=h^{-}+h^{+}
\longmapsto
h^{-}\otimes I-h^{+}\otimes X.
\]
Compression of the new sign qubit to $\Ket{-}$ gives exactly
\[
h^{-}+h^{+}=h.
\]
Thus the construction adds one qubit of locality and one pinned qubit
while preserving the restricted spectrum exactly.  In the
$n$-dependent case the final size changes from $n$ to $n+1$, and the
uniformity argument is identical to that of
\cref{thm:stoquastic-single-interaction-embedding}.
\end{proof}

\subsection{A QMA-Complete Pinned Stoquastic Single-Interaction Problem}

The QMA-complete interaction $g_{\mathrm H}$ of
\cref{thm:sih3-qma-complete} is real, so the sharper real-source embedding
applies directly.

\begin{thm}
\label{thm:pinned-stoq-sih4}
There exists a fixed $4$-local stoquastic Hermitian interaction such
that the $1$-Pinned $\mathrm{Stoq\text{-}SIH}_4$ problem is
QMA-complete.
\end{thm}

\begin{proof}
The fixed interaction $g_{\mathrm H}$ of
\cref{thm:sih3-qma-complete} is real in the computational basis.
Applying \cref{cor:real-stoquastic-embedding} with $p=0$ and $k=3$
gives
\[
\mathrm{SIH}_3
\longrightarrow
1\text{-Pinned }\mathrm{Stoq\text{-}SIH}_4,
\]
with exact preservation of the restricted spectrum.  QMA-hardness
therefore follows from \cref{thm:sih3-qma-complete}, while containment in QMA
follows from \cref{lem:pinned-containment}.
\end{proof}

\section{Geometrically Local Stoquastic \texorpdfstring{$\mathrm{DSIH}_9$}{DSIH9}}
\label{sec:stoq-n-dependent-hardness}

We now combine stoquasticity, geometric locality, and strict
single-interaction uniformity.  The interaction is allowed to depend
uniformly on the final physical system size, but on no other
instance-specific parameter.

\begin{thm}
\label{thm:StoqMA-complete}
There exists a uniform family of $9$-local stoquastic Hermitian
interactions $\{h^{(n)}\}_{n\geq1}$ for which geometrically local
$\mathrm{Stoq\text{-}DSIH}_9$ is StoqMA-complete.
\end{thm}

For hardness, we start from the geometrically local stoquastic
Hamiltonian source of \cref{thm:geo-stoq}.  Following the spatial
strategy of \cref{subsec:unpinned_geometric}, we use a fixed local
operation set and encode the resulting Hamiltonian with one
$9$-local stoquastic interaction.  The only system-size-dependent
coefficient will be the output-penalty strength.

\subsection{Spatially Local StoqMA Verification}

Let $L\in\mathrm{StoqMA}$.  We use a verifier in which Toffoli is the
only nontrivial computational gate; the spatialization below then
introduces native SWAP and identity operations for routing and time
advancement. This entails no loss of generality: the Toffoli gate is universal for classical reversible computation given access to qubits set as $\Ket{1}$.

We apply the geometrically local spatialization of
\cref{subsubsec: Stoq Geometrically Local,app:geometric-stoq-source}.
Let
\[
U_1,\ldots,U_T
\]
denote the resulting polynomial-length spatial circuit. Following the
notation of the source construction, let $I_0$, $I_1$, and $I_+$ denote
the physical work qubits whose prescribed initial states are
$\Ket{0}$, $\Ket{1}$, and $\Ket{+}$, respectively; the witness qubits
remain unconstrained.

The spatialization preserves the completeness and soundness parameters
of the original StoqMA verifier. Thus, if $x\in L_{\mathrm{yes}}$,
there exists a witness accepted with probability at least
$\epsilon_{\mathrm{yes}}$, whereas if $x\in L_{\mathrm{no}}$, every
witness is accepted with probability at most
$\epsilon_{\mathrm{no}}$, where
\[
\epsilon_{\mathrm{yes}}-\epsilon_{\mathrm{no}}
\geq
\frac{1}{\operatorname{poly}(|x|)}.
\]

After spatialization, every operation belongs to the fixed set
\[
\{\mathrm{Toffoli},\mathrm{SWAP},I\},
\]
and every physical work qubit participates in only $O(1)$ circuit
operations.  As in \cref{subsec:unpinned_geometric}, SWAP and $I$ are
kept as native spatial operations rather than compiled into Toffoli
gates, preserving the constant participation bound required for
bounded degree.

Finally, by padding with an initial and a final identity round, we may
assume without loss of generality that the first and last computational
rounds consist entirely of identity operations. The explicit lattice
layout, operation count, and localized initialization construction are
given in \cref{app:geometric-stoq-source}.

\subsection{The Stoquastic Single-Interaction Family}

We encode the geometrically local source Hamiltonian
$H^{\mathrm{stoq}_g}$ using a uniform family of $9$-local stoquastic
interactions
\[
\{h^{(n)}\}_{n\geq1},
\]
where $n$ is the total number of physical qubits in the final instance,
in accordance with \cref{def: DSIH}.  All matrix entries will be fixed
constants except the coefficient of the output-penalty branch.

Set
\[
\delta_n=n^{-8}.
\]
The exponent is chosen so that the output perturbation remains small
relative to the $\Omega(T^{-3})$ gap of the unperturbed history-state
Hamiltonian while still leaving an inverse-polynomial YES--NO gap.

The interaction $h^{(n)}$ acts on nine qubits
$(v_1,\ldots,v_9)$ and is defined by
\begin{align}
\label{eq:stoq-dsih9-h-def}
h^{(n)}(v_1,\ldots,v_9)
={}&
\frac{1}{4}
\left(
h^{\mathrm{stoq}}_{\mathrm{prop}}(v_1,\ldots,v_7)
+
h^{(n)}_{\mathrm{pen}}(v_1,\ldots,v_7)
\right)
\otimes I_{v_8,v_9}
\nonumber \\
&+
I_{v_1,\ldots,v_7}
\otimes
h_{\mathrm{check}}(v_8,v_9).
\end{align}

The propagation component contains three branches,
corresponding respectively to Toffoli, SWAP, and identity propagation:
\begin{align}
\label{eq:stoq-dsih9-prop-def}
h^{\mathrm{stoq}}_{\mathrm{prop}}(v_1,\ldots,v_7)
={}&
h^{\mathrm{stoq}}_{\mathrm{prop}_1}(v_1,\ldots,v_4)
\otimes
\Pi^{\Ket{0}}_{v_5}
\otimes
\Pi^{\Ket{1}}_{v_6}
\otimes
\Pi^{\Ket{1}}_{v_7}
\nonumber\\
&+
h^{\mathrm{stoq}}_{\mathrm{prop}_2}(v_1,v_2,v_3)
\otimes
I_{v_4}
\otimes
\Pi^{\Ket{1}}_{v_5}
\otimes
\Pi^{\Ket{1}}_{v_6}
\otimes
\Pi^{\Ket{0}}_{v_7}
\nonumber\\
&+
h^{\mathrm{stoq}}_{\mathrm{prop}_3}(v_1,v_2)
\otimes
I_{v_3}
\otimes
\Pi^{\Ket{1}}_{v_4}
\otimes
\Pi^{\Ket{0}}_{v_5}
\otimes
\Pi^{\Ket{0}}_{v_6}
\otimes
\Pi^{\Ket{0}}_{v_7}.
\end{align}

For Toffoli propagation, define
\begin{align}
\label{eq:stoq-dsih9-prop-toffoli}
h^{\mathrm{stoq}}_{\mathrm{prop}_1}(v_1,\ldots,v_4)
=
\frac{1}{2}
\Big(
& I_{v_1,v_2,v_3}\otimes\Ket{1}\Bra{1}_{v_4}
+
I_{v_1,v_2,v_3}\otimes\Ket{0}\Bra{0}_{v_4}
\nonumber\\
&-
\mathrm{Toffoli}_{v_1,v_2,v_3}
\otimes\Ket{1}\Bra{0}_{v_4}
-
\mathrm{Toffoli}_{v_1,v_2,v_3}^{\dagger}
\otimes\Ket{0}\Bra{1}_{v_4}
\Big).
\end{align}

For SWAP propagation, define
\begin{align}
\label{eq:stoq-dsih9-prop-swap}
h^{\mathrm{stoq}}_{\mathrm{prop}_2}(v_1,v_2,v_3)
=
\frac{1}{2}
\Big(
&I_{v_1,v_2}\otimes\Ket{1}\Bra{1}_{v_3}
+
I_{v_1,v_2}\otimes\Ket{0}\Bra{0}_{v_3}
\nonumber\\
&-
\mathrm{SWAP}_{v_1,v_2}
\otimes\Ket{1}\Bra{0}_{v_3}
-
\mathrm{SWAP}_{v_1,v_2}^{\dagger}
\otimes\Ket{0}\Bra{1}_{v_3}
\Big).
\end{align}

For identity propagation, define
\begin{align}
\label{eq:stoq-dsih9-prop-identity}
h^{\mathrm{stoq}}_{\mathrm{prop}_3}(v_1,v_2)
=
\frac{1}{2}
\Big(
&I_{v_1}\otimes\Ket{1}\Bra{1}_{v_2}
+
I_{v_1}\otimes\Ket{0}\Bra{0}_{v_2}
\nonumber\\
&-
I_{v_1}\otimes\Ket{1}\Bra{0}_{v_2}
-
I_{v_1}\otimes\Ket{0}\Bra{1}_{v_2}
\Big).
\end{align}

The remaining source-Hamiltonian terms are encoded by
\begin{align}
\label{eq:stoq-dsih9-pen-def}
h^{(n)}_{\mathrm{pen}}(v_1,\ldots,v_7)
={}&
\Big(
\Pi^{\Ket{0}}_{v_1}
\otimes
\Pi^{\Ket{0}}_{v_2}
\otimes
\Pi^{\Ket{1}}_{v_3}
\otimes
\Pi^{\Ket{1}}_{v_4}
\nonumber\\
&\quad+
\Pi^{\Ket{1}}_{v_1}
\otimes
\Pi^{\Ket{-}}_{v_2}
\otimes
\Pi^{\Ket{0}}_{v_3}
\otimes
\Pi^{\Ket{0}}_{v_4}
\nonumber\\
&\quad+
\delta_n
\Pi^{\Ket{1}}_{v_1}
\otimes
\Pi^{\Ket{-}}_{v_2}
\otimes
\Pi^{\Ket{0}}_{v_3}
\otimes
\Pi^{\Ket{1}}_{v_4}
\Big)
\nonumber\\
&\qquad\otimes
\Pi^{\Ket{0}}_{v_5}
\otimes
\Pi^{\Ket{1}}_{v_6}
\otimes
\Pi^{\Ket{0}}_{v_7}.
\end{align}

The three branches of \cref{eq:stoq-dsih9-pen-def} implement,
respectively, the clock and computational-basis initialization
penalties, the $\Ket{+}$-initialization penalties, and the output
penalty.  Only the output branch carries the $n$-dependent coefficient
$\delta_n$.

Finally,
\begin{equation}
\label{eq:stoq-dsih9-check}
h_{\mathrm{check}}(v_8,v_9)
=
I_{v_8,v_9}-
\Pi^{\Ket{0}}_{v_8}
\otimes
\Pi^{\Ket{1}}_{v_9}.
\end{equation}
This operator penalizes every state orthogonal to $\Ket{01}$.

As in \cref{subsec:unpinned_geometric}, the tensor-factor positions of
$h^{(n)}$ do not have fixed \emph{computational} and \emph{selector} roles.
Instead, projector positions already required by the target Hamiltonian
are shared with the branch-selection mechanism.

On the valid auxiliary configuration, the selector projectors are routed
so that exactly the intended propagation or penalty branch survives.
The explicit routing and its verification are given in
\cref{subsec:stoq-dsih-spatial-embedding,app:stoq-dsih-routing}.

The interaction family is stoquastic.  The computational-basis
projectors and $h_{\mathrm{check}}$ are diagonal, while
$\Pi^{\Ket{-}}$ has non-positive off-diagonal entries.  The Toffoli,
SWAP, and identity matrices have nonnegative computational-basis
entries and appear with a minus sign in the off-diagonal propagation
terms. 
Hence every off-diagonal entry of $h^{(n)}$ is non-positive.

Each propagation operator
$h^{\mathrm{stoq}}_{\mathrm{prop}_i}$ is positive semidefinite, and every
term in \cref{eq:stoq-dsih9-pen-def} is positive semidefinite. Therefore
\begin{equation}
\label{eq:stoq-dsih9-check-domination}
h^{(n)}
\succeq
I_{v_1,\ldots,v_7}
\otimes
h_{\mathrm{check}}(v_8,v_9).
\end{equation}
This property will provide the invalid-auxiliary-sector energy penalty,
exactly as in \cref{subsec:unpinned_geometric}.

Finally, $\{h^{(n)}\}_{n\geq1}$ is uniform in the sense of
\cref{def: DSIH}.  Its only $n$-dependent coefficient is
\[
\delta_n=n^{-8},
\]
which is exactly computable from $1^n$, while all remaining entries are
fixed constants.  Moreover,
\[
\|h^{(n)}\|=O(1).
\]

\subsection{Auxiliary-Block Construction and Spatial Embedding}
\label{subsec:stoq-dsih-spatial-embedding}

We use the same four-pair auxiliary-block architecture as in
\cref{subsec:unpinned_geometric}, now with the work--clock placement of
\cref{app:geometric-stoq-source}.  For each time step $t$, place a
dedicated block
\[
A^{(t)}
=
\left(
A^{(t)}_0,\ldots,A^{(t)}_7
\right)
\]
within the corresponding constant-diameter neighborhood.  Its valid
state is
\[
\Ket{\Omega_t}_A
=
\Ket{01}_{A^{(t)}_0,A^{(t)}_1}
\Ket{01}_{A^{(t)}_2,A^{(t)}_3}
\Ket{01}_{A^{(t)}_4,A^{(t)}_5}
\Ket{01}_{A^{(t)}_6,A^{(t)}_7},
\]
and
\[
\Ket{\Omega}_A
=
\bigotimes_{t=1}^{T}\Ket{\Omega_t}_A.
\]

The prefactor $1/4$ in the computational part of $h^{(n)}$ is
compensated by using four applications of $h^{(n)}$ for every target
term. These applications are indexed by
$j\in\{0,1,2,3\}$. On the valid auxiliary state all four applications
select the same computational branch, so each contributes one quarter
of the desired target operator and their sum reproduces it with
coefficient one.

At the same time, the check positions $(v_8,v_9)$ are occupied by

\[
\left(
A^{(t)}_{2j+6},
A^{(t)}_{2j+7}
\right),
\]

where auxiliary indices are interpreted modulo $8$. As $j$ ranges from
$0$ to $3$, these positions contain all four auxiliary pairs exactly
once. Thus the same four applications that recover the coefficient of
the target term also route every pair through $h_{\mathrm{check}}$.

We define the valid auxiliary subspace by
$
\mathcal{H}_{\mathrm{val}}
=
\mathcal{H}_{W,C}
\otimes
\operatorname{span}\{\Ket{\Omega}_A\},
$
and its orthogonal complement by 
$
\mathcal{H}_{\mathrm{inv}}
=
\mathcal{H}_{\mathrm{val}}^\perp.
$ 
Let 
$
V_{\mathrm{val}}
:
\mathcal{H}_{W,C}
\longrightarrow
\mathcal{H}_{\mathrm{val}}
$ 
be the isometry 
$
V_{\mathrm{val}}\Ket{\psi}
=
\Ket{\psi}_{W,C}\otimes\Ket{\Omega}_A.
$

\paragraph{Final physical system size.}
Let $n_{\mathrm{act}}$ be the number of work, clock, and auxiliary
qubits in the active construction.  Choose a polynomial $q$ such that
\[
\epsilon_{\mathrm{yes}}-\epsilon_{\mathrm{no}}
\geq
\frac{1}{q(|x|)},
\]
and assume without loss of generality that
$q(|x|)\geq |x|+1$.

Set
\[
n
=
\max\left\{
n_{\mathrm{act}},
T+1,
\left\lceil q(|x|)\right\rceil
\right\}.
\]
We add $n-n_{\mathrm{act}}$ isolated physical qubits, so that $n$ is
exactly the final system size used to select $h^{(n)}$.

These isolated qubits change only spectral multiplicities: they leave
the set of eigenvalues, the ground-state energy, interaction range, and
vertex degree unchanged.

Having fixed the final physical size, set
\[
\delta=\delta_n=n^{-8}.
\]

The total Hamiltonian is
\begin{equation}
\label{eq:stoq-dsih-global}
\widetilde H^{\mathrm{stoq}_g}
=
\widetilde H^{\mathrm{stoq}_g}_{\mathrm{in}}
+
\widetilde H^{\mathrm{stoq}_g}_{\mathrm{out}}
+
\widetilde H^{\mathrm{stoq}_g}_{\mathrm{prop}}
+
\widetilde H^{\mathrm{stoq}_g}_{\mathrm{clock}}.
\end{equation}

Each source term is implemented by four applications of $h^{(n)}$.
The appropriate work and clock qubits are routed into its computational
positions, while the local auxiliary block selects the intended branch
and cycles all four auxiliary pairs through the check positions.
The explicit ordered $9$-tuples are given in
\cref{app:stoq-dsih-routing}.

\begin{lemma}[Valid and invalid auxiliary subspaces]
\label{lem:stoq-dsih-auxiliary-subspaces}
The restriction of $\widetilde H^{\mathrm{stoq}_g}$ to the valid
auxiliary subspace is exactly the source Hamiltonian:
\begin{equation}
\label{eq:stoq-dsih-valid-restriction}
V_{\mathrm{val}}^\dagger
\widetilde H^{\mathrm{stoq}_g}
V_{\mathrm{val}}
=
H^{\mathrm{stoq}_g}.
\end{equation}
Moreover,

\[
\left.
\widetilde H^{\mathrm{stoq}_g}
\right|_{\mathcal H_{\mathrm{inv}}}
\succeq
I_{\mathcal H_{\mathrm{inv}}}.
\]

\end{lemma}

\begin{proof}

On the valid auxiliary state $\Ket{\Omega}_A$, the routing construction
described above selects the intended branch of $h^{(n)}$ for every
initialization, output, clock, and propagation term, while the remaining
branches are annihilated by their selector projectors. For each target
term, the four associated applications contribute four copies of one
quarter of the desired operator, and all $h_{\mathrm{check}}$
contributions vanish on the valid $\Ket{01}$ auxiliary pairs. Hence
\[
V_{\mathrm{val}}^\dagger
\widetilde H^{\mathrm{stoq}_g}
V_{\mathrm{val}}
=
H^{\mathrm{stoq}_g},
\]
which proves \cref{eq:stoq-dsih-valid-restriction}.

For the invalid subspace, let $H_A$ denote the sum of all
$h_{\mathrm{check}}$ contributions appearing in the construction.
Since the non-check part of every application of $h^{(n)}$ is positive
semidefinite,

\[
\widetilde H^{\mathrm{stoq}_g}
\succeq
H_A.
\]

For each time-step block $A^{(t)}$, the construction includes four
applications whose check positions cycle through all four auxiliary
pairs.  Hence

\[
H_A
\succeq
\sum_{t=1}^{T}
\sum_{r=0}^{3}
h_{\mathrm{check}}
\left(
A^{(t)}_{2r},
A^{(t)}_{2r+1}
\right).
\]

The operators on the right are commuting projectors, each having kernel
$\operatorname{span}\{\Ket{01}\}$. Their joint kernel is therefore
exactly

\[
\mathcal H_{W,C}
\otimes
\operatorname{span}\{\Ket{\Omega}_A\}
=
\mathcal H_{\mathrm{val}},
\]

and their sum has eigenvalue at least $1$ on
$\mathcal H_{\mathrm{inv}}$. Thus

\[
\left.
\widetilde H^{\mathrm{stoq}_g}
\right|_{\mathcal H_{\mathrm{inv}}}
\succeq
I_{\mathcal H_{\mathrm{inv}}},
\]

as claimed.
\end{proof}

\paragraph{Geometric locality and tuple validity.}
Each constant-size block $A^{(t)}$ is placed within constant distance
of the work and clock qubits used at time $t$.  Hence every application
of $h^{(n)}$ has constant geometric diameter.  Each source term is
replaced by only four applications, and every auxiliary block
participates in only $O(1)$ applications, so bounded degree is also
preserved.

The explicit routing in \cref{app:stoq-dsih-routing} verifies that the
ordered $9$-tuples are pairwise distinct and contain nine distinct
physical qubits.

\subsection{Correctness and StoqMA-Completeness}
\label{sec:stoq-dsih-com-sound}

\begin{proof}[Proof of \cref{thm:StoqMA-complete}]
By \cref{lem:stoq-dsih-auxiliary-subspaces}, the valid auxiliary subspace
$\mathcal H_{\mathrm{val}}$ is invariant under
$\widetilde H^{\mathrm{stoq}_g}$ and
\begin{equation}
\label{eq:stoq-dsih-valid-spectrum}
V_{\mathrm{val}}^\dagger
\widetilde H^{\mathrm{stoq}_g}
V_{\mathrm{val}}
=
H^{\mathrm{stoq}_g}.
\end{equation}
Moreover,

\[
\left.
\widetilde H^{\mathrm{stoq}_g}
\right|_{\mathcal H_{\mathrm{inv}}}
\succeq
I_{\mathcal H_{\mathrm{inv}}}.
\]

Because every invalid auxiliary sector has energy at least $1$, while
the source energies considered below are inverse polynomial, the
ground-energy analysis reduces to the valid sector.

\paragraph{Completeness.}
Suppose $x\in L_{\mathrm{yes}}$. By the completeness analysis of
\cref{subsubsec: Stoq Geometrically Local}, there exists a witness and
corresponding history state $\Ket{\eta}$ satisfying

\[
\Bra{\eta}
H^{\mathrm{stoq}_g}
\Ket{\eta}
\leq
\frac{\delta_n}{T+1}
\left(
1-\epsilon_{\mathrm{yes}}
\right).
\]

Therefore
$V_{\mathrm{val}}\Ket{\eta}$ has the same energy under
$\widetilde H^{\mathrm{stoq}_g}$. We choose the YES threshold
\begin{equation}
\label{eq:stoq-dsih-yes-threshold}
a
=
\frac{\delta_n}{T+1}
\left(
1-\epsilon_{\mathrm{yes}}
\right).
\end{equation}

\paragraph{Soundness.}
Suppose $x\in L_{\mathrm{no}}$.  Write
\[
H^{\mathrm{stoq}_g}
=
H_0+\delta_n H^{\mathrm{stoq}_g}_{\mathrm{out}},
\]
where $H_0$ contains the initialization, propagation, and clock terms.
By \cref{eq:geo-stoq-unperturbed-gap}, its spectral gap satisfies
\[
\Gamma_T=\Omega(T^{-3}).
\]

Since $T\leq n$ and
$\Gamma_T=\Omega(T^{-3})$, we have
\[
\Gamma_T=\Omega(n^{-3}).
\]
Thus the choice $\delta_n=n^{-8}$ gives
\[
\frac{\delta_n}{\Gamma_T}=O(n^{-5}),
\]
so the perturbative regime required by the Projection Lemma holds for
all sufficiently large $n$.

Since $H^{\mathrm{stoq}_g}_{\mathrm{out}}$ is a projector, the
Projection Lemma~\cite[Lemma~1]{Kempe2004} applies for all sufficiently
large $n$, yielding
\[
\lambda_{\min}
\left(
H^{\mathrm{stoq}_g}
\right)
\geq
\frac{\delta_n}{T+1}
\left(1-\epsilon_{\mathrm{no}}\right)
-
\frac{\delta_n^2}{\Gamma_T-2\delta_n}.
\]
Since $\delta_n=n^{-8}$, $T\leq n$, and
$\Gamma_T=\Omega(T^{-3})$,
\[
\frac{\delta_n^2}{\Gamma_T-2\delta_n}
=
O(n^{-13}).
\]
Therefore, for all sufficiently large $n$,
\[
\lambda_{\min}
\left(
H^{\mathrm{stoq}_g}
\right)
\geq
\frac{\delta_n}{T+1}
\left(1-\epsilon_{\mathrm{no}}\right)
-
n^{-12}.
\]
We consequently choose
\[
b
=
\frac{\delta_n}{T+1}
\left(1-\epsilon_{\mathrm{no}}\right)
-
n^{-12}.
\]

\paragraph{Promise gap.}
It remains to verify that $b-a$ is inverse polynomial in the final
physical system size $n$. By the definition of $n$,

\[
T+1\leq n,
\qquad
q(|x|)\leq n,
\]

where

\[
\epsilon_{\mathrm{yes}}
-
\epsilon_{\mathrm{no}}
\geq
\frac{1}{q(|x|)}.
\]

Therefore
\begin{align}
b-a
&=
\frac{\delta_n}{T+1}
\left(
\epsilon_{\mathrm{yes}}
-
\epsilon_{\mathrm{no}}
\right)
-
n^{-12}
\nonumber\\
&\geq
n^{-8}
\cdot
\frac{1}{n}
\cdot
\frac{1}{n}
-
n^{-12}
\nonumber\\
&=
n^{-10}-n^{-12}
\nonumber\\
&\geq
\frac{1}{2}n^{-10}
\label{eq:stoq-dsih-promise-gap}
\end{align}
for all sufficiently large $n$. The finitely many smaller sizes do not
affect the polynomial-time reduction.

The construction has polynomial size, and the interaction family
$\{h^{(n)}\}_{n\geq1}$ is uniformly computable, stoquastic, and
depends only on the final physical system size.  By the geometric
analysis of \cref{subsec:stoq-dsih-spatial-embedding} and the explicit
routing of \cref{app:stoq-dsih-routing}, the output hypergraph has
bounded interaction range and bounded degree, and all ordered
$9$-tuples are valid and pairwise distinct.

Thus the problem is StoqMA-hard.  Containment in StoqMA follows from
\cref{lem:stoquastic-containment}, proving StoqMA-completeness.
\end{proof}

\section{Further Consequences of Single-Interaction Compilation}
\label{sec:further-consequences}

This section develops several consequences of the
single-interaction compilation framework beyond the basic hardness
results established above. We first show that the Hamiltonian quantum
PCP conjecture admits an equivalent single-interaction normal form.
We then consider three additional promise settings: frustration-free,
exponentially precise, and guided Hamiltonians.

Together, these results show that the single-interaction restriction
is compatible not only with standard Local Hamiltonian hardness,
but also with several qualitatively different forms of low-energy
structure and promise-gap behavior.

\subsection{A Single-Interaction Normal Form for Hamiltonian qPCP}
\label{subsec:qpcp-normal-form}

\begin{cor}[Single-interaction normal form for constant-relative-gap Local Hamiltonian]
\label{cor:sih-qpcp-normal-form}
Fix $k=O(1)$. Suppose that there exists a constant $c>0$ such that
$k$-Local Hamiltonian remains QMA-hard when restricted to instances
\[
H=\sum_{i=1}^{M} H_i,
\qquad
0\preceq H_i\preceq I,
\]
whose promise gap satisfies
\[
b-a\ge cM.
\]
Then there exist a constant $c'>0$ and a fixed $(3k+3)$-local
positive-semidefinite interaction $h_k$ satisfying
\[
0\preceq h_k\preceq I
\]
such that the corresponding $\mathrm{SIH}_{3k+3}$ problem remains
QMA-hard when restricted to instances with $\widetilde M$
interaction terms and promise gap
\[
\widetilde b-\widetilde a\ge c'\widetilde M.
\]

Consequently, the Hamiltonian quantum PCP conjecture holds if and only
if it holds under the additional restriction that every local term is
an unweighted copy of one fixed interaction.
\end{cor}

\begin{proof}
Let
\[
H=\sum_{i=1}^{M} H_i
\]
be an instance of the assumed constant-relative-gap $k$-Local
Hamiltonian problem, with
\[
0\preceq H_i\preceq I
\]
and
\[
b-a\ge cM.
\]

Apply \cref{thm:universal-single-interaction-reduction} with
\[
\varepsilon=\frac{b-a}{4},
\]
and let $\widetilde H$ denote the resulting
$\mathrm{SIH}_{3k+3}$ Hamiltonian, with fixed common interaction $h_k$.
Writing
\[
g:=\frac{b-a}{M},
\]
the assumption gives $g\geq c$. Moreover,
\[
Q
=
\max\left\{
1,
\left\lceil
\frac{M4^k}{2\varepsilon}
\right\rceil
\right\}
=
\max\left\{
1,
\left\lceil
\frac{2\cdot4^k}{g}
\right\rceil
\right\},
\]
and hence
\[
Q
\le
\max\left\{
1,
\left\lceil
\frac{2\cdot4^k}{c}
\right\rceil
\right\}
=
O(1).
\]

The output promise gap is
\[
\widetilde b-\widetilde a
=
Q\bigl((b-a)-2\varepsilon\bigr)
=
\frac{Q(b-a)}{2}.
\]

Let $R$ denote the number of signed-Pauli occurrences produced by the
normalization. Since $0\preceq H_i\preceq I$,
\[
R
\le
M4^k\left(Q+\frac12\right).
\]
The unpinned compiler uses $2k+2$ applications of the common
interaction per occurrence, so
\[
\widetilde M
=
(2k+2)R
\le
(2k+2)M4^k\left(Q+\frac12\right).
\]
Therefore the output fractional promise gap satisfies
\[
\begin{aligned}
\widetilde g
:=
\frac{\widetilde b-\widetilde a}{\widetilde M}
&\ge
\frac{Q}{2(2k+2)4^k(Q+\frac12)}
\frac{b-a}{M} \\
&\ge
\frac{g}{3(2k+2)4^k} \\
&\ge
\frac{c}{3(2k+2)4^k}
=:c'>0.
\end{aligned}
\]
Hence a constant fractional promise gap is preserved up to a fixed
constant factor.

The converse implication is immediate, since a fixed-interaction SIH
instance is a special case of a Local Hamiltonian instance of the same
locality. Hence the Hamiltonian quantum PCP conjecture is equivalent to
its fixed unweighted single-interaction restriction.

It remains only to put the fixed interaction into the
positive-semidefinite normalization stated in the corollary.  Let
\[
\lambda_-:=\lambda_{\min}(h_k),
\qquad
D:=\max\left\{
1,
\lambda_{\max}(h_k)-\lambda_{\min}(h_k)
\right\},
\]
and define
\[
h_k'
=
\frac{h_k-\lambda_- I}{D}.
\]
Then
\[
0\preceq h_k'\preceq I.
\]

If $\widetilde H$ contains $\widetilde M$ occurrences of $h_k$, replacing
each occurrence by $h_k'$ gives
\[
\widetilde H'
=
\frac{
\widetilde H-\widetilde M\lambda_- I
}{D}.
\]
Accordingly, the thresholds transform as
\[
\widetilde a'
=
\frac{
\widetilde a-\widetilde M\lambda_-
}{D},
\qquad
\widetilde b'
=
\frac{
\widetilde b-\widetilde M\lambda_-
}{D},
\]
and hence
\[
\frac{\widetilde b'-\widetilde a'}{\widetilde M}
=
\frac{1}{D}
\frac{\widetilde b-\widetilde a}{\widetilde M}.
\]
Since $D$ is a fixed constant depending only on the fixed interaction,
the output fractional promise gap remains bounded below by a positive
constant. Renaming $h_k'$ as $h_k$ gives the claimed normalization
\[
0\preceq h_k\preceq I.
\]

\end{proof}

\begin{rem}
The reduction preserves a constant fractional promise gap up to a
constant factor, but it does not provide a general fractional-gap
amplification guarantee.  In particular, the preceding calculation
controls the output gap relative to the number of compiled interaction
occurrences and shows that a constant fractional gap remains constant.
Thus the result should be viewed as a single-interaction normal form
for the Hamiltonian quantum PCP conjecture, rather than as a quantum
gap-amplification procedure.
\end{rem}

\subsection{Frustration-Free Single-Interaction Hamiltonians}
\label{subsec:ff-sih}

\begin{thm}[Frustration-free SIH]
\label{thm:ff-sih9}
There exists a fixed $9$-local positive-semidefinite interaction
$h_{\mathrm{ff}}^{(9)}$ such that
$\mathrm{FF\text{-}SIH}_{9}$ is QMA-complete.
\end{thm}

\begin{proof}

Containment in QMA follows from
\cref{lem:ff-sih-containment}.

For QMA-hardness, let $L\in\mathrm{QMA}$. By
\cref{thm:qma-perfect-completeness}, there is a polynomial-size
perfectly complete verifier over the fixed gate set
\[
\mathcal G=\{H,X,\mathrm{Toffoli}\}.
\]
Apply the unary-clock circuit-to-Hamiltonian construction of
\cref{subsubsec: Circuit-to-Hamiltonian}, and write
\[
K
=
H_{\mathrm{in}}
+
H_{\mathrm{out}}
+
H_{\mathrm{prop}}
+
H_{\mathrm{clock}}.
\]
Since the largest verifier gate is Toffoli, $K$ is $6$-local, and every
local term of $K$ is positive semidefinite.

Perfect completeness gives
\[
x\in L_{\mathrm{yes}}
\quad\Longrightarrow\quad
\lambda_{\min}(K)=0,
\]
whereas for $x\in L_{\mathrm{no}}$,
\cref{lemma: kitaev soundness} gives
\[
\lambda_{\min}(K)
\ge
\frac{c_0}{T^3}\left(1-\sqrt{s}\right).
\]
Since the perfectly complete verifier has
\[
1-s\ge\frac{1}{\operatorname{poly}(|x|)}
\]
and
\[
1-\sqrt{s}
=
\frac{1-s}{1+\sqrt{s}}
\ge
\frac{1-s}{2},
\]
it follows that
\[
\lambda_{\min}(K)
\ge
\frac{1}{\operatorname{poly}(|x|)}.
\]

Because frustration freeness is not stable under approximation, we
do not invoke the universal finite-alphabet normalization.  Instead,
we encode the finite set of local term types appearing in $K$
directly into a fixed $9$-local interaction using the branch-sharing
construction of \cref{sec:stoq-n-dependent-hardness}.

Starting from the interaction structure of
\cref{eq:stoq-dsih9-h-def}, we retain the same auxiliary check
$h_{\mathrm{check}}$ and selector patterns, but replace the three
propagation branches
\[
\mathrm{Toffoli},\ \mathrm{SWAP},\ I
\]
by
\[
\mathrm{Toffoli},\ H,\ X,
\]
respectively.

More precisely, the first propagation branch of
\cref{eq:stoq-dsih9-prop-def} is unchanged. In the second branch,
we replace $h^{\mathrm{stoq}}_{\mathrm{prop2}}$ by the standard
circuit-to-Hamiltonian propagation operator for $H$, tensored with the
identity on the unused work position; in the third branch, we replace
$h^{\mathrm{stoq}}_{\mathrm{prop3}}$ by the corresponding propagation
operator for $X$.  Finally, we replace
$h^{(n)}_{\mathrm{pen}}$ by the fixed projector branch
\[
h^{\mathrm{ff}}_{\mathrm{pen}}(v_1,\ldots,v_7)
=
\Pi^{\Ket{0}}_{v_1}
\otimes
\Pi^{\Ket{0}}_{v_2}
\otimes
\Pi^{\Ket{1}}_{v_3}
\otimes
\Pi^{\Ket{1}}_{v_4}
\otimes
\Pi^{\Ket{0}}_{v_5}
\otimes
\Pi^{\Ket{1}}_{v_6}
\otimes
\Pi^{\Ket{0}}_{v_7},
\]
which is routed to reproduce the initialization, output, and clock
penalties of $K$.

The three propagation branches are distinguished by the same
auxiliary patterns as in \cref{sec:stoq-n-dependent-hardness}, while
the penalty branch has selector pattern
\[
(v_5,v_6,v_7)=(0,1,0),
\]
which annihilates all three propagation branches.  Hence the same
four-application auxiliary-block construction reproduces every local
term of $K$ exactly on the valid auxiliary sector.

Assign each local-term occurrence of $K$ its own fresh auxiliary block
$A^{(i)}=(A^{(i)}_0,\ldots,A^{(i)}_7)$ with valid state
$\Ket{01}^{\otimes 4}$.
The four cyclic applications associated with one source occurrence are distinguished by 
the auxiliary pair occupying the two check positions, while distinct
source occurrences use disjoint auxiliary blocks. Hence all compiled 
ordered $9$-tuples are pairwise distinct, and every tuple contains nine 
distinct physical qubits. Since $K$ contains only
$\operatorname{poly}(|x|)$ local terms, both the number of auxiliary
qubits and the number of interaction occurrences remain polynomial.

Each circuit-to-Hamiltonian propagation operator is positive
semidefinite, as are $h^{\mathrm{ff}}_{\mathrm{pen}}$ and
$h_{\mathrm{check}}$. Hence
\[
h_{\mathrm{ff}}^{(9)}\succeq0.
\]
Moreover, the domination argument of
\cref{eq:stoq-dsih9-check-domination} applies unchanged: on the valid
auxiliary sector the compiled Hamiltonian is exactly $K$, whereas every
invalid auxiliary sector has energy at least $1$.

Unlike the construction of \cref{sec:stoq-n-dependent-hardness},
no system-size-dependent perturbative coefficient is needed here.
All coefficients are fixed constants, so
$h_{\mathrm{ff}}^{(9)}$ is independent of both the input and the
system size. Therefore YES instances have ground energy zero, while NO instances
have ground energy at least
\[
\min\left\{
1,\,
\frac{c_0}{T^3}\left(1-\sqrt{s}\right)
\right\}
=
\frac{1}{\operatorname{poly}(|x|)}.
\]
Thus the resulting $\mathrm{FF\text{-}SIH}_9$ problem is QMA-hard.
Together with QMA containment, it is QMA-complete.

\end{proof}

\subsection{Precise Single-Interaction Hamiltonians}
\label{subsec:precise-sih}

The precise setting requires no new single-interaction construction:
we reuse the geometrically local $\mathrm{SIH}_8$ interaction of
\cref{subsec:unpinned_geometric} and track an exponentially small
promise gap through the same reduction.

\begin{thm}[Precise geometrically local single-interaction Hamiltonians]
\label{thm:precise-geometric-sih8}
For the same fixed $8$-local interaction $h$ as in
\cref{thm:unpinned-fixed}, the geometrically local
$\mathrm{Precise\text{-}SIH}_8$ problem is PSPACE-complete.
\end{thm}

\begin{proof}
Containment in PSPACE follows from
\cref{lem:precise-sih-containment}, since geometric locality is an
additional restriction on the allowed instances.

For hardness, let $L\in\mathrm{PSPACE}$.  By
\cref{thm:qma-exp-pspace}, using the verifier parameters underlying
the Precise Local Hamiltonian construction of Fefferman and Lin, we
may take a polynomial-size verification procedure with completeness
$c$ and soundness $s$ satisfying
\[
1-c=\epsilon,
\qquad
1-s=d-\epsilon,
\qquad 
d=2^{-g(|x|)}
\]
for some polynomial $g$, where
$\epsilon=2^{-\operatorname{poly}(|x|)}$ may be chosen sufficiently small.

Compile the verifier to the fixed Barenco gate $\mathbb{B}$ as in
\cref{subsubsec:fixed-gate-verifier-compilation}, choosing the synthesis
accuracy so that the compiled unitary $\widetilde U$ satisfies
\[
\|\widetilde U-U\|\leq\eta,
\qquad
\eta=2^{-\operatorname{poly}(|x|)}.
\]
The inverse-free Solovay--Kitaev construction has polylogarithmic
overhead in $1/\eta$, so the compiled circuit remains polynomial in
size. Apply the spatial embedding of
\cref{subsec:unpinned_geometric} to this circuit. The additional SWAP
and identity operations are exact and therefore introduce no further
acceptance error. Let $T$ denote the resulting circuit length.

The operator-norm synthesis error changes every acceptance
probability by at most $2\eta$.
Writing
\[
\alpha:=\epsilon+2\eta,
\]
the compiled verifier therefore satisfies
\[
1-\widetilde c\leq\alpha,
\qquad
1-\widetilde s\geq d-\alpha.
\]
The corresponding spatial circuit-to-Hamiltonian instance has
completeness threshold
\[
a
\leq
\frac{\alpha}{T+1},
\]
while the soundness bound of
\cref{lemma: kitaev soundness} gives, for some constant $c_0>0$,
\[
\begin{aligned}
b
&\geq
\frac{c_0}{T^3}
\left(1-\sqrt{\widetilde s}\right) \\
&\geq
\frac{c_0}{2T^3}
\left(1-\widetilde s\right) \\
&\geq
\frac{c_0}{2T^3}
\left(d-\alpha\right).
\end{aligned}
\]

Choose the exponentially small parameters $\epsilon$ and $\eta$ so that
\[
\alpha
\leq
\min\left\{
\frac d2,\,
\frac{c_0d}{16T^2}
\right\}.
\]
Such a choice is possible because $T$ grows only polynomially in
$|x|$, $\log(1/\epsilon)$, and $\log(1/\eta)$. Choosing the precision
parameters polynomially large therefore makes $\epsilon$ and $\eta$
exponentially small while keeping $T$ polynomial in $|x|$. Hence
\[
b
\geq
\frac{c_0d}{4T^3},
\qquad
a
\leq
\frac{c_0d}{16T^3},
\]
and therefore
\[
b-a
\geq
\frac{3c_0d}{16T^3}
=
2^{-\operatorname{poly}(|x|)}.
\]

Apply the geometrically local single-interaction construction of
\cref{subsec:unpinned_geometric}. This step introduces no further
approximation: by \cref{cor:geometric-sih-low-energy-spectrum},
\[
\lambda_{\min}
\left(
\widetilde H^{\mathrm{geo}}
\right)
=
\lambda_{\min}
\left(
H^{\mathrm{geo}}
\right).
\]
Hence the compiled instance has the same thresholds $a$ and $b$.

The construction is geometrically local and uses the same fixed
$8$-local interaction $h$ as \cref{thm:unpinned-fixed}. Its output size
$n_{\mathrm{out}}$ is polynomial in $|x|$, and therefore
\[
b-a
\geq
2^{-\operatorname{poly}(n_{\mathrm{out}})}.
\]
Thus geometrically local
$\mathrm{Precise\text{-}SIH}_8$ is PSPACE-hard. Together with
containment in PSPACE, it is PSPACE-complete.
\end{proof}

\subsection{Guided Single-Interaction Hamiltonians}
\label{subsec:guided-sih}

\begin{thm}[Guided single-interaction Hamiltonians]
\label{thm:guided-sih9}
There exists a fixed $9$-local Hermitian interaction $h_{\mathrm G}$
such that $\mathrm{Guided\text{-}SIH}_{9}$ is BQP-complete.
Moreover, BQP-hardness holds even for instances with a unique ground
state separated from the remainder of the spectrum by an
inverse-polynomial gap and with a semi-classical subset guiding state
having overlap
\[
1-\frac{1}{\operatorname{poly}(n)}
\]
with the ground state.
\end{thm}

\begin{proof}
Containment in BQP follows from
\cref{lem:guided-sih-containment}.

For hardness, let $L\in\mathrm{BQP}$. 
Here we use the universal normalization before the exact compiler,
so we must control its effect on the unique ground state, spectral
gap, and guiding-state overlap.
By \cref{cor:guided-local-hamiltonian-hardness}, together with the
construction of Cade et al.~\cite{CadeEtAl2023}, there is a BQP-hard family of $2$-local Hamiltonian instances $H$ with thresholds $a<b$
satisfying
\[
b-a\ge\frac{1}{\operatorname{poly}(n)},
\]
a unique ground state $\Ket{\phi}$ with spectral gap
\[
\Delta_{\mathrm{sp}}
\ge
\frac{1}{\operatorname{poly}(n)},
\]
and a semi-classical subset state $\Ket{u}$ such that
\[
|\Braket{u|\phi}|
\ge
1-\frac{1}{q(n)}
\]
for some polynomial $q$.

Choose
\[
\varepsilon
\le
\min\left\{
\frac{b-a}{4},
\frac{\Delta_{\mathrm{sp}}}{4q(n)^2}
\right\}.
\]
Since all quantities on the right are inverse polynomial,
$1/\varepsilon= O( \operatorname{poly}(n))$.

Apply the universal finite-alphabet normalization of
\cref{lem:universal-finite-alphabet-normalization}. This gives an
integer $Q=\operatorname{poly}(n)$ and a Hamiltonian $K$ over the
fixed signed Pauli alphabet
\[
\mathcal A_2^\star
=
\left\{
\pm P:
P\in\{I,X,Y,Z\}^{\otimes2}
\right\}
\]
such that
\[
\left\|
H-\frac{1}{Q}K
\right\|
\le
\varepsilon.
\]
Write
\[
H'=\frac{1}{Q}K.
\]

By Weyl's inequality, the spectral gap of $H'$ is at least
\[
\Delta_{\mathrm{sp}}-2\varepsilon
\ge
\frac{\Delta_{\mathrm{sp}}}{2},
\]
so $H'$ still has a unique ground state; denote it by
$\Ket{\phi'}$.

The guiding state remains close to this perturbed ground state.
Indeed,
\[
\begin{aligned}
\Bra{\phi'}H\Ket{\phi'}
&\le
\Bra{\phi'}H'\Ket{\phi'}+\varepsilon\\
&\le
\Bra{\phi}H'\Ket{\phi}+\varepsilon\\
&\le
\lambda_{\min}(H)+2\varepsilon.
\end{aligned}
\]
Since $\Ket{\phi}$ is the unique ground state of $H$ with gap
$\Delta_{\mathrm{sp}}$,
\[
\Bra{\phi'}H\Ket{\phi'}
\ge
\lambda_{\min}(H)
+
\Delta_{\mathrm{sp}}
\left(
1-|\Braket{\phi|\phi'}|^2
\right).
\]
Hence
\[
|\Braket{\phi|\phi'}|^2
\ge
1-\frac{2\varepsilon}{\Delta_{\mathrm{sp}}}.
\]
Choosing the global phase of $\Ket{\phi'}$ so that
$\Braket{\phi|\phi'}=|\Braket{\phi|\phi'}|$, we obtain
\[
\|\Ket{\phi'}-\Ket{\phi}\|
\le
2\sqrt{\frac{\varepsilon}{\Delta_{\mathrm{sp}}}}.
\]
Therefore
\[
\begin{aligned}
|\Braket{u|\phi'}|
&\ge
|\Braket{u|\phi}|
-
\|\Ket{\phi'}-\Ket{\phi}\|\\
&\ge
1-\frac{1}{q(n)}
-
2\sqrt{\frac{\varepsilon}{\Delta_{\mathrm{sp}}}}\\
&\ge
1-\frac{2}{q(n)}.
\end{aligned}
\]

We now apply the exact unpinned compiler of
\cref{thm:unpinned-spectral-preservation} to $K$, with separation
parameter $\gamma=1$. Since the signed two-qubit Pauli alphabet has
\[
|\mathcal A_2^\star|
=
2\cdot 4^2
=
32,
\]
the selector length is
\[
\ell
=
\left\lceil\log_2 32\right\rceil
=
5,
\]
and the resulting single interaction has locality
\[
2+\ell+2=9.
\]

Let
\[
V\Ket{\psi}
=
\Ket{\psi}\otimes\Ket{\Omega}_A
\]
be the valid-sector isometry. Exact spectral preservation gives
\[
\operatorname{Ground}(\widetilde H)
=
V\operatorname{Ground}(K).
\]
Hence the unique ground state of the compiled Hamiltonian is
\[
\Ket{\widetilde \phi}
=
\Ket{\phi'}\otimes\Ket{\Omega}_A.
\]
The corresponding guiding state
\[
\Ket{\widetilde u}
=
\Ket{u}\otimes\Ket{\Omega}_A
\]
is again a semi-classical subset state, since
$\Ket{\Omega}_A$ is a computational-basis state, and
\[
|\Braket{\widetilde u|\widetilde\phi}|
=
|\Braket{u|\phi'}|
\ge
1-\frac{2}{q(n)}.
\]

Since $K=QH'$, its spectral gap is at least
\[
Q\left(\Delta_{\mathrm{sp}}-2\varepsilon\right).
\]
Therefore \cref{thm:unpinned-spectral-preservation} gives
\[
\widetilde\Delta_{\mathrm{sp}}
\ge
\min\left\{
Q\left(\Delta_{\mathrm{sp}}-2\varepsilon\right),
1
\right\}.
\]
Since $Q\ge1$ and
$\varepsilon\le\Delta_{\mathrm{sp}}/4$, we have
\[
\widetilde\Delta_{\mathrm{sp}}
\ge
\min\left\{
\frac{\Delta_{\mathrm{sp}}}{2},
1
\right\},
\]
which is inverse polynomial.

Finally, the approximation
\[
\left\|
H-\frac1QK
\right\|
\le\varepsilon
\]
transforms the source thresholds to
\[
\widetilde a
=
Q(a+\varepsilon),
\qquad
\widetilde b
=
Q(b-\varepsilon).
\]
The exact compiler preserves these thresholds, and
\[
\widetilde b-\widetilde a
=
Q\bigl((b-a)-2\varepsilon\bigr)
\ge
\frac{Q}{2}(b-a)
\ge
\frac{1}{\operatorname{poly}(n)}.
\]
Let $n_{\mathrm{out}}$ denote the final compiled system size.
Since
\[
n\le n_{\mathrm{out}}\le\operatorname{poly}(n),
\]
the promise gap and spectral gap remain inverse polynomial in
$n_{\mathrm{out}}$, while
\[
1-\frac{2}{q(n)}
=
1-\frac{1}{\operatorname{poly}(n_{\mathrm{out}})}
\]
after redefining the polynomial.  Thus all promises of
Guided-SIH are satisfied with respect to the final physical system
size.

Thus $\mathrm{Guided\text{-}SIH}_{9}$ is BQP-hard and hence
BQP-complete.

\end{proof}

\section{Discussion and Open Problems}
\label{sec:discussion}

The results of this work show that substantial Hamiltonian complexity
survives an unusually strong form of interaction uniformity. Once the
local interaction is fixed, all nontrivial instance dependence is
carried by the ordered interaction hypergraph: neither the local matrix
nor its coupling strength varies from term to term.

Nevertheless, \cref{thm:sih3-qma-complete} gives QMA-completeness
already for a fixed three-local interaction, while
\cref{thm:stoq-sih3-stoqma-complete,thm:pinned-stoq-sih4} show that the
same restriction is compatible with both StoqMA-complete and pinned
QMA-complete stoquastic problems. Together with the $n$-dependent and
geometrically local constructions, these results indicate that
heterogeneous local interactions are not, by themselves, an essential
source of hardness in Local Hamiltonian problems.

\paragraph{Single interactions as a normal form.}
The general compiler of
\cref{sec:general_methodology} separates two distinct issues.  For a
fixed finite interaction alphabet, the compilation to one interaction
is exact and preserves the low-energy spectrum below a freely chosen
separation scale.  For arbitrary local Hamiltonians, the additional
finite-alphabet normalization of
\cref{sec:universal-finite-alphabet-normalization} introduces an
inverse-polynomial approximation and a known global energy scale.
This distinction is important when transferring structural promises:
the exact compiler preserves ground spaces exactly and preserves
spectral gaps exactly whenever they lie below the chosen separation
scale, whereas the normalization by multiplicity need not preserve
frustration freeness, geometric locality, or stoquasticity.
The consequences in \cref{sec:further-consequences} illustrate both
regimes.  In particular, the Hamiltonian quantum PCP conjecture admits
the single-interaction normal form of
\cref{cor:sih-qpcp-normal-form}, while the frustration-free, precise,
and guided settings require more specialized uses of the compilation
framework.

\paragraph{Low locality and fixed-interaction classification.}
The QMA-complete $\mathrm{SIH}_3$ construction leaves a natural
locality boundary unresolved by the present work: the complexity of
strict unweighted $\mathrm{SIH}_2$. This regime already contains
physically and algorithmically important special cases. In particular,
up to an overall constant normalization and the max-versus-min
convention, unweighted Quantum Max-Cut is a strict fixed-interaction
two-local problem generated by repeated copies of the singlet
projector. By contrast, weighted Quantum Max-Cut permits
edge-dependent coupling strengths and therefore belongs to the
weighted fixed-interaction framework discussed in
\cref{subsubsec:s-hamiltonians}.

More generally, it would be useful to classify the complexity of
$\mathrm{SIH}(h)$ as $h$ ranges over fixed two-qubit interactions.
Existing interaction-set classifications allow freedoms, such as
independently variable coupling strengths, that are absent in the
strict SIH model. The fixed-interaction hardness results established
here begin at locality three and therefore do not resolve this
two-local classification problem.

\paragraph{Geometric locality and multiplicity.}
The multiplicity-based unweighting used in the $\mathrm{SIH}_3$ construction
does not automatically preserve geometric locality.  Replacing a
coupling strength by repeated occurrences may increase the degree of
the participating work qubits by a polynomial factor, even when the
source Hamiltonian has bounded degree.  This is why the low-locality
$\mathrm{SIH}_3$ result does not subsume the geometrically local $\mathrm{SIH}_8$
construction, where the computation and the auxiliary checks are
distributed explicitly so that both interaction range and vertex
degree remain bounded.

This leaves open the problem of reducing the locality of geometrically
local SIH while retaining strict unweighted interactions.  More
generally, it would be interesting to find alternatives to
coefficient-to-multiplicity normalization that preserve bounded degree
and geometric locality with only constant overhead.

\paragraph{Stoquasticity and pinning.}
The stoquastic results highlight the role that pinning can play under
the single-interaction restriction.  We obtain a fixed three-local
stoquastic interaction whose unpinned $\mathrm{Stoq\text{-}SIH}_3$ problem is
StoqMA-complete, whereas the real-interaction embedding shows that a
single pinned qubit suffices to obtain a QMA-complete
$\mathrm{Stoq\text{-}SIH}_4$ problem. Thus, even when every local term is the same fixed stoquastic
interaction, pinning a single auxiliary qubit is sufficient for
QMA-hardness.

This leaves a natural locality question.  In particular, the present
construction does not determine whether QMA-hardness can already occur
for a $1$-Pinned $\mathrm{Stoq\text{-}SIH}_3$ problem, or whether locality four is
intrinsic to this particular embedding strategy.  More generally, it
would be useful to understand how the computational power of pinning
depends on the locality and algebraic structure of the fixed
stoquastic interaction.

\paragraph{Fixed versus size-dependent interactions.}
The distinction between SIH and DSIH isolates another form of
uniformity.  In DSIH the common interaction may depend on the final
physical system size, but not directly on the instance itself.  The
constructions in this work use this freedom in two different ways.  In
\cref{thm:pinned-n-dependent,thm:unpinned-n-dependent}, the final system
size encodes the parameters required to reproduce the localized
fixed-Hamming-weight XY Hamiltonian.  In the geometrically local
stoquastic construction of \cref{thm:StoqMA-complete}, the dependence
is substantially weaker: all entries of the interaction are constant
except for the perturbative coefficient
\[
\delta_n=n^{-8}.
\]

This raises the question of how essential such size dependence is.
Can the DSIH hardness constructions be replaced by fixed-interaction
SIH constructions with comparable locality and boundary conditions?
In particular, it would be interesting to determine whether the
system-size encoding used in the XY-based construction can be replaced
by an instance-independent mechanism, and whether the perturbative
scale in the geometrically local stoquastic construction can be
generated using a fixed interaction rather than being built directly
into $h^{(n)}$.  More generally, understanding when size dependence can
be eliminated would clarify the computational distinction between SIH
and DSIH.

\paragraph{Beyond ground-energy decision problems.}
The exact finite-alphabet compiler preserves substantially more than
the minimum eigenvalue.  By
\cref{thm:unpinned-spectral-preservation}, the valid auxiliary sector
reproduces the source Hamiltonian exactly, while the spectrum below a
controlled separation scale is preserved together with multiplicities.
This suggests studying single-interaction versions of problems whose
promise or objective depends on richer low-energy spectral information,
such as spectral-gap questions, ground-space degeneracy, the density of
low-energy states, or specified excited-state energies. Such extensions
would require careful choice of the source problem and separation scale,
but the exact compiler provides a natural starting point.

The same observation also raises a stronger simulation question.  The
universal reduction of
\cref{thm:universal-single-interaction-reduction} is designed as a
complexity-theoretic reduction for ground-energy decision problems:
after finite-alphabet normalization, it reproduces the relevant
energies only up to a controlled approximation and a known global
scale.  It does not by itself constitute a general analogue-simulation
theorem preserving encoded eigenstates, local observables, or the full
spectrum.  Determining whether similarly severe single-interaction
restrictions are compatible with such stronger simulation guarantees
would extend the present framework beyond decision-problem
complexity.

Taken together, these results suggest that much of the computational
content of Local Hamiltonian problems need not reside in a heterogeneous
collection of local matrices or in independently tunable coupling
strengths.  In the single-interaction setting, substantial complexity
can instead be encoded into the interaction hypergraph, auxiliary
structure, boundary conditions, or, in the DSIH setting, a uniformly
recoverable dependence on the final system size.  The exact
finite-alphabet compiler makes this principle particularly explicit:
once the relevant local interactions are drawn from a fixed finite
alphabet, they can be collapsed to a single interaction while retaining
the controlled low-energy structure needed by a range of complexity
constructions.

The remaining questions concern the limits of this compression.
Understanding how far locality can be reduced, when geometric or
stoquastic structure can be retained, whether size dependence or
pinning can be eliminated, and which stronger spectral or simulation
properties survive the single-interaction restriction would help
clarify which ingredients of a many-body Hamiltonian are genuinely
responsible for computational hardness.

\section*{Acknowledgements}
\addcontentsline{toc}{section}{Acknowledgements}

The author thanks Prof.~Dorit Aharonov for her guidance, valuable discussions, and feedback throughout the development of this work.

\section*{Use of AI Tools}

The mathematical ideas, constructions, proofs, and conclusions
presented in this work were developed by the author. OpenAI's ChatGPT and Google's Gemini
were used as assistive tools for checking mathematical arguments and
internal consistency, identifying possible gaps and presentation
issues, reviewing the manuscript's positioning relative to the literature and citation consistency, and
improving the clarity, organization, and language of the manuscript.
They were used to review and refine material produced by the author rather
than as authoritative sources or substitutes for the author's
mathematical reasoning. The author remained responsible for developing,
verifying, and presenting the mathematical content, as well as for the
accuracy of all citations and the final form of the manuscript.

\appendix

\section{Geometrically Local Source Hamiltonians}
\label{app:geometric-sources}

\subsection{The Geometrically Local QMA Source}
\label{app:geometric-source}

We now record the particular spatial construction used in the
paper. The reduction first compiles the verifier to a one-dimensional
nearest-neighbor circuit, then unrolls its execution across a
two-dimensional lattice, and finally places the unary clock along the
resulting spatial computation path.

\paragraph{Structural assumptions on the input circuit.}
We may assume that the verifier acts on a one-dimensional array of
$N$ qubits using nearest-neighbor gates and nearest-neighbor SWAP
operations. Any polynomial-size verifier can be compiled into this
form with only polynomial overhead: a gate acting on nonadjacent qubits
is implemented by routing its operands together with SWAPs, applying
the gate, and routing them back.

\paragraph{Two-dimensional lattice embedding.}
To ensure that the resulting Hamiltonian satisfies
\cref{def: geometrically local}, we embed the one-dimensional register
of $N$ qubits into a two-dimensional lattice consisting of
$L=NR$ physical qubits arranged in $R$ columns of height $N$. The
execution is divided into
\[
R=\operatorname{poly}(|x|)
\]
computational rounds, each containing one nontrivial nearest-neighbor
gate together with the required identity operations. The computation
proceeds through the lattice as follows:

\begin{enumerate}
    \item \textbf{Gate execution.}
    During round $i$, the operations for that round are applied along
    column $i$, beginning at the top of the column.

    \item \textbf{Spatial translation.}
    After the round is complete, the logical state of the $N$ qubits is
    transferred to column $i+1$ by nearest-neighbor SWAP operations.
    The SWAP sequence begins at the bottom of the column and proceeds
    upward.

    \item \textbf{Next round.}
    Once the transfer is complete, round $i+1$ begins at the top of the
    new column.
\end{enumerate}

The resulting sequence of operations follows a snaking path through the
two-dimensional grid. Each physical work qubit participates in at most
one computational operation and at most two SWAP operations, so its
circuit participation is bounded by a constant. Consecutive operations
also occur within constant geometric distance. This is the spatial
structure used to localize the work--clock interactions below.

\paragraph{Geometrically local Hamiltonian.}
Let the spatialized verifier contain $R$ computational rounds on columns
of height $N$. The resulting two-dimensional work register contains

\[
L=RN
\]

physical qubits, and the snaking circuit contains

\[
T=(2R-1)N
\]

local operations, including computational gates, identities, and SWAPs.
By padding with one initial and one final identity round, we may assume
without loss of generality that the first and last computational rounds
contain only identity operations. This changes the circuit size by only
$O(N)$.

\paragraph{Spatial embedding of the clock register.}
To ensure that the resulting Hamiltonian terms remain geometrically local,
we specify the physical placement of the $T$-qubit clock register. The
clock qubits are distributed along the same snaking path followed by the
spatial computation. Each logical time step
$t\in\{1,\ldots,T\}$ is associated with a position on this path, and the
corresponding clock qubit $C_t$ is placed within constant distance of the
work qubits involved in the operation $U_t$. Consecutive clock qubits
$C_t$ and $C_{t+1}$ are likewise placed within constant distance of one
another. Hence the propagation and clock terms associated with
neighboring time steps have bounded geometric range.

With this placement established, the geometrically local Hamiltonian is constructed as a sum of four terms, similar to Kitaev's original formulation:
\begin{align}
    H^{\mathrm{geo}} = H^{\mathrm{geo}}_{\mathrm{in}} + H^{\mathrm{geo}}_{\mathrm{out}} + H^{\mathrm{geo}}_{\mathrm{prop}} + H^{\mathrm{geo}}_{\mathrm{clock}}
\end{align}

Let $I_0$ and $I_1$ denote the sets of physical work qubits in the
initial computational column whose prescribed states are $\Ket{0}$ and
$\Ket{1}$, respectively. Thus $|I_0|=n_0$ and $|I_1|=n_1$, while the
witness qubits are unconstrained. For every initialized work qubit
$W_i\in I_0\cup I_1$, let $t_i$ denote the earliest time step at which
$W_i$ participates in a nonidentity operation, with SWAP counted as
nonidentity.

By padding the spatial circuit with an initial and a final identity
round, we may assume that $ 2\leq t_i\leq T-1 $ for every initialized qubit. Since every operation acting on $W_i$ before time $t_i$ is the identity, its state immediately before $U_{t_i}$ is equal to its prescribed initial state.

In the valid unary-clock subspace, the projector $ \Pi^{\Ket{1}}_{C_{t_i-1}} \otimes \Pi^{\Ket{0}}_{C_{t_i}} $ selects precisely the computational snapshot immediately before $U_{t_i}$. Hence checking the prescribed state of $W_i$ at this clock position is equivalent to checking its state at the beginning of the computation, while keeping the interaction geometrically local.

These terms act as follows:

\paragraph{Propagation.}
The propagation Hamiltonian
\[
H^{\mathrm{geo}}_{\mathrm{prop}}
=
\sum_{t=1}^{T} H_{\mathrm{prop}}(t)
\]
has the same form as in Kitaev's construction; see
\cref{eq:hprop}. Because the spatialized circuit uses only local gates
and nearest-neighbor SWAP operations, each propagation term has
constant geometric range.

\paragraph{Clock penalty.}
The clock Hamiltonian $H^{\mathrm{geo}}_{\mathrm{clock}}$ has the
standard unary-clock form of \cref{eq:hclock}. Because the clock qubits
are placed along the snaking computation path, each clock penalty acts
only on neighboring clock sites.

\paragraph{Initialization.}
The initialization Hamiltonian checks the prescribed
computational-basis states of the initialized input qubits at the local
times defined above:
\begin{align}
    \label{eq:geo-in}
    H^{\mathrm{geo}}_{\mathrm{in}}
    ={}&
    \sum_{i\in I_0}
    \Pi^{\Ket{1}}_{W_i}
    \otimes
    \Pi^{\Ket{1}}_{C_{t_i-1}}
    \otimes
    \Pi^{\Ket{0}}_{C_{t_i}}
    \nonumber\\
    &+
    \sum_{i\in I_1}
    \Pi^{\Ket{0}}_{W_i}
    \otimes
    \Pi^{\Ket{1}}_{C_{t_i-1}}
    \otimes
    \Pi^{\Ket{0}}_{C_{t_i}}.
\end{align}

\paragraph{Output.}
Without loss of generality, we route the designated output qubit to the
bottommost site of the final column. The final output check is therefore
within constant distance of the clock site corresponding to $t=T$.
\begin{align}
    H^{\mathrm{geo}}_{\mathrm{out}} = \Pi^{\Ket{0}}_{W_{L}} \otimes\Pi^{\Ket{1}}_{C_T}
\end{align}

\begin{rem}
\label{rem:localized-initialization}
Our initialization Hamiltonian differs slightly from the one used by
Oliveira and Terhal~\cite{Oliveira2005}.

First, we check an initialized qubit $W_i$ immediately before its first
nonidentity operation rather than at the first position at which the
spatial cursor encounters it. By definition of $t_i$, every earlier
operation acting on $W_i$ is the identity. Hence, on any
propagation-consistent history state, the state of $W_i$ at time
$t_i-1$ is exactly its initial state. The localized check therefore
enforces the same input condition as an initialization check at time
zero.

Second, Oliveira and Terhal use a three-clock projector of the form
\[
\Pi^{\Ket{100}}_{C_{t_i-1}C_{t_i}C_{t_i+1}},
\]
whereas we use
\[
\Pi^{\Ket{10}}_{C_{t_i-1}C_{t_i}}.
\]
On the valid unary-clock subspace, both projectors select precisely the
history time $t_i-1$. Moreover,
\[
\Pi^{\Ket{10}}_{C_{t_i-1}C_{t_i}}
\otimes I_{C_{t_i+1}}
=
\Pi^{\Ket{100}}_{C_{t_i-1}C_{t_i}C_{t_i+1}}
+
\Pi^{\Ket{101}}_{C_{t_i-1}C_{t_i}C_{t_i+1}}
\succeq
\Pi^{\Ket{100}}_{C_{t_i-1}C_{t_i}C_{t_i+1}}.
\]
Thus the two-clock check has the same action on legal clock states and
does not weaken the energetic penalty on illegal clock configurations.
\end{rem}

\paragraph{Completeness and soundness.}
It remains to verify that the localized Hamiltonian retains an
inverse-polynomial energy separation between YES and NO instances.

\begin{lemma}[Completeness]
\label{lemma: OT completness}
If $x\in L_{\mathrm{yes}}$, there exists a history state
$\Ket{\eta}$ such that
\[
\Bra{\eta}H^{\mathrm{geo}}\Ket{\eta}
\leq
\frac{\varepsilon}{T+1}.
\]
\end{lemma}

\begin{proof}
Choose an accepting witness and its corresponding history state.
The propagation and clock terms annihilate this state. By the definition
of $t_i$, every initialized qubit retains its prescribed initial value
until the localized initialization check is applied, so
$H^{\mathrm{geo}}_{\mathrm{in}}$ also annihilates the history state.
The only possible energy contribution is therefore the output term,
whose expectation is at most $\varepsilon/(T+1)$.
\end{proof}

\begin{lemma}[Soundness, adapted from \cite{Oliveira2005}]
\label{lemma: OT soundness}
If $x\in L_{\mathrm{no}}$, then for every state $\Ket{\Psi}$,
\[
\Bra{\Psi}H^{\mathrm{geo}}\Ket{\Psi}
\geq
\frac{c}{T^3}
\left(
1-\varepsilon-\sqrt{\varepsilon}
\right)
\]
for some constant $c>0$ independent of $T$.
\end{lemma}

\begin{proof}
The Oliveira--Terhal soundness proof applies to the present construction
with the modified initialization term described in
\cref{rem:localized-initialization}. By definition of $t_i$, no operation
before time $t_i$ acts nontrivially on an initialized qubit $W_i$, so its
localized initialization check enforces the same input condition used in
their argument. Moreover, by \cref{rem:localized-initialization}, our
two-clock projector has the same action as their three-clock projector on
the legal unary-clock subspace and can only add positive semidefinite
penalty outside it. Hence the angle estimate in the Oliveira--Terhal
soundness proof is unaffected, and their lower bound

\[
\frac{c}{T^3}
\left(
1-\varepsilon-\sqrt{\varepsilon}
\right)
\]

continues to hold.
\end{proof}

Thus the construction has bounded interaction range and bounded degree,
while the completeness and soundness bounds retain an inverse-polynomial
promise gap.

\subsection{The Geometrically Local Stoquastic Source}
\label{app:geometric-stoq-source}

\paragraph{Spatialized StoqMA verifier.}

Let $L\in\mathrm{StoqMA}$ and let $U_x$ be a verifier as in
\cref{def: StoqMA}. Following Waite and Bremner~\cite{Waite2025}, we use
the same two-dimensional column-and-SWAP spatialization and snaking-clock
placement as in \cref{app:geometric-source}. The nontrivial
computational operations act on at most three consecutive work qubits,
while the routing operations are nearest-neighbor SWAPs. Since these
operations are permutation matrices, the resulting propagation
Hamiltonian remains stoquastic.

Let $I_0$, $I_1$, and $I_+$ denote the physical work qubits whose
prescribed initial states are respectively
$\Ket{0}$, $\Ket{1}$, and $\Ket{+}$; the witness qubits remain
unconstrained. The spatialization preserves the verifier completeness
and soundness parameters $\epsilon_{\mathrm{yes}}$ and
$\epsilon_{\mathrm{no}}$.

As in \cref{app:geometric-source}, we pad with initial and final identity
rounds and, for every
$W_i\in I_0\cup I_1\cup I_+$, let $t_i$ denote its first nonidentity
time step, with SWAP counted as nonidentity. Then
\[
2\leq t_i\leq T-1,
\]
and the state of $W_i$ immediately before $U_{t_i}$ is still its
prescribed initial state. The clock qubits are placed along the same
snaking computation path, so the corresponding localized initialization
checks have constant geometric range.

We construct
\begin{equation}
\label{eq:geo-stoq-hamiltonian}
H^{\mathrm{stoq}_g}
=
H^{\mathrm{stoq}_g}_{\mathrm{in}}
+
H^{\mathrm{stoq}_g}_{\mathrm{prop}}
+
H^{\mathrm{stoq}_g}_{\mathrm{clock}}
+
\delta H^{\mathrm{stoq}_g}_{\mathrm{out}},
\end{equation}
where the four components are as follows.

\paragraph{Initialization.}
The initialized physical work qubits are checked immediately before their
first nonidentity operation:
\begin{align}
H^{\mathrm{stoq}_g}_{\mathrm{in}}
={}& 
\sum_{i\in I_0}
\Pi^{\Ket{1}}_{W_i}
\otimes
\Pi^{\Ket{1}}_{C_{t_i-1}}
\otimes
\Pi^{\Ket{0}}_{C_{t_i}}
\nonumber\\
&+
\sum_{i\in I_1}
\Pi^{\Ket{0}}_{W_i}
\otimes
\Pi^{\Ket{1}}_{C_{t_i-1}}
\otimes
\Pi^{\Ket{0}}_{C_{t_i}}
\nonumber\\
&+
\sum_{i\in I_+}
\Pi^{\Ket{-}}_{W_i}
\otimes
\Pi^{\Ket{1}}_{C_{t_i-1}}
\otimes
\Pi^{\Ket{0}}_{C_{t_i}}.
\label{eq:geo-stoq-in}
\end{align}
The projector $\Pi^{\Ket{-}}$ has non-positive off-diagonal entries, so
all initialization terms are stoquastic.

The localized-initialization argument of
\cref{rem:localized-initialization} applies directly. In particular, the
two-clock selector identifies the same legal history time as the
three-clock selector used in the spatially sparse construction, while the
definition of $t_i$ guarantees that the checked work qubit still contains
its prescribed initial state.

\paragraph{Propagation.}
The propagation Hamiltonian has the standard circuit-to-Hamiltonian form

\[
H^{\mathrm{stoq}_g}_{\mathrm{prop}}
=
\sum_{t=1}^{T}
H_{\mathrm{prop}}(t).
\]

Every circuit operation is a permutation matrix, so the off-diagonal
entries contributed by the propagation terms are non-positive. A
three-qubit computational gate together with the local unary-clock check
acts on at most six qubits, which determines the $6$-locality of the
construction.

\paragraph{Clock.}
The clock Hamiltonian is the usual unary-clock penalty

\[
H^{\mathrm{stoq}_g}_{\mathrm{clock}}
=
\sum_{t=1}^{T-1}
\Pi^{\Ket{0}}_{C_t}
\otimes
\Pi^{\Ket{1}}_{C_{t+1}}.
\]

It is diagonal and therefore stoquastic.

\paragraph{Output.}
Without loss of generality, the designated output qubit is routed to the
bottommost work site of the final column, denoted $W_L$, so that it lies
within constant distance of the final clock site. We define
\begin{equation}
\label{eq:geo-stoq-out}
H^{\mathrm{stoq}_g}_{\mathrm{out}}
=
\Pi^{\Ket{-}}_{W_L}
\otimes
\Pi^{\Ket{1}}_{C_T}.
\end{equation}
The coefficient $\delta>0$ in
\cref{eq:geo-stoq-hamiltonian} is chosen sufficiently small relative to
the spectral gap of the unperturbed history-state Hamiltonian.

Each term in \cref{eq:geo-stoq-hamiltonian} has constant geometric
diameter. Moreover, the spatial circuit construction ensures that every
physical work qubit participates in only $O(1)$ circuit operations, and
the snaking clock construction gives the same bounded-participation
property for the clock qubits. Hence
$H^{\mathrm{stoq}_g}$ satisfies both conditions of
\cref{def: geometrically local}. Every local term is also stoquastic.

\paragraph{Energy bounds.}
Let

\[
H_0
=
H^{\mathrm{stoq}_g}_{\mathrm{in}}
+
H^{\mathrm{stoq}_g}_{\mathrm{prop}}
+
H^{\mathrm{stoq}_g}_{\mathrm{clock}}.
\]

The spatially sparse analysis of Waite and Bremner~\cite{Waite2025}
shows that $H_0$ has a history-state ground space separated from the rest
of its spectrum by
\begin{equation}
\label{eq:geo-stoq-unperturbed-gap}
\Gamma_T
=
\Omega(T^{-3}).
\end{equation}

\begin{lemma}[Completeness, adapted from \cite{Waite2025}]
\label{lemma: geometric StoqMA completness}
If $x\in L_{\mathrm{yes}}$, then there exists a history state
$\Ket{\eta_{\mathrm{stoq}}}$ satisfying

\[
\Bra{\eta_{\mathrm{stoq}}}
H^{\mathrm{stoq}_g}
\Ket{\eta_{\mathrm{stoq}}}
\leq
\frac{\delta}{T+1}
\left(1-\epsilon_{\mathrm{yes}}\right).
\]

\end{lemma}

\begin{proof}
Choose an accepting witness and the corresponding propagation-consistent
history state. The initialization, propagation, and clock Hamiltonians
annihilate this state. The only possible contribution is the output
penalty. The final clock time occurs with probability $1/(T+1)$, while
the probability that the output qubit is in $\Ket{-}$ is at most
$1-\epsilon_{\mathrm{yes}}$. Therefore

\[
\Bra{\eta_{\mathrm{stoq}}}
H^{\mathrm{stoq}_g}
\Ket{\eta_{\mathrm{stoq}}}
\leq
\frac{\delta}{T+1}
\left(1-\epsilon_{\mathrm{yes}}\right).
\]

\end{proof}

For soundness, Waite and Bremner~\cite{Waite2025} treat
$\delta H^{\mathrm{stoq}_g}_{\mathrm{out}}$ as a perturbation of $H_0$.
Using the gap
$\Gamma_T=\Omega(T^{-3})$, choosing $\delta$ sufficiently small relative
to $\Gamma_T$ preserves an inverse-polynomial separation between the YES
and NO ground-energy thresholds.

\section{Technical Routing Details for Geometric Reductions}
\label{app:geometric-routing}

\subsection{Explicit Routing for the Geometrically Local \texorpdfstring{$\mathrm{SIH}_8$}{SIH8} Reduction}
\label{app:geometric-sih-routing}

\paragraph{Initialization.}
Using the notation and routing conventions of
\cref{subsec:unpinned_geometric}, the initialization Hamiltonian is

\begin{align}
        \widetilde{H}^{\mathrm{geo}}_{\mathrm{in}} = &\sum_{i\in I_0}\sum_{j=0}^{3} h(C_{t_{i}-1}, C_{t_i}, A^{(t_i)}_{2j} ,A^{(t_i)}_{2j+2} , W_i, A^{(t_i)}_{2j+4},  A^{(t_i)}_{2j+6} , A^{(t_i)}_{2j+7}) \nonumber \\
        + &\sum_{i \in I_1}\sum_{j=0}^{3} h(C_{t_{i}-1},C_{t_i},W_i, A^{(t_i)}_{2j} ,A^{(t_i)}_{2j+3} ,  A^{(t_i)}_{2j+4},  A^{(t_i)}_{2j+6} , A^{(t_i)}_{2j+7}) 
\end{align}

\paragraph{Output.}
The designated output qubit is routed to the bottommost site of the
final column, so the output penalty remains geometrically local to the
final clock site at $t=T$.

 \begin{align}
        \widetilde{H}^{\mathrm{geo}}_{\mathrm{out}} = \sum_{j=0}^{3} h(C_{T}, W_L, A^{(T)}_{2j}, A^{(T)}_{2j+2}, A^{(T)}_{2j+3}, A^{(T)}_{2j+4},  A^{(T)}_{2j+6} , A^{(T)}_{2j+7})
\end{align}

\paragraph{Clock Penalty.} The clock penalty ensures that the clock register remains in the valid unary subspace. Because the clock qubits are physically embedded along the continuous snaking path of the computation, sequential clock qubits are adjacent, preserving geometric locality.

\begin{align}
    \widetilde{H}^{\mathrm{geo}}_{\mathrm{clock}} = \sum_{t=1}^{T-1}\sum_{j=0}^{3} h(C_{t+1},C_{t}, A^{(t)}_{2j}, A^{(t)}_{2j+2}, A^{(t)}_{2j+3}, A^{(t)}_{2j+4}, A^{(t)}_{2j+6} , A^{(t)}_{2j+7})
\end{align}

\paragraph{Propagation.} The propagation term ensures that the computation evolves according to the circuit. Because the spatialized circuit uses only local gates and
nearest-neighbor SWAP operations, the work and clock qubits appearing
in each propagation term lie within a constant-diameter neighborhood. Let $(W_1^t, W_2^t)$ be the target qubits for the $t$-th gate.  

\begin{align}
    \widetilde{H}^{\mathrm{geo}}_{\mathrm{prop}} = \sum_{t=1}^T \widetilde{H}^{\mathrm{geo}}_{\mathrm{prop}}(t)
\end{align}

For intermediate time steps $2 \leq t \leq T-1$, we have three options depending on the applied gate $U_t$:
\begin{itemize}
    
    \item[] For $t$ such that $U_t=\mathbb{B}$:
    \begin{align}
            \widetilde{H}^{\mathrm{geo}}_{\mathrm{prop}}(t) = & \sum_{j=0}^{3} h(W^{t}_{1},W^{t}_{2}, C_t, C_{t+1}, C_{t-1}, A^{(t)}_{2j+1}, A^{(t)}_{2j+6}, A^{(t)}_{2j+7}) 
    \end{align}

    \item[] For $t$ such that $U_t=\mathrm{SWAP}$:
    \begin{align}
            \widetilde{H}^{\mathrm{geo}}_{\mathrm{prop}}(t) = & \sum_{j=0}^{3} h(W^{t}_{1},W^{t}_{2}, C_t, A^{(t)}_{2j+1}, C_{t-1}, C_{t+1}, A^{(t)}_{2j+6}, A^{(t)}_{2j+7}) 
    \end{align} 

    \item[] For $t$ such that $U_t=I$:
    \begin{align}
            \widetilde{H}^{\mathrm{geo}}_{\mathrm{prop}}(t) = & \sum_{j=0}^{3} h(W^t_1, C_t, C_{t-1}, C_{t+1}, A^{(t)}_{2j}, A^{(t)}_{2j+2}, A^{(t)}_{2j+6}, A^{(t)}_{2j+7}) 
    \end{align} 

\end{itemize}
    
By construction, $U_1=U_T=I$. At the two endpoints we therefore use

\begin{align}
        \widetilde{H}^{\mathrm{geo}}_{\mathrm{prop}}(1) = & \sum_{j=0}^{3} h(W_1, C_1, A^{(1)}_{2j+1}, C_2, A^{(1)}_{2j+2}, A^{(1)}_{2j+4}, A^{(1)}_{2j+6}, A^{(1)}_{2j+7}) \\
        \widetilde{H}^{\mathrm{geo}}_{\mathrm{prop}}(T)= & \sum_{j=0}^{3} h(W_L, C_T, C_{T-1}, A^{(T)}_{2j}, A^{(T)}_{2j+2}, A^{(T)}_{2j+4}, A^{(T)}_{2j+6}, A^{(T)}_{2j+7}) 
\end{align}

\paragraph{Geometric locality and bounded degree.}

We first verify bounded interaction range. Begin with the lattice
embedding of $H^{\mathrm{geo}}$. Every target term of
$H^{\mathrm{geo}}$ is supported within a constant-diameter
neighborhood. We rescale this lattice by a sufficiently large constant
factor and, for each time step $t$, place the eight qubits of $A^{(t)}$
in a constant-radius neighborhood of the clock qubit $C_t$. Since the
rescaling factor and the size of each auxiliary block are constants
independent of the input size, the auxiliary blocks can be embedded
injectively into the lattice without collisions.

Every application of $h$ associated with time $t$ acts only on qubits
belonging to the constant-diameter support of the corresponding term of
$H^{\mathrm{geo}}$, together with qubits from the local block
$A^{(t)}$. Hence there exists a constant, independent of the instance
size, that bounds the distance between any two qubits appearing in a
single application of $h$. Thus the resulting interaction hypergraph has
bounded interaction range.

We next verify bounded degree. Every work and clock qubit participates
in only $O(1)$ terms of $H^{\mathrm{geo}}$. Each such target term is
replaced by exactly four applications of $h$, so the degree of every
work and clock qubit increases by at most a constant factor.

For the auxiliary qubits, fix a time step $t$. The block $A^{(t)}$ is
used in four propagation applications and, when $t\leq T-1$, four
clock-penalty applications. It is additionally used in four
output-penalty applications only when $t=T$.

Initialization contributes only a constant number of additional
applications. Indeed, if $t_i=t$, then $W_i$ participates for the first
time in a nonidentity operation at time $t$. Since each nonidentity
operation $U_t$ acts on at most two work qubits, at most two initialized
qubits can satisfy $t_i=t$. Each corresponding initialization term uses
four applications of $h$. Therefore the total number of applications of
$h$ involving any fixed auxiliary block $A^{(t)}$ is bounded by a
constant, and hence every auxiliary qubit also has bounded degree.

Consequently, $\widetilde H^{\mathrm{geo}}$ has both bounded interaction
range and bounded degree and is geometrically local in the sense of
\cref{def: geometrically local}.

The ordered $8$-tuples produced by the reduction are pairwise distinct,
and every tuple contains eight distinct physical qubits. For a fixed
target term, different values of $j$ use different ordered auxiliary
check pairs, while applications associated with different time steps use
different auxiliary blocks $A^{(t)}$. Distinct target terms at the same
time step differ in at least one work or clock position, and distinct
initialization terms differ in the position containing the initialized
work qubit $W_i$. Finally, the work, clock, and auxiliary registers are
disjoint, the work qubits of each two-qubit gate are distinct, and the
displayed auxiliary indices are pairwise distinct modulo $8$. Hence the
construction satisfies the ordered-hyperedge requirements of
\cref{def: SIH}.

\subsection{Explicit Routing for the Geometrically Local \texorpdfstring{$\mathrm{Stoq\text{-}DSIH}_9$}{Stoq-DSIH9} Reduction}
\label{app:stoq-dsih-routing}

In the expressions below, auxiliary indices are evaluated modulo $8$.
For each initialized physical work qubit
$W_i\in I_0\cup I_1\cup I_+$, we use the time $t_i$ defined in
\cref{app:geometric-stoq-source}: namely, the first time at which
$W_i$ participates in a nonidentity operation, with $\mathrm{SWAP}$
counted as nonidentity.

\paragraph{Initialization.}
For the three possible prescribed initial states, define
\begin{align}
\widetilde H^{\mathrm{stoq}_g}_{\mathrm{in}}
={}&
\sum_{i\in I_0}
\sum_{j=0}^{3}
h^{(n)}
\left(
A^{(t_i)}_{2j},
C_{t_i},
C_{t_i-1},
W_i,
A^{(t_i)}_{2j+2},
A^{(t_i)}_{2j+3},
A^{(t_i)}_{2j+4},
A^{(t_i)}_{2j+6},
A^{(t_i)}_{2j+7}
\right)
\nonumber\\
&+
\sum_{i\in I_1}
\sum_{j=0}^{3}
h^{(n)}
\left(
W_i,
C_{t_i},
A^{(t_i)}_{2j+1},
C_{t_i-1},
A^{(t_i)}_{2j+2},
A^{(t_i)}_{2j+3},
A^{(t_i)}_{2j+4},
A^{(t_i)}_{2j+6},
A^{(t_i)}_{2j+7}
\right)
\nonumber\\
&+
\sum_{i\in I_+}
\sum_{j=0}^{3}
h^{(n)}
\left(
C_{t_i-1},
W_i,
C_{t_i},
A^{(t_i)}_{2j},
A^{(t_i)}_{2j+2},
A^{(t_i)}_{2j+3},
A^{(t_i)}_{2j+4},
A^{(t_i)}_{2j+6},
A^{(t_i)}_{2j+7}
\right).
\label{eq:stoq-dsih-in}
\end{align}

On the valid auxiliary state these three routings reproduce,
respectively,
$
\Pi^{\Ket{1}}_{W_i}
\Pi^{\Ket{1}}_{C_{t_i-1}}
\Pi^{\Ket{0}}_{C_{t_i}},
\Pi^{\Ket{0}}_{W_i}
\Pi^{\Ket{1}}_{C_{t_i-1}}
\Pi^{\Ket{0}}_{C_{t_i}},
$
and
$
\Pi^{\Ket{-}}_{W_i}
\Pi^{\Ket{1}}_{C_{t_i-1}}
\Pi^{\Ket{0}}_{C_{t_i}}.
$

\paragraph{Output.}
Let $W_L$ be the designated output qubit at the bottom of the final
computational column, as in \cref{app:geometric-stoq-source}. Define
\begin{align}
\widetilde H^{\mathrm{stoq}_g}_{\mathrm{out}}
=
\sum_{j=0}^{3}
h^{(n)}
\left(
C_T,
W_L,
A^{(T)}_{2j},
A^{(T)}_{2j+1},
A^{(T)}_{2j+2},
A^{(T)}_{2j+3},
A^{(T)}_{2j+4},
A^{(T)}_{2j+6},
A^{(T)}_{2j+7}
\right).
\label{eq:stoq-dsih-out}
\end{align}
On the valid auxiliary state, the surviving branch is
$
\delta_n
\Pi^{\Ket{1}}_{C_T}
\otimes
\Pi^{\Ket{-}}_{W_L}.
$

\paragraph{Clock Penalty.}
Define
\begin{align}
\widetilde H^{\mathrm{stoq}_g}_{\mathrm{clock}}
=
\sum_{t=1}^{T-1}
\sum_{j=0}^{3}
h^{(n)}
\left(
A^{(t)}_{2j},
C_t,
C_{t+1},
A^{(t)}_{2j+1},
A^{(t)}_{2j+2},
A^{(t)}_{2j+3},
A^{(t)}_{2j+4},
A^{(t)}_{2j+6},
A^{(t)}_{2j+7}
\right).
\label{eq:stoq-dsih-clock}
\end{align}
On the valid auxiliary state this reproduces
$
\Pi^{\Ket{0}}_{C_t}
\otimes
\Pi^{\Ket{1}}_{C_{t+1}}.
$

\paragraph{Propagation.}
Let $(W_1^t,W_2^t,W_3^t)$ denote the work qubits involved in the
$t$-th operation, with unused positions omitted for SWAP and identity
operations. Define

\[
\widetilde H^{\mathrm{stoq}_g}_{\mathrm{prop}}
=
\sum_{t=1}^{T}
\widetilde H^{\mathrm{stoq}_g}_{\mathrm{prop}}(t).
\]

For an intermediate step $2\leq t\leq T-1$ with
$U_t=\mathrm{Toffoli}$, set
\begin{align}
\widetilde H^{\mathrm{stoq}_g}_{\mathrm{prop}}(t)
=
\sum_{j=0}^{3}
h^{(n)}
\left(
W_1^t,
W_2^t,
W_3^t,
C_t,
C_{t+1},
C_{t-1},
A^{(t)}_{2j+1},
A^{(t)}_{2j+6},
A^{(t)}_{2j+7}
\right).
\label{eq:stoq-dsih-prop-toffoli}
\end{align}

For $U_t=\mathrm{SWAP}$, set
\begin{align}
\widetilde H^{\mathrm{stoq}_g}_{\mathrm{prop}}(t)
=
\sum_{j=0}^{3}
h^{(n)}
\left(
W_1^t,
W_2^t,
C_t,
A^{(t)}_{2j},
A^{(t)}_{2j+1},
C_{t-1},
C_{t+1},
A^{(t)}_{2j+6},
A^{(t)}_{2j+7}
\right).
\label{eq:stoq-dsih-prop-swap}
\end{align}

For $U_t=I$, set
\begin{align}
\widetilde H^{\mathrm{stoq}_g}_{\mathrm{prop}}(t)
=
\sum_{j=0}^{3}
h^{(n)}
\left(
W_1^t,
C_t,
A^{(t)}_{2j+1},
C_{t-1},
C_{t+1},
A^{(t)}_{2j+2},
A^{(t)}_{2j+4},
A^{(t)}_{2j+6},
A^{(t)}_{2j+7}
\right).
\label{eq:stoq-dsih-prop-identity}
\end{align}

By construction, $U_1=U_T=I$. At the endpoints we therefore use
\begin{align}
\widetilde H^{\mathrm{stoq}_g}_{\mathrm{prop}}(1)
={}&
\sum_{j=0}^{3}
h^{(n)}
\left(
W_1^1,
C_1,
A^{(1)}_{2j},
A^{(1)}_{2j+1},
C_2,
A^{(1)}_{2j+2},
A^{(1)}_{2j+4},
A^{(1)}_{2j+6},
A^{(1)}_{2j+7}
\right) 
\label{eq:stoq-dsih-prop-first}\\
\widetilde H^{\mathrm{stoq}_g}_{\mathrm{prop}}(T)
={}&
\sum_{j=0}^{3}
h^{(n)}
\left(
W_1^T,
C_T,
A^{(T)}_{2j+1},
C_{T-1},
A^{(T)}_{2j},
A^{(T)}_{2j+2},
A^{(T)}_{2j+4},
A^{(T)}_{2j+6},
A^{(T)}_{2j+7}
\right).
\label{eq:stoq-dsih-prop-last}
\end{align}

\paragraph{Geometric locality and bounded degree.}

After a constant rescaling of the lattice, each auxiliary block
$A^{(t)}$ can be placed in a constant-radius neighborhood of the work
and clock qubits involved at time $t$. Since each block has constant
size, every application of $h^{(n)}$ therefore has constant geometric
diameter.

Every target term of $H^{\mathrm{stoq}_g}$ is replaced by four
applications of $h^{(n)}$, so the degree of each work and clock qubit
increases by only a constant factor.

For a fixed $t$, the block $A^{(t)}$ participates in four propagation
applications and, when $t\leq T-1$, four clock applications. It
participates in four output applications only when $t=T$.
Initialization contributes only $O(1)$ additional applications: if
$t_i=t$, then $W_i$ participates in its first nonidentity operation at
time $t$, and $U_t$ acts on at most three work qubits. Hence at most
three initialized qubits can satisfy $t_i=t$. Therefore every auxiliary
qubit also has bounded degree.

Consequently,
$\widetilde H^{\mathrm{stoq}_g}$ has bounded interaction range and
bounded degree in the sense of
\cref{def: geometrically local}.

The ordered $9$-tuples produced by the reduction are pairwise distinct,
and every tuple contains nine distinct physical qubits. For a fixed
target term, different values of $j$ use different ordered auxiliary
check pairs, while applications associated with different time steps use
different auxiliary blocks $A^{(t)}$. Distinct target terms at the same
time step differ in at least one work or clock position, and distinct
initialization terms differ in the position containing the initialized
work qubit $W_i$. Finally, the work, clock, and auxiliary registers are
disjoint, the work qubits involved in each circuit operation are
distinct, and the displayed auxiliary indices are pairwise distinct
modulo $8$. Hence the construction satisfies the ordered-hyperedge
requirements of \cref{def: DSIH}.

\section{Uniformity and System-Size Bookkeeping}
\label{app:uniformity}

\subsection{Decoding the DSIH System-Size Encoding}
\label{app:dsih-size-decoding}

\begin{proof}[Proof of \cref{lem:DSIH-size-decoding}]
Let
\[
n=N^4+\Lambda(N+1)+k+2
\]
and set
\[
r=\Lambda(N+1)+k.
\]
Since
\[
\Lambda\leq N(N-1)
\]
and $k\leq N$, we have
\[
r
\leq
N(N-1)(N+1)+N
=
N^3.
\]
On the other hand,
\[
(N+1)^4-N^4
=
4N^3+6N^2+4N+1
>
N^3.
\]
Therefore
\[
N^4
\leq
n-2
<
(N+1)^4,
\]
and hence
\[
N
=
\left\lfloor(n-2)^{1/4}\right\rfloor.
\]

Once $N$ is known, define
\[
r=n-2-N^4.
\]
Because $0\leq k\leq N$, Euclidean division by $N+1$ gives uniquely
\[
\Lambda
=
\left\lfloor\frac{r}{N+1}\right\rfloor,
\qquad
k
=
r\bmod(N+1).
\]
All of these operations can be performed in polynomial time.
\end{proof}

\subsection{Compatibility of the Unpinned Compiler with the Encoded System Size}
\label{app:unpinned-dsih-system-size}

\begin{proof}[Proof of \cref{lem:unpinned-dsih-system-size}]
The target Hamiltonian $H_{\mathrm{int}}$ contains
\[
M
=
|E|+N(N-1)
=
\frac{\Lambda}{2}+N(N-1)
\]
local interaction terms.

Since the interaction alphabet has size $m=2$, the compiler has selector
length
\[
\ell=1.
\]
Thus each target term receives $\ell+1=2$ dedicated auxiliary pairs, or four
ancilla qubits. The total number of non-isolated physical qubits is therefore
\[
N+4M
=
N+2\Lambda+4N(N-1).
\]

Using the encoded system size,
\begin{align*}
n-(N+4M)
&=
N^4+\Lambda(N+1)+k+2
-
\left(
N+2\Lambda+4N(N-1)
\right)\\
&=
N^4-4N^2+3N+k+2+\Lambda(N-1).
\end{align*}
For $N\geq2$,
\[
N^4-4N^2
=
N^2(N^2-4)
\geq0,
\]
and every remaining term in the preceding expression is nonnegative.
Hence
\[
n-(N+4M)\geq0.
\]

We therefore use the $N$ graph qubits and the $4M$ compiler ancillas as the
active register and fill the remaining
\[
n_{\mathrm{dummy}}
=
n-(N+4M)
\]
positions with isolated dummy qubits. Thus the final number of physical qubits is exactly $n$.
\end{proof}

\begingroup
\setlength{\emergencystretch}{1em}
\printbibliography
\endgroup

\end{document}